\documentclass[11pt,a4paper]{article}
\usepackage{amssymb,amsmath, amsfonts}

\usepackage{stix}
\usepackage{graphicx,graphics}
\usepackage{mathtools}
\usepackage[english]{babel}
\usepackage[utf8]{inputenc}
\usepackage{csquotes}
\usepackage{epsfig,url}
\usepackage{bbm}
\usepackage{theorem}
\usepackage{a4wide}
\usepackage{enumerate}
\usepackage{verbatim} 
\usepackage{color}
\usepackage{esint}

\usepackage[T1]{fontenc}

\newcommand{\ud}{\;\mathrm{d}}

\providecommand{\B}{\mathcal{B}}
\providecommand{\K}{\mathcal{K}}
\providecommand{\Pe}{\mathcal{P}}

\providecommand{\ol}{\overline}

\providecommand{\eps}{\varepsilon}
\providecommand{\vphi}{\varphi}

\providecommand{\supp}{\operatorname{supp}}
\providecommand{\La}{\mathcal{L}}

\newtheorem{theorem}{Theorem}[section]
\newtheorem{definition}[theorem]{Definition}
\newtheorem{lemma}[theorem]{Lemma}
\newtheorem{notation}[theorem]{Notation}

{\theorembodyfont{\upshape}
\newtheorem{remark}[theorem]{Remark}

}
\numberwithin{equation}{section}
\numberwithin{theorem}{section}
\newcommand{\qed}{\hfill$\Box$}
\newenvironment{proof}{\begin{trivlist}\item[]{\em Proof:}\/}{\qed\end{trivlist}}
\newenvironment{proofof}[1][Proof]{\noindent \textit{#1.} }{\ \qed}
 
\providecommand{\A}{\mathcal{A}}

\newcommand{\Reals}{{\mathbb R}}
\newcommand{\Complex}{{\mathbb C\hspace{0.05 ex}}}
\newcommand{\Naturals}{{\mathbb N}}

\newcommand{\cf}{{\mathbbm 1}}

\newcommand{\D}{\mathcal{D}}

\title{From a non-Markovian system to the Landau equation}

\author{ Juan J. L. Vel\'azquez
\thanks{\emailjuan}, Raphael Winter \thanks{\emailalessia}  \\[1em]
$\,^\ddag$\UBaddress}
\date{\today}
\newcommand{\email}[1]{E-mail: \tt #1}
\newcommand{\emailjuan}{\email{velazquez@iam.uni-bonn.de}}
\newcommand{\emailalessia}{\email{raphaelwinter@iam.uni-bonn.de}}
\newcommand{\UBaddress}{\em University of Bonn, Institute for Applied Mathematics\\
\em Endenicher Allee 60, D-53115 Bonn, Germany}

\date{\today}
	
\begin{document}
  
\maketitle
\begin{abstract}
	In this paper, we prove that in macroscopic times of order one, the solutions to the truncated BBGKY hierarchy
	(to second order) converge in the weak coupling limit to the solution of the nonlinear spatially homogeneous Landau equation. The truncated problem describes the formal leading order behavior of the underlying particle dynamics, and can be reformulated as a non-Markovian hyperbolic equation which converges to the Markovian evolution described by the parabolic Landau equation. The analysis in this paper is motivated
	by Bogolyubov's derivation of the kinetic equation by means of a multiple time scale analysis of the BBGKY hierarchy.
\end{abstract}

\tableofcontents
\newpage

\section{Introduction} \label{Sec:Introduction}

A central objective in kinetic theory is the derivation of effective equations for
macroscopic densities of particles in a plasma or gas. 
Two of the main equations in this context are the Boltzmann equation
and the Landau equation, and a large portion of the mathematical
research in this area is devoted to the study of these equations.
For an extensive overview of mathematical kinetic theory we refer to \cite{spohn_kinetic_1980,villani_review_2002}.  
For the Boltzmann equation, rigorous results have been proved, both 
on the level of the equation itself, and on the level of its
derivation from particle systems. Results on 
well-posedness, entropic properties
of solutions, and rate of convergence to equilibrium can be found in 
\cite{desvillettes_trend_2005,diperna_cauchy_1989,toscani_sharp_1999,wennberg_stability_1993}.
For the derivation of the equation from interacting particle systems we refer to \cite{gallagher_newton_2013,lanford_time_1975,pulvirenti_validity_2014,pulvirenti_boltzmann-grad_2017}, and to  \cite{boldrighini_boltzmann_1983,desvillettes_linear_1999,gallavotti_grad-boltzmann_1999,spohn_lorentz_1978} for the derivation of the linear equation from Lorentz models.  

Many of these problems, 
including the derivation starting from particle systems, are still open for the Landau equation. 
The goal is to describe the evolution of the macroscopic velocity distribution of (initially randomly) distributed particles $(X_i,V_i)_{i \in I} \in (\Reals^3 \times \Reals^3)^I$ (where $I$ is a countable or finite index set)
evolving according to the Hamiltonian dynamics:
\begin{equation} \label{Hamiltonian}
\begin{aligned}
\partial_\tau X_i(\tau) 	&= V_i(\tau)  \\
\partial_\tau V_i(\tau)	&= - \theta^2 \sum_{j\neq i} \nabla \phi(X_i(\tau)-X_j(\tau)), \quad  \text{$\theta>0$ scaling parameter}.
\end{aligned}
\end{equation}
Here $\phi=\phi(x)$ is the interaction potential, and in the rest of the paper we use the notation
$\nabla \phi= \nabla_x \phi$ and assume $\phi$ is radially symmetric.
When the strength of the potential is small, i.e. $\theta^2\rightarrow 0$, and for large times $t\gg1$, the
evolution of the particles is governed by many small deflections.
 Let $Z>0$ be the average number
of particles per unit of volume, to be made precise later.
It is widely accepted that for a suitable choice of $\phi$ and
rescaling of $\theta\rightarrow 0$ and $Z$, the number density $f(t,v)$ of a spatially homogeneous system satisfies the Landau equation (cf. \cite{spohn_kinetic_1980}):
\begin{equation} \label{Landauintr}
\begin{aligned}
\partial_t f(t,v) 	&= \sum_{i,j=1}^3 \partial_{v_i}  \left(\int_{\Reals^3} a_{i,j}(v-v')
(\partial_{v_j} - \partial_{v'_j})\big( f(t,v) f(t,v') \big)\ud{v'}  \right) \\
f(0,v)	&= f_0(v).
\end{aligned}
\end{equation} 
Here $t$ is a macroscopic time scale that we will specify later, and
the matrix valued function $a$ is determined by the pair interaction potential $\phi$:
\begin{align} \label{adefinition} 
a_{i,j}(w) &= \frac{\pi^2}{4}\int_{\Reals^3} k_ik_j \delta(k\cdot w) |\hat{\phi}(k)|^2 \ud{k} 	
= \frac{\Lambda}{|w|} \left(\delta_{i,j}-\frac{w_i w_j}{|w|^2}\right) \quad \text{for some $\Lambda>0$}.
\end{align}
In the  most physically relevant case -- that of Coulomb interaction, i.e. $\phi(x)=\frac{c}{|x|}$ -- considered in \cite{landau_kinetische_1936}, the constant $\Lambda$ is logarithmically divergent.

The equation \eqref{Landauintr}-\eqref{adefinition} was introduced by Landau in \cite{landau_kinetische_1936} (see also \cite{lifshitz_course_1981}).  However, Landau did not take
as a starting point the dynamics of the particles (cf. \eqref{Hamiltonian}).
Instead he studied the Boltzmann equation in the limit of grazing collisions, which was assumed to be a good approximation for the dynamics of the system \eqref{Hamiltonian}. A rigorous version of Landau's argument
can be found in \cite{alexandre_landau_2004}.

A rather general approach to deriving kinetic equations from \eqref{Hamiltonian} was later developed
by Bogolyubov (cf. \cite{bogoliubov_problems_1962}). We will briefly summarize this method here. Consider a countable system of particles $(X_i(0),V_i(0))_{i \in I} \in (\Reals^3 \times \Reals^3)^I$, distributed according to an uncorrelated, translation invariant grand canonical ensemble.
Furthermore, assume the velocities $V_i$ are of order one.
We consider scaling limits of a single scaling parameter $\eps\rightarrow 0$,
as is customary in the modern literature on kinetic equations (cf. \cite{balescu_equilibrium_1975,gallagher_newton_2013,lanford_time_1975,pulvirenti_validity_2014,spohn_kinetic_1980}). We set the strength of the potential $\theta^2$ and the particle density $Z$ as:
\begin{align} \label{eq:classicalscale}
\theta^2 = \eps^{\beta}, \quad Z=\eps^{1-2\beta}. 
\end{align}
For reasons we will explain later, we choose $\beta \in (0,1)$.
We can then consider the $n$-particle correlation functions $F_n(x_1,v_1,\ldots, x_n,v_n)$. In order to work with functions
of order one, we define the rescaled functions $f_n$ by:
\begin{align*}
F_n(\tau,x_1,v_1,\ldots,x_n,v_n)=Z^n f_n(\tau,x_1,v_1,\ldots,x_n,v_n).	
\end{align*}   
Then the correlation functions $f_n$ satisfy the so-called BBGKY hierarchy (see e.g.~\cite{balescu_equilibrium_1975}):
\begin{equation} \label{BBGKY}
\begin{aligned}
\partial_\tau f_n + \sum_{i=1}^n v_i \nabla_{x_i} f_n &- \eps^{1-\beta}
\sum_{i=1}^n \int    \nabla \phi(x_i-x_{n+1}) \nabla_{v_i}
f_{n+1} \ud{x_{n+1}} \ud{v_{n+1}}  \\
&= \eps^\beta \sum_{i \neq j} \nabla \phi(x_i-x_j) \nabla_{v_i} f_n .	
\end{aligned}
\end{equation}
Since $\beta\in (0,1)$, we have $Z\theta^2 = \eps^{1-\beta}\rightarrow 0$. The physical meaning
of this will be explained below. Under this assumption, Bogolyubov's argument
yields the Landau equation \eqref{Landauintr} as the limiting equation for $f_1$. In the case $\beta=1$, i.e. $Z\theta^2=1$, Bogolyubov's technique
can also be applied, however here the limiting equation is the Balescu-Lenard equation (see \cite{balescu_statistical_1963,balescu_equilibrium_1975,  lenard_bogoliubovs_1960}).
In this case, the particles of the system must be viewed as interacting as part of an effective medium, in which the interaction of pairs of particles is modified due to collective effects. In the physics literature this is
characterized by means of the so-called dielectric function, that gives a nontrivial correction
to the limit kinetic equation. We will not however consider this issue in the present paper. 

Our assumption $Z \theta^2 \rightarrow 0$ has a clear interpretation in terms of dimensionless quantities. 
Observe that $Z\theta^2$ describes the ratio of the average potential to the average kinetic energy of a particle:
\begin{equation*}
\frac{\langle \theta^2 \sum_{j \in I: i\neq j} \phi(X_i-X_j)\rangle}{\langle V_i^2 \rangle} \sim Z \theta^2= \eps^{1-\beta}.
\end{equation*}
Since $Z \theta^2 \rightarrow 0$, the kinetic energy of the particles particles is much larger than their potential energy, hence the absence of collective effects.
Our objective is to study the evolution of the one particle function $f_1$. We will refer to the timescale on which this evolution takes place
as macroscopic time. To simplify notation, we set $(x_i,v_i)=\xi_i$ and introduce the (rescaled) truncated  correlation functions $g_2$, $g_3, \ldots$ defined by:
\begin{align*}
g_2(\xi_1,\xi_2) 		&= f_2(\xi_1,\xi_2) - f_1(\xi_1) f_1(\xi_2)   \\
g_3(\xi_1,\xi_2,\xi_3) 	&= f_3(\xi_1,\xi_2,\xi_3) -  f_1(\xi_1) g_2(\xi_2,\xi_3) - f_1(\xi_2) g_2(\xi_1,\xi_3)  - f_1(\xi_3) g_2(\xi_1,\xi_2) \\
&- f_1(\xi_1) f_1(\xi_2)f_1(\xi_3)  \\
g_4(\xi_1,\xi_2,\xi_3,\xi_4) &= \ldots.	 
\end{align*} 
From \eqref{BBGKY} we can derive equations for $g_2$, $g_3$ and higher order truncated correlation functions. A crucial observation is that 
we can expect to have a separation of orders of magnitude $f_1 \gg g_2 \gg g_3$ as $\theta^2=\eps^\beta \rightarrow 0$. 
To see this, we consider now the exact equations satisfied by $g_2$ and $f_1$.
For ease of notation, we introduce the function $\sigma$ with $\sigma(1)=2$, $\sigma(2)=1$, to relabel
the indexes of $\xi_1$, $\xi_2$. 
By a straightforward algebraic computation, the BBGKY hierarchy \eqref{BBGKY} implies:
\begin{equation} \label{BBGKYwithint}
\begin{aligned}
\partial_\tau f_1 &= \eps^{1-\beta} \nabla_v \cdot \left(\int \nabla \phi(x_1-x_3) g_2(\xi_1,\xi_3) \ud{\xi_3}  \right) \\
\partial _\tau g_2+&\sum_{i=1}^2 v_i \nabla_{x_i} g_2 -\eps^{1-\beta}\sum_{i=1}^{2}\int \nabla \phi (x_i-x_3)\nabla_{v_{i}}(f_{1}(\xi_i)g_2(\xi _{\sigma \left( i\right) },\xi_{3})+g_3(\xi_1,\xi_2,\xi_3)) \ud{\xi_3} \\
&= \eps^{\beta}\sum_{i=1}^{2}\nabla _{v_{i}}\left(f_{1}(\xi_1)f_{1}(\xi_2)+g_{2}(\xi_1,\xi_2)\right)\nabla \phi (x_i-x_{\sigma(i)}).
\end{aligned}
\end{equation} 
Indeed, the sources on the right-hand side of the equation are of order $\eps^\beta \ll 1$,
leading us to expect $f_1 \gg g_2$. A similar argument suggests $g_2 \gg g_3$. 
Therefore, we approximate \eqref{BBGKYwithint} by:
\begin{align} \label{BBGKYwithint2f}
\partial_\tau f_1 &= \eps^{1-\beta} \nabla_v \cdot \left(\int \nabla \phi(x_1-x_3) g_2(\xi_1,\xi_3) \ud{\xi_3} \right) \\
\partial _\tau g_2+&\sum_{i=1}^2 v_i \nabla_{x_i} g_2 -\eps^{1-\beta}\sum_{i=1}^{2}\int \nabla \phi (x_i-x_3)\nabla_{v_{i}} f_{1}(\xi_i)g_2(\xi _{\sigma \left( i\right) },\xi_{3}) \ud{\xi_3} \label{BBGKYwithint2g} \\
&=\eps^\beta \sum_{i=1}^{2}\nabla _{v_{i}}\left(f_{1}(\xi_1)f_{1}(\xi_2)\right)\nabla \phi (x_i-x_{\sigma(i)}) \notag.
\end{align}
Since the source term on the right-hand side of \eqref{BBGKYwithint2g} is of order $\eps^\beta$, 
it is convenient to define the function $\tilde{g}_2= \eps^{-\beta}g_2$. Then we can rewrite \eqref{BBGKYwithint2f}-\eqref{BBGKYwithint2g} as:
\begin{align}
\partial_\tau f_1 &= \eps \nabla_v \cdot \left(\int \nabla \phi(x_1-x_3) \tilde{g}_2(\xi_1,\xi_3) \ud{\xi_3} \right)\label{BBGKYwithint3f} \\
\partial _\tau \tilde{g}_2+&\sum_{i=1}^2 v_i \nabla_{x_i} \tilde{g}_2-\eps^{1-\beta} \sum_{i=1}^{2}\int \nabla \phi (x_i-x_3)\nabla_{v_{i}} f_{1}(\xi_i)\tilde{g}_2(\xi _{\sigma \left( i\right) },\xi_{3}) \ud{\xi_3} \label{BBGKYwithint3g} \\
&=  \sum_{i=1}^{2}\nabla _{v_{i}}\left(f_{1}(\xi_1)f_{1}(\xi_2)\right)\nabla \phi (x_i-x_{\sigma(i)})\notag.
\end{align} 
It is now apparent that the contribution of the integral term in \eqref{BBGKYwithint3g}
is negligible in the approximation used in this paper and therefore that term can be dropped. Moreover, the stabilization of $\tilde{g}_2$ to a steady state takes place in  times  $\tau$ of order one. On the other hand, the changes in $f_1$ take place in times $\tau$ of order $1/\eps$, suggesting we should define the macroscopic time scale as $t=\eps \tau$. The separation of time scales is a key point in the argument by Bogolyubov. It implies that, on the macroscopic timescale, the truncated correlation 
$\tilde{g}_2(t)$ can be expected to be a functional $\tilde{g}_2(t)=  A_2[\eps,f_1(t)]$ of $f_1$. More generally, 
Bogolyubov argues that on the timescale $t$ all truncated correlation functions $g_k$ evolve in a similar adiabatic manner.
This ansatz allows us to derive the limiting kinetic equation for $f_1(t)$ in a straightforward fashion. The integral term in \eqref{BBGKYwithint3g}
can be neglected, since it is of lower order. Therefore \eqref{BBGKYwithint3f}-\eqref{BBGKYwithint3g} can be approximated by ($\nabla \phi(x)=-\nabla \phi(-x)$ by radial symmetry):
\begin{equation} \label{g2eq}
\begin{aligned}
\partial_\tau f_1 &= \eps \nabla_v \cdot \left(\int \nabla \phi(x_1-x_3) \tilde{g}_2(\xi_1,\xi_3) \ud{\xi_3} \right) \\
\partial _\tau \tilde{g}_2+\sum_{i=1}^2 v_i \nabla_{x_i} \tilde{g}_2 	&= (\nabla_{v_1} - \nabla_{v_2})\left(f_{1}(\xi_1)f_{1}(\xi_2)\right)\nabla \phi (x_1-x_2).
\end{aligned}
\end{equation}
Now the functional $A_2[f_1]$ can be computed explicitly by solving the steady state equation for $\tilde{g}_2$ in \eqref{g2eq}.
We substitute $\tilde{g}_2=A_2[f_1]$ in the equation for $f_1$ and identify the Landau equation \eqref{Landauintr} as the limiting equation on the macroscopic time scale $t$.
For the scaling limit with $\beta=1$ in \eqref{eq:classicalscale}, the functional $A_2[f_1]$ was computed explicitly in \cite{lenard_bogoliubovs_1960}, solving the steady state equation
associated to \eqref{BBGKYwithint3g}. The resulting limit equation for $f_1(t)$ is the Balescu-Lenard equation, which will not be considered
in this paper.

It is possible to go from  \eqref{g2eq} to the Landau equation, reformulating the problem as a non-Markovian evolution. To this end, we rewrite \eqref{g2eq} as a single equation, involving only terms depending on $f_1$. We can integrate the equation for $\tilde{g}_2$ along characteristics (by assumption the initial correlations vanish):
\begin{align*}
\tilde{g}_2(\tau,\xi_1,\xi_2) &= \int_0^\tau  (\nabla_{v_1} - \nabla_{v_2})(f_1(s,\xi_1) f_1(s,\xi_2)) \nabla \phi(x_1-x_2-(\tau-s)(v_1-v_2)) \ud{s}. 
\end{align*} 
We obtain a closed equation for the function $f_1$ by plugging
this formula back into \eqref{g2eq}. The function $f_1$ changes on the macroscopic timescale $ t= \eps \tau $. In order to keep the velocities $v$ of order one,
we must change the spatial variable, using as the unit of length the mean free path, i.e. the flight length after which the velocity of a particle deviates by an amount of order one. We therefore define the macroscopic length
scale $y=\eps x$. Notice that due to the translation invariance of the system, $f_1(t,y,v)=f_1(t,v)$ is independent of the spatial variable. Let $f_\eps(t,v)$ be the particle density function on the
macroscopic timescale, then $f_\eps$ satisfies the equation
\begin{equation} \label{Mainequation}
\begin{aligned} 
\partial_t f_\eps 	&= 	\frac1\eps \nabla_v \cdot \left(\int_0^t K[f_\eps(s)]\big(\frac{t-s}\eps,v\big) \nabla f_\eps(s,v)  
-  \nabla_v \cdot K[f_\eps(s)]\big(\frac{t-s}\eps,v\big) f_\eps(s,v)  \ud{s}\right)  \\ 
f_\eps(0,v)	&= f_0(v) ,
\end{aligned}
\end{equation}
where $K$ is given by the formula
\begin{align*}
K[f](\tau,v)			&:=	 \int \int \nabla \phi(x) \otimes \nabla \phi(x-\tau(v-v'))	f(v')  \ud{v'}\ud{x}.	
\end{align*}
By Bogolyubov's argument, \eqref{Mainequation} should give the leading order behavior of
the one particle function $f_\eps$, which should converge to a solution $f$ of the Landau equation \eqref{Landauintr}. Note that the equation
\eqref{Mainequation} yields a nonlinear non-Markovian evolution for $f_\eps$,
while $f$ is given by a Markovian, parabolic equation. The convergence of solutions $f_\eps$ of an equation with memory effects to a kinetic equation
is a characteristic feature of kinetic particle limits, as indicated in
\cite{balescu_equilibrium_1975,bobylev_particle_2013,spohn_kinetic_1980}.

Notice that in the class of scaling limits \eqref{eq:classicalscale}, for $\beta=1/2$ we obtain the classical weak coupling limit (cf. \cite{bobylev_particle_2013,spohn_kinetic_1980}). In this case, the (microscopic) density $Z$ remains of order one. Therefore, the interaction potential takes
the form $\phi_\eps(y)= \sqrt{\eps} \phi(y/\eps)$ in macroscopic variables, which has a range of order $\eps$. The number of collisions per macroscopic unit of time is $1/\eps$, and the transferred momentum produced by each collision is of order $\sqrt{\eps}$. Assuming that the collisions are independent, this makes the variance of the deflections on the macroscopic time scale of order one, due to the central limit theorem. We remark that the scaling 
\eqref{eq:classicalscale} is more general than the classical weak coupling, since  $Z\rightarrow 0$ or $Z\rightarrow \infty$ are possible, depending on the choice of $\beta \in (0,1)$. In these cases, the diffusion in the velocity variable also follows from an analogue of the central limit theorem. For instance if $Z\rightarrow \infty$, a particle interacts with $Z$ particles during a macroscopic time of order $\eps$, which yields a deflection of $\sqrt{Z \theta^4}=\sqrt{\eps}$. Since the range of the potential is of order $\eps$, these deflections become independent after macroscopic times of order $\eps$ and therefore the deflection of a particle in a macroscopic unit of time is of order one. For $Z\rightarrow 0$, the macroscopic time between collisions is $\eps/Z=\eps^{2\beta}$, and the deflection in each collision is $\eps^{\beta}$. Therefore, another central limit theorem argument gives the diffusive behavior in the velocity variable. 

There are multiple gaps to bridge in order to make Bogolyubov's argument rigorous. First one has to prove the well-posedness
of the infinite system of ODEs \eqref{Hamiltonian}. Sufficient conditions on the potential and initial data for this can be found, for example in \cite{spohn_large_2012}. Proving the separation of orders of magnitude $f_1 \gg g_2 \gg \ldots$
and the validity of the truncation of the BBGKY hierarchy is a key problem, and still open. We will see later that this
assumption cannot be expected to hold in general, at least when the relative
velocity of particles becomes very small.

 Actually, this fact
is closely related to the onset of the singularity $|v_1-v_2|^{-1}$
in the Landau equation (cf. the term $\Lambda/|w|$ in \eqref{adefinition}).
The easiest way to understand this singularity is through a careful analysis of the
mutual deflection of two particles with very close velocities, i.e. $v_1-v_2 \approx 0$. An implicit assumption made in the derivation of the Landau equation
is that the particles move along near-rectilinear trajectories. Two particles
moving along near-rectilinear trajectories with velocities $v_1$, $v_2$ which
come sufficiently close to interact, will interact during a collision time
of order $|v_1-v_2|^{-1}$. Hence, the resulting deflection is of
order $\theta^2 |v_1-v_2|^{-1}$. If $|v_1-v_2|\ll \theta^2$, this quantity is not small, and this contradicts the assumption of near-rectilinear motion.
Therefore the underlying assumption behind the derivation of the Landau equation breaks down for particles with very small relative velocity.
Nevertheless, if the velocities satisfy the condition $1\gg |v_1-v_2| \gg\theta^2$, the rectilinear approximation is valid, in spite
of the fact that the collision time diverges like $|v_1-v_2|^{-1}$.
This is the reason for the onset of the factor $1/|w|$ in \eqref{adefinition}.

We remark that the introduction of this singularity does not pose a serious
physical difficulty concerning the validity of the Landau equation, since
it is an integrable singularity. This is due to the fact that
the number of pairs of particles with small relative velocities is a sufficiently small fraction of the total number of pairs of interacting particles, and therefore can be neglected. In particular, the fraction of interacting particles
with $|v_1-v_2| \ll \theta^2$ which experience relevant deflections in their collisions vanishes in the limit $\theta \rightarrow 0$.

We emphasize that the singularity $1/|w|$ appearing in the diffusion matrix
in the Landau equation (cf. \eqref{adefinition}) is a consequence of the
collision dynamics of particles with small relative velocity, and therefore  independent of the particular choice of the interaction potential $\phi$.
In particular this singularity is not specifically related to the choice of the
Coulomb interaction between the particles. 
It is interesting to point out the difference with the Boltzmann equation, where
the homogeneity of the collision kernel is closely related to the homogeneity of the interaction potential (cf. \cite{villani_review_2002}).

We notice that the assumption $f_1\gg g_2$ can be expected to fail in the region
of very small relative velocities due to the same geometric considerations as above (cf. \cite{bobylev_particle_2013}). Indeed, the function $g_2(x_1,v_1,x_2,v_2)$ measures the deflections of interacting particles with velocities $v_1$, $v_2$.
For small relative velocities, the truncated correlation function $g_2$ can be
of the same order as $f_1$. It is worth remarking that dropping the term
$g_2$ on the right-hand side of \eqref{BBGKY} is equivalent to approximating
the trajectories of interacting particles by straight lines. As seen before, this
fails in the region $|v_1-v_2|\ll \theta^2$, which is vanishing in the limit $\theta\rightarrow 0$. Notice that this observation yields some insight into the type of functional spaces in which the approximation $f_1 \gg g_2$ can be expected to hold.

In this paper, we will prove that Bogolyubov's adiabatic approach to deriving the Landau equation \eqref{Landauintr} from 
the system \eqref{Mainequation} is indeed correct, when the singularity $v \approx v'$ is cut out.
To be precise, we consider the Landau-type equation
\begin{equation} \label{cutLandauintr}
\begin{aligned}
\partial_t f 	&= \sum_{i,j=1}^3 \partial_{v_i} \left(\int_{\Reals^3} a_{i,j}(v-v')
(\partial_{v_j} - \partial_{v'_j})\big( f(t,v) f(t,v') \big) \eta(|v-v'|^2)\ud{v'}  \right)  \\
f(0,v)	&= f_0(v),
\end{aligned}
\end{equation} 
where $\eta(r)$ vanishes for $r$ small. We will derive the equation \eqref{cutLandauintr} from
the system \eqref{Mainequation}, where $K$ is now given by:
\begin{align} \label{Kcutoff}
K[f](t,v)			&:= 	 \int \int \nabla \phi(x) \otimes \nabla \phi(x-t(v-v'))	f(v') \eta(|v-v'|^2)  \ud{v'}\ud{x}.	
\end{align}
The reason for introducing the artificial cutoff $\eta(r)$ in the region
of small relative velocity is that the estimates in this paper are
presently not strong enough to deal with the case $\eta \equiv 1$.
As indicated above, the effect of collisions with small relative velocities
can be expected to be small, and therefore the 
Landau equation (cf. \eqref{Landauintr}, \eqref{adefinition}) and the modified Landau equation (cf. \eqref{cutLandauintr}, \eqref{Kcutoff}) might be expected to
exhibit similar physical properties. In particular, the asymptotic behavior 
of the matrix $K[f](t,v)$ as $v\rightarrow \infty$ is preserved.

The main results of the paper are the existence of strong solutions $f_\eps$ to \eqref{Mainequation}
with $K$ as in \eqref{Kcutoff}, and the convergence 
of these solutions to a strong solution $f$ of the Landau equation \eqref{cutLandauintr} for macroscopic times
of order one. 
We assume that $f_0$ is close to the Maxwellian steady state of the limit equation and choose a particular short range 
potential $\phi$. In contrast to the diffusive, parabolic Landau equation, equation \eqref{Mainequation} is  hyperbolic. 
We show that regularity and decay of the initial datum $f_0$ are conserved. 
Furthermore, the evolution given by \eqref{Mainequation} is clearly non-Markovian, since the time derivative depends on the whole history of the function $f_\eps$ until time $t$. In the limit $\eps \rightarrow 0$, this memory effect disappears and we recover the Markovian dynamics of the Landau equation. 

As mentioned above, the derivation of the Landau-type equations 
from particle systems is still largely open. The linear Landau equation has been derived in \cite{basile_diffusion_2014,desvillettes_rigorous_2001} as a scaling limit
of systems with a single particle traveling through a random (but fixed) configuration of scatterers.

Furthermore, it is shown in \cite{bobylev_particle_2013} that the Landau equation \eqref{Landauintr} 
is consistent with a scaling limit of interacting particle systems. More precisely it is shown
that the time derivative of the macroscopic density of particles in the weak coupling limit at $t=0$ is correctly predicted by the
Landau equation. The technique follows a similar line of reasoning to that of Bogolyubov, truncating the BBGKY hierarchy to a system like
\eqref{Mainequation}, and proving convergence to the Landau equation on a timescale shorter than the macroscopic. It is worth noticing, that in
\cite{bobylev_particle_2013} the convergence of solutions of the truncated hierarchies to the solution of the the Landau equation is established in the sense of weak convergence. In this paper, the convergence of the solutions
$u_\eps$ of the non-Markovian problem \eqref{Mainequation} to the solution $u$ of the Landau-type equation \eqref{cutLandauintr} is proved in strong norms, up to macroscopic times of order one. 
Given that estimates in stronger norms, which allow for strong convergence, are technical to obtain, it is natural to ask why this is needed. The reason for this is that our technique for controlling the nonlinearity in
\eqref{Mainequation} up to macroscopic times of order one is based on a 
linearization of the problem in strong norms, combined with estimates of quadratic or higher order terms. This is only possible in very strong norms that in particular yield estimates for the time derivative of the solution.
It is certainly possible to prove the convergence $u_\eps \rightarrow u$ in the weak topology. However, since stronger estimates were needed to prove well-posedness of the non-Markovian problem up to macroscopic times, the convergence is readily established in stronger norms.

On the other hand, convergence of the solutions of the non-Markovian evolution to the solution of the Landau equation in weak topology, as used in \cite{bobylev_particle_2013}, would be in some sense the natural result, considering that the solutions of the non-Markovian equation exhibit significant changes on the microscopic time scale. Indeed, one important assumption made in this paper is that the initial data for \eqref{Mainequation}, i.e. the initial distribution of particles, is close to a Maxwellian equilibrium.
This smallness condition is needed in order to control the effect of these
oscillations on the macroscopic evolution of the one particle function.

In \cite{guo_landau_2002}, global well-posedness of the spatially inhomogeneous Landau equation was
proved for initial data close to equilibrium in a periodic box. Lower bounds on the entropy dissipation in the Landau equation can be found in \cite{desvillettes_entropy_2015}.
 A concept of weak solutions for the homogeneous Landau equation \eqref{Landauintr}, namely $H$-solutions, was introduced in \cite{villani_new_1998}. This paper also gives  sufficient conditions under which the Landau equation can be obtained as a grazing collision limit, taking as a starting point the Boltzmann equation. In the grazing collision limit, the collision kernel in the Boltzmann equation is concentrated on the set of collisions with small transferred momentum. The Landau equation has also been derived from the Boltzmann equation in the grazing collision limit in the spatially inhomogeneous case (cf. \cite{alexandre_landau_2004}).

Given that the paper \cite{guo_landau_2002} proves global well-posedness for the Landau equation near the Gaussian distribution in the spatially inhomogeneous case, it is natural to ask why such that a result cannot be obtained for the non-Markovian equation \eqref{Mainequation}. To explain this we describe the analogies and differences between the approach in 
\cite{guo_landau_2002} and that of this paper.

The approach of \cite{guo_landau_2002} is based in a linearization near the Maxwellian distribution of velocities. A dissipation formula allows one to obtain global estimates for the difference between the solutions of the
inhomogeneous Landau equation and the Maxwellian, that can be used to
prove global stability results. In this paper, we 
consider the equation \eqref{Mainequation}, which
unlike the Landau equation is non-Markovian and, due to this, not
pointwise dissipative in time. The techniques used in this paper are more reminiscent of the theory of symmetric hyperbolic systems (\cite{john_partial_1991}, \cite{majda_compressible_1984}), which
usually only yields local well-posedness in time, due to the fact that quadratic or higher order terms must be estimated.
We generalize these methods to the case of a non-Markovian evolution with
memory effects. The key ingredient in our approach is the derivation of a coercivity estimate averaged in time (cf. Lemma \ref{coercivitylemma}) for the solutions obtained 
with $f_{\varepsilon}=f_0$ frozen inside the operator $K\left[f_{\varepsilon}\left(  s\right)  \right]  $ on the right-hand side of \eqref{Mainequation}.
Our proof strategy for  Lemma \ref{coercivitylemma}  in this paper does
not rule out solutions of the linearized problem which separate exponentially from the initial distribution function. The estimates in Lemma \ref{coercivitylemma} are
based on a Laplace transform argument and the derivation of estimates for some elliptic equations with complex coefficients, where the Laplace 
transform argument $z$ remains at a positive
distance from the imaginary axis. Obtaining global-in-time estimates for solutions of \eqref{Mainequation} would require us to prove that coercivity still holds for the operator linearized around the Maxwellian, even when the complex parameter $z$ approaches 
to the imaginary axis. Such an estimate might be true, but seems to
require more involved arguments than the ones used in this paper. Notice that
a coercivity estimate strong enough to provide decay of the perturbations with respect to the Maxwellian for long times might be easier to obtain in a compact domain (for instance a torus) than in the whole space.

Due to the mathematical difficulties arising from the singularity $|v_1-v_2|^{-1}$ for relative velocities in the Landau
equation \eqref{Landauintr}-\eqref{adefinition}, a number of
Landau-type equations, in which the singularity has been weakened, have been studied. As for our modification of the Landau equation \eqref{cutLandauintr},
these equations cannot be directly derived as a scaling limit of interacting
particle systems \eqref{Hamiltonian}.
These modified Landau equations are obtained by replacing 
the singularity $|v-v'|^{-1}$ by $|v-v'|^{\gamma+2}$.
The well-posedness of these equations, as well as stability of Maxwellians
and the dissipation of entropy have been studied in
\cite{desvillettes_spatially_2000,desvillettes_spatially_2000-1,silvestre_upper_2017,strain_exponential_2008,villani_spatially_1998}.

The present paper is structured as follows: In Section \ref{Sec:Preliminary}, we give a precise formulation of the main results Theorem \ref{mainthm1pf} and \ref{mainthm2pf}, as well as the proofs of some auxiliary results.
In Section~\ref{Sec:Linear} we prove the result in the linear case. 
Section~\ref{Sec:Freezing} proves that the a priori estimates are stable under certain small perturbations, and that these smallness assumptions are conserved by the equation. In Section~\ref{sec:Existence} we give the proofs of the two main theorems.  

\section{Main results, notation and auxiliary lemmas} \label{Sec:Preliminary}

\subsection{Formulation of the main results}
Our goal is to prove the existence of a strong solution to the equation
\begin{equation} \label{cuteq}
\begin{aligned} 
\partial_t u_\eps 	= 	&\frac1\eps \nabla_v \cdot \left(\int_0^t K[u_\eps(s)]\left(\frac{t-s}\eps,v\right) \nabla u_\eps(s,v) \ud{s}\right) \\
-  	&\frac1\eps \nabla_v \cdot \left(\int_0^t P[u_\eps(s)]\left(\frac{t-s}\eps,v\right) u_\eps(s,v)  \ud{s}\right)  \\ 
u_\eps(0,v)	&= u_0(v) ,
\end{aligned}
\end{equation}
where $K$ and $P$ denote the following operators:
\begin{equation} \label{KPdef}	
\begin{aligned}
K,P : W^{1,1}(\Reals^3) 	 &\longrightarrow L^\infty(\Reals^+ \times \Reals^3) \\
K[u](t,v)			&:= 	 \int \int \nabla \phi(x) \otimes \nabla \phi(x-t(v-v'))	u(v') \eta(|v-v'|^2) \ud{v'}\ud{x}	\\
P[u](t,v)			:=	\nabla_v \cdot K(t,v) &= \int \int \nabla \phi(x) \otimes  \nabla\phi(x-t(v-v'))
\nabla u(v')  \eta(|v-v'|^2)	\ud{v'}	\ud{x}.
\end{aligned}
\end{equation}
We will specify the potential $\phi$ and the cutoff function $\eta \in C^\infty(\Reals)$ below.
Formally, as $\eps \rightarrow 0$, the functions $u_\eps$ converge to a strong solution $u$ of:
\begin{equation} \label{limiteq}
\begin{aligned}
\partial_t 	u 		&= \nabla \cdot \left( \K[u] \nabla u\right) - \nabla \cdot \left( \Pe[u]  u\right) \\
u(0,v)	&= u_0(v) 	\\
\K[u](v)	&=  \frac{\pi^2}{4}\int (k \otimes k) |\hat{\phi}(k)|^2 \delta(k\cdot(v-v')) \eta(|v-v'|^2)u(v') \ud{k} \ud{v'} \\
\Pe[u](v)	&=  \frac{\pi^2}{4} \int (k \otimes k) |\hat{\phi}(k)|^2 \delta(k\cdot(v-v')) \eta(|v-v'|^2) \nabla u(v') \ud{k} \ud{v'}. 			
\end{aligned}
\end{equation}
We will prove this result for $u_0$ close to the Maxwellian distribution $m$, which is the steady state of the limit equation \eqref{limiteq}. 
Furthermore we choose the potential $\phi$ to have a particular form, making the computations considerably easier.
\begin{notation} \label{defpotential}
	Let $\eta \in C^\infty(\Reals)$ be a fixed cutoff function with $0\leq \eta\leq 1$, $\eta(r)=1$ for $|r|\geq \kappa$
	and $\eta(r)=0$ for $|r|\leq \frac{\kappa}2$ for some $\frac12 >\kappa>0$ that we will not further specify in the following
	analysis. We choose the potential $\phi(x)$ to be given by
	\begin{align}
	\phi(x) = \sqrt{\frac{2}{\pi}} K_0(|x|),
	\end{align}
	where $K_0$ is the modified Bessel function of second type.
\end{notation}
\begin{remark}
	The potential $\phi$ is monotone decreasing, decays exponentially at infinity and diverges logarithmically at the origin. Our approach 
	also seems to work  for other potentials with analogous properties, but becomes significantly less technical with this particular choice.
	The Fourier transform of the potential is given by:
	\begin{align} \label{FTpotential}
	\hat{\phi}(k) &= \frac1{(1+|k|^2)^\frac{3}{2}}.
	\end{align}
\end{remark}

The function spaces we are going to work with in the forthcoming analysis are the following ones.

\begin{definition}
	Let $\lambda(v), \tilde{\lambda}(v)$ be the weight functions given by $\lambda(v):=e^{|v|}$, $\tilde{\lambda}(v):=\frac{e^{|v|}}{1+|v|}$.
	For $n \in \Naturals$ and $\nu=\lambda,\tilde{\lambda}$, we define the weighted Sobolev space $H^n_\nu$ as the closure of $C^\infty_c\big(\Reals^3\big)$
	with respect to the norm:
	\begin{align}
	\|u\|^2_{H^{n}_{\nu}} 		&:= \sum_{\alpha \in \Naturals^3, |\alpha|\leq n} \|\nu^\frac12 (\cdot) \nabla^\alpha u(\cdot)\|^2_{L^2}. 
	\end{align}
	In the case $n=0$ we also write $H^n_\nu=L^2_\nu$.
	For functions $f(t,v)$ with an additional time dependence, we define the spaces $V^{n}_{A,\nu}$ as
	the closure of $C^\infty_c \big([0,\infty) \times \Reals^3;\Reals^d\big)$ with respect to:
	\begin{align} \label{Vndef}
	\|f\|^2_{V^{n}_{A,\nu}} 	&:= \int_0^\infty e^{-At} \sum_{j=1}^d \|f_j(t,\cdot)\|^2_{H^n_\nu}  \ud{t}, \quad \text{where $A\geq 1$.} 
	\end{align}	
	Let $X^n_{A,\nu}$ be the function space given by:
	\begin{equation} \label{Xdef}
	\begin{aligned}
	X^n_{A,\nu} 		:= \{(f,g) \in V^{n}_{A,\nu} \times V^{n-1}_{A,\nu}: f&= \nabla \cdot g,\, \supp f,g \subset [0,1]\times \Reals^3\},\\
	\text{with norm }\|(f,g)\|_{X^n_{A,\nu}}	&:= \|f\|_{V^n_{A,\nu}}+ \|g\|_{V^{n-1}_{A,\nu}}.
	\end{aligned}
	\end{equation} 
	For $u=(f,g) \in X^n_{A,\nu}$ we write $\partial_t u = (\partial_t f,\partial_t g)$ whenever the right-hand side is well-defined. 
\end{definition}
\begin{remark}
	The validity of our analysis is not subject to the choice of the particular exponent in the weight
	function, and weights of the form $\lambda_c(v)=e^{c |v|}$ or fast power law decay would work equally well.
\end{remark}
The choice of the weight functions $\lambda, \tilde{\lambda}$ is motivated by the following compactness property,
that we will later use to prove the existence of fixed points.
\begin{lemma} \label{Rellich}
	Let $(u_i)_{i\in \Naturals} =((f_i,g_i))_{i\in \Naturals} \subset X^{n+1}_{A,\lambda}$ be a bounded sequence, such that the sequence
	$(\partial_t f_i,\partial_t g_i) \in X^{n+1}_{A,\lambda}$ is bounded as well. Then the sequence $(u_i)$ is
	precompact in $X^n_{A,\tilde{\lambda}}$. 
\end{lemma}
\begin{proof}
	For some $C>0$ there holds $\|(f_i,g_i)\|_{X^{n+1}_{A,\lambda}} + \|(\partial_t f_i, \partial_t g_i)\|_{X^{n+1}_{A,\lambda}} \leq C$.
	Denote by $(\varphi_R)_{R>0} \in C^\infty_c$ a standard sequence of cutoff functions that is one on $B_R$ and vanishes outside
	of $B_{R+1}$. We construct a convergent subsequence $u_{\ell(k)}$ inductively. The region $[0,1] \times B_{R+1}$ is compact, so by
	Rellich's theorem the sequences $(f_i \varphi_1)$, $(g_i \varphi_1)$	have convergent subsequences 
	$f_{\ell_1(i)} \varphi_1 \rightarrow F_1$, $g_{\ell_1(i)} \varphi_1 \rightarrow G_1$  in $V^n_{A,\lambda}$
	and $V^{n-1}_{A,\lambda}$ respectively. Since $V^{n,d}_{A,\lambda}\hookrightarrow V^{n,d}_{A,\tilde{\lambda}} $ embed continuously
	(actually Lipschitz with constant $L\leq1$), the sequences are also convergent in the latter spaces. Now we inductively	extract further convergent subsequences 
	$f_{\ell_k(i)} \varphi_k \rightarrow F_k$ and $g_{\ell_k(i)} \varphi_k \rightarrow G_k$. By construction we have $F_m=F_k$, $G_m=G_k$  
	on $B_k$ for $m\geq k$. We pick a sequence $u_{\ell(k)}$ such that:
	\begin{align*}
	\|f_{\ell(k)} \varphi_k - F_k\|_{V^n_{A,\tilde{\lambda}}} + \|g_{\ell(k)} \varphi_k - G_k\|_{V^{n-1}_{A,\tilde{\lambda}}}\leq \frac{1}{k}.
	\end{align*}
	The sequences $f_{\ell(k)}$, $g_{\ell(k)}$ are Cauchy sequences in $V^n_{A,\tilde{\lambda}}$ and $V^{n-1}_{A,\tilde{\lambda}}$ respectively. 
	To see this, take $i,j\geq k$ and bound:
	\begin{align*}
	\|f_{\ell(i)}-f_{\ell(j)}\|_{V^n_{A,\tilde{\lambda}}}  		 \leq 	&\|(f_{\ell(i)}-f_{\ell(j)}) \varphi_k \|_{V^n_{A,\tilde{\lambda}}} + \|(f_{\ell(i)}-f_{\ell(j)})(1-\varphi_k)\|_{V^n_{A,\tilde{\lambda}}} \\
	\leq 	& \frac{2}{k}  + \frac{1}{k} \|(f_{\ell(i)}-f_{\ell(j)})(1-\varphi_k)\|_{V^n_{A,\lambda}} \longrightarrow 0, 
	\end{align*}
	where we have used that $\tilde{\lambda}(v)\leq \frac{1}{|k|} \lambda(v)$ for $|v|\geq k$. Hence $f_{\ell(k)}$ is a Cauchy sequence.
	The proof for $g_{\ell(k)}$ is similar. Therefore $u_{\ell(k)}$ is precompact in $X^n_{A,\tilde{\lambda}}$.
\end{proof}

We can now formulate the precise statement for the existence of solutions $u_\eps$ of \eqref{cuteq} and convergence to a solution
of the nonlinear Landau equation \eqref{limiteq}. 
\begin{theorem} \label{mainthm1pf}
	Let $m_0,\sigma>0$ and $m(\sigma^2,m_0)$ be the Maxwellian with mass $m_0$ and standard deviation $\sigma$:
	\begin{align} \label{defmaxwellian}
	m(\sigma^2,m_0)(v):= m_0 \frac{e^{-\frac12 \frac{|v|^2}{\sigma^2}}}{(\sigma \sqrt{2\pi})^3}.
	\end{align}
	Let $n\geq 6$ and $v_0 \in H^n_{\lambda}$ satisfy: 
	\begin{align*}
	0 \leq v_0(v) \leq C e^{-\frac12 |v|}. 
	\end{align*}
	There exist $A,C(A)>0$, $\delta_1,\eps_0 \in (0,\frac12]$  such that for all  $\eps,\delta_2 \in (0,\eps_0]>0$  the equation
	\begin{equation} \label{thmepseq}
	\begin{aligned} 
	\partial_t u_\eps 	= 	&\frac1\eps \nabla \cdot \left(\int_0^t K[u_\eps(s)]\left(\frac{t-s}\eps,v\right) \nabla u_\eps(s,v) \ud{s}\right) \\
	-  	&\frac1\eps \nabla \cdot \left(\int_0^t P[u_\eps(s)]\left(\frac{t-s}\eps,v\right) u_\eps(s,v)  \ud{s}\right)  \\ 
	u_\eps(0,\cdot)	&= u_0(\cdot) = m(v) + \delta_2 v_0(v)
	\end{aligned}	
	\end{equation}
	has a strong solution $u_\eps \in V^n_{A,\lambda} \cap C^1([0,\delta_1]; H^{n-2}_{\lambda})$ up to time $\delta_1$ with
	uniform bound:
	\begin{align} \label{apriounif}
	\|u_\eps\|_{V^n_{A,\lambda}} + \|\partial_t u_\eps\|_{V^{n-2}_{A,\lambda}} \leq C(A).
	\end{align} 
\end{theorem}
\begin{remark}
	Our result is valid for small initial perturbations $u_0+\delta_2 v_0$ of the Maxwellian and small times $0\leq t\leq \delta_1$. Notice that the functions $u_\eps$ are solutions
	to \eqref{thmepseq} up to time $\delta_1$, but are defined also for later times.
	In the following, we will write $C,c>0$ for generic large/small constants that are not dependent on other parameters. 
\end{remark}

\begin{theorem}	\label{mainthm2pf}
	For $n\geq 6$ pick $A\geq 1$, $\delta_1\in (0,\frac12]$ and $\eps,\delta_2$ small enough such that Theorem~\ref{mainthm1pf} ensures the existence
	of solutions $u_\eps \in V^n_{A,\lambda} \cap C^1([0,\delta_1]; H^{n-2}_{\lambda})$ of \eqref{thmepseq}. Along
	a sequence $\eps_j\rightarrow 0$ the $u_{\eps_j}$ converge $u_{\eps_j} \rightarrow u$ in $V^{n-3}_{A,\tilde{\lambda}}$,
	$u_{\eps_j} \rightharpoonup u$	in $V^n_{A,\lambda}$, $\partial_t u_{\eps_j} \rightharpoonup \partial_t u$	in $V^{n-2}_{A,\lambda}$. 
	The function $u \in V^n_{A,\lambda} \cap C^1([0,\delta_1];H^{n-4}_{\lambda})$ solves the limit equation up to times $0\leq t\leq \delta_1$:
	\begin{equation} \label{uequation}
	\begin{aligned}
	\partial_t 	u 		&= \nabla \cdot \left( \K[u] \nabla u\right) - \nabla \cdot \left( \Pe[u]  u\right) \\
	u(0,v)			&= m(v) + \delta_2 v_0(v)  	\\
	\K[u](v)				&=  \frac{\pi^2}{4}\int (k \otimes k) |\hat{\phi}(k)|^2 \delta(k\cdot(v-v')) \eta(|v-v'|^2)u(v') \ud{k} \ud{v'} \\
	\Pe[u](v)			&=  \frac{\pi^2}{4}\int (k \otimes k) |\hat{\phi}(k)|^2 \delta(k\cdot(v-v')) \eta(|v-v'|^2) \nabla u(v') \ud{k} \ud{v'}. 			
	\end{aligned}
	\end{equation}
\end{theorem}
In order to show the existence of a strong solution to \eqref{thmepseq}, we will consider mollifications of the equations first, and derive
a priori estimates that are independent of the mollification. We introduce the following notation.
\begin{notation}
	Let $\vphi_\gamma$ be a standard mollifier on $\Reals^3$. For $0<\gamma\leq 1$, define the regularized gradient ${}^\gamma \nabla$ 
	as	${}^\gamma \nabla f(v) := \nabla (\vphi_\gamma * f)$. We define ${}^\gamma \nabla$ to be the standard gradient for $\gamma=0$.
	We will use the following conventions for Laplace transform and Fourier transform:
	\begin{align}
	\La(u)(z) 	&= \int_0^\infty u(t) e^{-z t} \ud{t} \\
	\hat{u}(k)		&= \frac{1}{(2\pi)^{3/2}} \int_{\Reals^3} u(v) e^{-i k \cdot v} \ud{v}.
	\end{align}
\end{notation}
Now we observe that if $u_\eps = u_0 + f_\eps$ is a solution of \eqref{thmepseq}, an equivalent way of stating this is 
\begin{equation} \label{cutfixedeq}
\begin{aligned}
\partial_t u_\eps 	= 	&\frac1\eps {}^\gamma \nabla \cdot \left(\int_0^t K[u_0+f_\eps(s,\cdot)]\left(\frac{t-s}\eps,v\right) {}^\gamma \nabla u_\eps(s,v) 						\ud{s}\right) \\
-  	&\frac1\eps {}^\gamma \nabla \cdot \left(\int_0^t P_\gamma[u_0+f_\eps(s,\cdot)]\left(\frac{t-s}\eps,v\right) u_\eps(s,v)  \ud{s}\right)  \\ 
u_\eps(0,\cdot)	&= u_0(\cdot), \quad  					P_\gamma		= {}^\gamma \nabla \cdot K, \quad \text{$K$ as defined in \eqref{KPdef}}
\end{aligned}
\end{equation}
holds for $\gamma=0$.
We will show a priori estimates for the above equation for $0<\gamma\leq 1$ and later recover the case $\gamma=0$ 
as a limit. We start our analysis by writing $K$ and $P$ in a more convenient form.
\begin{lemma} \label{KPfirstvers}
	The operator $K$ defined in \eqref{KPdef} and $P_\gamma= {}^\gamma \nabla \cdot K$  can be expressed by the formulas:
	\begin{align}
	K[u](t,v)	&=  \int (k \otimes k) |\hat{\phi}(k)|^2 \cos( t (v-v')\cdot k) \eta(|v-v'|^2)u(v') \ud{k} \ud{v'} \label{Ktensor} \\
	P_\gamma[u](t,v)	&=  \int (k \otimes k) |\hat{\phi}(k)|^2 \cos( t (v-v')\cdot k) \eta(|v-v'|^2) {}^\gamma \nabla u(v') \ud{k} \ud{v'}.
	\end{align}
\end{lemma}
\begin{proof}
	The formula for $P_\gamma$ follows from the one for $K$, so we only prove this one. Plancherel's theorem 
	allows to rewrite:
	\begin{align*}
	K[u](t,v) &= \int  \nabla \phi(x) \otimes \nabla \phi(x-t(v-v'))
	u(v') \eta(|v-v'|^2) \ud{v'}\ud{x}	 \\
	&= \int  (k \otimes k) |\hat{\phi}(k)|^2 e^{- i t k\cdot (v-v')}
	u(v') \eta(|v-v'|^2) \ud{v'}\ud{k}.
	\end{align*}
	Since $K$ only takes real values, we can symmetrize the exponential and obtain
	\begin{align*}
	&\int  (k \otimes k) |\hat{\phi}(k)|^2 e^{- it k\cdot(v-v')}
	u(v') \eta(|v-v'|^2) \ud{v'}\ud{k} \\
	=		&\int  (k \otimes k) |\hat{\phi}(k)|^2 \cos\left( t k\cdot (v-v') \right)
	u(v') \eta(|v-v'|^2) \ud{v'}\ud{k}	,
	\end{align*}
	proving the claim.
\end{proof}
We will omit the index $\gamma\geq 0$ in notation, when there is no risk of confusion.
Controlling the nonlinearity inside $K$ and $P$ strongly relies on being able to bound spatial derivatives of $u_\eps$. Therefore
we consider differentiations of the equation. Let $\alpha \in \Naturals^3$ be a multi-index. With 
the convention $\binom{\alpha}{\beta}=\prod_{j=1}^3 \binom{\alpha_j}{\beta_j}$, the function 
$D^\alpha u_\eps = \frac{\partial^\alpha u_\eps}{\partial v_1^{\alpha_1}\partial   v_2^{\alpha_2}\partial v_3^{\alpha_3}}$ (formally)
satisfies the equation:
\begin{align*}
\partial_t D^\alpha u_\eps = \sum_{\beta_1+\beta_2=\alpha} \binom{\alpha}{\beta_1} 
\frac{1}{\eps} \left(\nabla \cdot \left(\int_0^t D^{\beta_1} K D^{\beta_2} \nabla u_\eps \ud{s}\right)     
-   \nabla \cdot \left(\int_0^t D^{\beta_1} P D^{\beta_2} u_\eps  \ud{s}\right) \right).
\end{align*} 
In order to have a short notation for the terms appearing on the right-hand side of the equation above, we introduce
the following notation. 
\begin{notation} \label{Adef}
	Let $n\in \Naturals$ and $\alpha, \beta$ be multi-indices with $\beta \leq \alpha$, $|\alpha|\leq n-1$ and $\nu,u_\eps \in V^n_{A,\tilde{\lambda}}$. 
	For $\gamma \in (0,1]$ we define:
	\begin{equation}\label{curlyA}
	\begin{aligned}
	\A^{\alpha,\beta}_\gamma[\nu](u_\eps) = &\frac1\eps 
	\left(\int_0^t D^\beta K[\nu(s)]\left(\frac{t-s}\eps,v\right) {}^\gamma \nabla D^{\alpha-\beta} u_\eps(s,v) \ud{s}\right) \\
	-  		&\frac1\eps  \left(\int_0^t D^\beta P_\gamma[\nu(s)]\left(\frac{t-s}\eps,v\right) D^{\alpha-\beta}u_\eps(s,v)  \ud{s}\right).
	\end{aligned}
	\end{equation}
	Furthermore, for $m\in \Naturals$, $u \in V^m_{A,\lambda}$, we set: 
	\begin{align}
	|u|_{F^m}(z,v) := \sum_{|\beta|\leq m} |\La(D^\beta u)(z,v)|. \label{FMdef}
	\end{align} 
\end{notation}
The equation \eqref{cutfixedeq} has an averaged in time coercivity property, which we will prove by
showing nonnegativity for certain quadratic functionals $Q$. This allows to show that $u_\eps$ inherits
decay and regularity properties from the initial datum. 
We have the following basic a priori estimate for solutions $u_\eps$ of \eqref{cutfixedeq}:
\begin{lemma} \label{lemReduc}
	Let $n\in \Naturals$, $A,\eps,\gamma>0$ and $u_\eps \in C^1([0,T];H^{n}_{\lambda}) $ be a solution to \eqref{cutfixedeq} for $T>0$ arbitrary. Then for $|\alpha|\leq n$ we can 			bound:
	\begin{align*}
	A \int_0^T  \int \lambda(v) |D^{\alpha}u_\eps(t,v)|^2 e^{-At} \ud{t} \ud{v} \leq - 2 Q^{\alpha}_{\eps,A}[u_0+f_\eps](u_\eps \cf_{[0,T]}) + \|\lambda^\frac12 D^\alpha u_0\|^2_{L^2}.	
	\end{align*} 
	Here $Q^\alpha_{\eps,A}[\nu](u)$ is given by (we drop the index $\gamma$ if there is no risk of confusion):
	\begin{align}
	Q_{\eps,A}^\alpha[\nu ](u) = &\sum_{\beta \leq \alpha} \binom{\alpha}{\beta} Q_{\eps,A}^{\alpha,\beta} [\nu ](u) \label{Qdef} \\
	Q_{\eps,A}^{\alpha,\beta}[\nu](u) = 	& \int_0^\infty \int \frac{e^{-At}}{\eps}    {}^\gamma \nabla (D^\alpha u(t)\lambda)   \int_0^t  D^{\alpha-\beta} K[\nu(s)](\frac{t-s}{\eps})  {}^\gamma \nabla D^\beta u(s) \ud{s}\ud{v} 	\ud{t} \label{Kterm}\\
	-	&\int_0^\infty\int \frac{e^{-At}}{\eps}    {}^\gamma \nabla (D^\alpha  u(t) \lambda)  \int_0^t  D^{\alpha-\beta} P_\gamma[\nu(s)](\frac{t-s}{\eps})		 D^\beta u(s) \ud{s}\ud{v} \ud{t}.	\label{Pterm}
	\end{align}
\end{lemma}
\begin{proof}
	Follows by a simple computation:
	\begin{align*}
	&A \int_0^T   \int_{\Reals^3} \lambda(v) |D^\alpha u_\eps(t,v)|^2 e^{-At} \ud{t} \ud{v} \\
	=  &-  \int_0^T
	\int_{\Reals^3} \lambda(v)  D^\alpha u_\eps(t,v)^2 \partial_t(e^{-At}) \ud{t} \ud{v} \\
	\leq &2  \int_0^T
	\int_{\Reals^3} \lambda(v)  D^\alpha u_\eps(t,v) \partial_t D^\alpha u_\eps (t,v) e^{-At} \ud{t} \ud{v}	+  \int \lambda(v) 	|D^\alpha u_0|^2 	\ud{v}  \\
	= & -2  Q^\alpha_{\eps,A}[u_0+f_\eps](u_\eps\cdot \cf_{[0,T]}) + \|\lambda^\frac12 D^\alpha u_0 \|^2_{L^2},
	\end{align*} 
	where in the last line the equation is used.
\end{proof}

The following analogue of Plancherel's theorem for Laplace transforms
will be useful throughout the paper.
\begin{lemma} \label{Plancherel}
	Let $\mu_A(\ud{t}) := e^{-At} \ud{t}$. Then for $u,v \in L^2(\mu_A)$ we have:
	$$ (2\pi)^{\frac12} \int_0^\infty e^{-At} \ol{u}(t) v(t) \mu_A(\operatorname{dt}) = \int_\Reals
	\ol{\La(u)}\left(\frac{A}{2}+i\omega\right)
	\La(v) \left(\frac{A}{2}+i\omega\right) \ud{\omega}.$$
\end{lemma}
Our proof strongly relies on the geometry of both complex and real vectors. To avoid
confusion  we introduce the following notation.
\begin{definition} \label{Vectornotation}
	For $v,w \in \Reals^3$ we will use the notation $v\cdot w = \sum_i v_i w_i$ for the Euclidean scalar product.
	The inner product of complex vectors $V, W \in \Complex^3$ we denote by $\langle V,W\rangle= \sum_i \ol{V}_i W_i$.
	We will use the notation $|\cdot |$ for the vector norms induced by each of the inner products, as well
	as the matrix norm induced by this norm.	
	Moreover for $0 \neq V\in \Complex^3$ and $W\in \Complex^3$ we define the 
	orthogonal projections $P_V W$ and $P_V^\perp W$ as:
	\begin{align}
	P_V W 			&:= \left(\frac{\langle V, W \rangle}{|V|}\right) \frac{V}{|V|}, \quad		P_V^\perp W	:= W - P_V W.
	\end{align}
\end{definition}
For future reference, we compute the Laplace transform of $K[u](t,v)$ in $t$.
With our particular choice of potential, some of the integrals are explicitly computable,
as is stated in the following auxiliary Lemma.
\begin{lemma} \label{Mlemma}
	For $\Re(z)\geq 0$, $v\in \Reals^3$  let $M_1(z,v),M_2(z,v)$ be the matrix-valued functions defined by
	\begin{equation}\label{defM}
	\begin{aligned}
	M_1(z,v)	&:= \frac{\pi^2}{4 |v|} \frac{1}{1+\frac{z}{|v|}} P_v^\perp, \quad 			 			M_2(z,v)	:= \frac{\pi^2}{4 |v|} \frac{ \frac{z}{|v|}}{(1+\frac{z}{|v|})^2} P_{v} .		
	\end{aligned}
	\end{equation} 
	Then we have the following identity:
	\begin{align}
	&\int (k \otimes k) |\hat{\phi}(k)|^2 \frac{z}{z^2+(k\cdot v)^2} \ud{k} =  M_1(z,v)+M_2(z,v).
	\end{align}
\end{lemma}
\begin{proof}
	We decompose $k \in \Reals^3$ into $k= u w + w^\perp$, where $w=\frac{v}{|v|}$.
	We insert the explicit form of the Fourier transform of $\phi$ (cf. \eqref{FTpotential}) to rewrite the integral as (here $a^{\otimes 2}= a \otimes a$):
	\begin{align*}
	\int (k \otimes k) |\hat{\phi}(k)|^2 \frac{z}{z^2+(k\cdot v)^2} \ud{k} 
	=	&\int_\Reals \int_{\operatorname{span}(w)^\perp}  \frac{(u w + w^\perp )^{\otimes 2}}{(1+u^2 +|w^\perp|^2)^3} \ud{w^\perp}\frac{z}{z^2+(u |v|)^2}  					\ud{u} \\
	=	&\frac1{|v'|}\int_\Reals \int_{\operatorname{span}(w)^\perp}  \frac{(u w + w^\perp )^{\otimes 2}}{(1+u^2 +|w^\perp|^2)^3} \ud{w^\perp}\frac{\frac{z}{|				v|}}{(\frac{z}{|v|})^2+ u ^2}  					\ud{u} \\
	=	&\frac1{|v|}\int_\Reals \int_{\operatorname{span}(w)^\perp}  \frac{((u w)^{\otimes 2}+(w^\perp)^{\otimes 2})}{(1+u^2 +|w^\perp|^2)^3} \ud{w^\perp}\frac{\frac{z}{|				v'|}}{(\frac{z}{|v|})^2+ u ^2}  					\ud{u},
	\end{align*} 
	where we used that the mixed terms $u w \otimes w^\perp$ do not contribute to the integral due to the symmetry of the integrand. Now the inner integral
	is explicit:
	\begin{align*}
	\int_{\operatorname{span}(w)^\perp}  \frac{((u w)^{\otimes 2}+(w^\perp)^{\otimes 2})}{(1+u^2 +|w^\perp|^2)^3} \ud{w^\perp} 
	=	&u^2  \int_0^\infty  \frac{2 \pi r P_w}{(1+u^2+r^2)^3 } \ud{r} + \int_0^\infty \frac{\pi r^3 P_{w}^\perp}{(1+u^2 +r^2)^3} \ud{r} \\
	=	&\frac{\pi u^2}{2 (1+u^2)^2} P_w + \frac{\pi}{4(1+u^2)} P_{w}^\perp.	
	\end{align*}
	Inserting this back into the full integral gives two explicit integrals:
	\begin{align*}
	&\frac1{|v|}\int_\Reals \int_{\operatorname{span}(w)^\perp}  \frac{((u w)^{\otimes 2}+(w^\perp)^{\otimes 2})}{(1+u^2 +|w^\perp|^2)^3} 								\ud{w^\perp}\frac{\frac{z}{|	v|}}{(\frac{z}{|v|})^2+ u ^2}  					\ud{u}	\\
	=	&\frac1{|v|}\int_\Reals \left( \frac{\pi u^2}{2 (1+u^2)^2} P_w + \frac{\pi}{4(1+u^2)} P_{w}^\perp \right)\frac{\frac{z}{|	v|}}{(\frac{z}{|				v|})^2+ u ^2}  					\ud{u} \\
	=	& \frac{\pi^2}{4 |v|} \left(\frac{ \frac{z}{|v|}}{(1+\frac{z}{|v|})^2} P_w + \frac{1}{1+\frac{z}{|v|}} P_{w}^\perp \right) \\
	=	& M_1(z,v)+ M_2(z,v),
	\end{align*}
	which implies the statement of the lemma.	
\end{proof}
Now the Laplace transform $\La(K[u])$ can be rewritten in a more explicit form.
\begin{lemma} \label{Laplacerep}
	Let $u \in H^n_{\tilde{\lambda}}$, $n\geq 2$ and $\La(K[u])(z,v)$ be the Laplace transform of $K[u]$, i.e.
	\begin{align*}
	\La(K[u])(z,v) = \int_0^\infty K[u](t,v) e^{-z t} \ud{t}.
	\end{align*}
	Then $\La(K[u])$ is given by the formula:
	\begin{align} 
	\La(K[u])(z,v) 	&=  \int (M_1+M_2)(z,v-v') u(v')\eta(|v-v'|^2)  \ud{v'} \label{KLapl} .
	\end{align}
	In particular, the matrix $\La(K[u])$ is symmetric. For the operator $P_\gamma$ introduced
	in \eqref{cutfixedeq} we have the formula:
	\begin{align}
	\La(P_\gamma[u])(z,v) &= \int (M_1+M_2)(z,v-v') {}^\gamma \nabla u(v')\eta(|v-v'|^2) \ud{v'} \label{PLapl}.
	\end{align}
\end{lemma}
\begin{proof}
	Follows from $\La(\cos(\alpha t))(z)=\frac{z}{z^2+\alpha^2}$, Lemma \ref{KPfirstvers} and Lemma \ref{Mlemma}.
\end{proof}

\subsection{Strategy of the proofs of Theorems \ref{mainthm1pf} and \ref{mainthm2pf}}
We can now outline the structure of this paper, and introduce the key
steps in the proofs of the Theorems \ref{mainthm1pf} and \ref{mainthm2pf}:
\begin{enumerate}[(i)]
	\item  In Section \ref{Sec:Linear} we prove that the linear equation 
	\begin{equation} \label{linearequation}
	\begin{aligned}
	\partial_t u_\eps 	= 	&\frac1\eps \nabla \cdot \left(\int_0^t K[u_0]\left(\frac{t-s}\eps,v\right) \nabla u_\eps(s,v) \ud{s}\right) \\
	-  	&\frac1\eps \nabla \cdot \left(\int_0^t P[u_0]\left(\frac{t-s}\eps,v\right) u_\eps(s,v)  \ud{s}\right)  \\ 
	u_\eps(0,v)	&= u_0(v) ,
	\end{aligned}
	\end{equation}
	has a solution $u_\eps \in V^n_{A,\lambda}\cap C^1(\Reals^+;H^{n-2}_\lambda)$. The proof is based on the fact that the 
	equation is dissipative in a time averaged sense, and strongly relies on the convolution
	structure of the equation in Laplace variables. Symbolically the equation in Laplace
	variables looks similar to:
	\begin{align*}
	z \La(u)(z,v) &= \nabla \cdot (\tilde{K}(z,v) \nabla \La(u)(z,v)) + u_0(v). 
	\end{align*}
	We show that for $\Re(z)>0$, the real part of the matrix $\tilde{K}(z,v)$ is nonnegative. This
	is quantified in Lemma \ref{coercivitylemma} in terms of the quadratic operators $Q^\alpha_{\eps,A}[u_0]$ (cf. \eqref{Qdef}).
	\item  In order to solve the nonlinear problem, we have to allow for time dependent 
	functions inside the operator $K$. We therefore consider equation \eqref{cutfixedeq} for a 
	fixed function $f_\eps$ and mollified derivatives ${}^\gamma \nabla$:
	\begin{equation} \label{cutfixedeq2} 
	\begin{aligned}
	\partial_t u_\eps 	= 	&\frac1\eps {}^\gamma \nabla \cdot \left(\int_0^t K[u_0+f_\eps(s,\cdot)]\left(\frac{t-s}\eps,v\right) {}^\gamma \nabla u_\eps(s,v) 						\ud{s}\right) \\
	-  			&\frac1\eps {}^\gamma \nabla \cdot \left(\int_0^t P_\gamma[u_0+f_\eps(s,\cdot)]\left(\frac{t-s}\eps,v\right) u_\eps(s,v)  \ud{s}\right)  \\ 
	u_\eps(0,\cdot)			&= u_0(\cdot), \quad  			P_\gamma		= {}^\gamma \nabla \cdot K.
	\end{aligned}
	\end{equation}
	In Subsection \ref{Psisubsection} we identify a closed, nonempty, convex subset $\Omega$  of $X^n_{A,\tilde{\lambda}}$ (defined in \eqref{omegadef})
	such that the local in time solution operator $\Psi_{\delta_1}$ to \eqref{cutfixedeq2}:
	\begin{equation} \label{Psi1}
	\begin{aligned}
	\Psi_{\delta_1}: \Omega 	&\longrightarrow X^n_{A,\lambda} \\
	(f,F) 	&\mapsto \left((u-u_0)\kappa_{\delta_1}, \A^{0,0}_\gamma[f](u)\kappa_{\delta_1}	\right), \quad \text{where $u$ solves \eqref{cutfixedeq2}}
	\end{aligned}
	\end{equation}
	is well-defined. Here $\kappa_{\delta_1}$ is a cutoff function that localizes to small times.  
	Notice that the solution operator maps from $X^n_{A,\tilde{\lambda}}$
	to $X^n_{A,\lambda}$, thus we gain decay.		
	The proof is based on proving that replacing the constant kernel $K[u_0]$ by $K[u_0+f]$
	amounts to a small perturbation. The main assumption for this, and the defining property of the set $\Omega$ is that for some $A,R>0$ and small $\delta>0$,
	we can bound $\La(f)$ on the line $\Re(z)=\frac{A}{2}$ by:
	\begin{align} \label{decayassumption}
	|\La(f)(z,v)| \leq \left(\frac{\delta}{1+|z|^2} + \frac{R \eps |z|}{(1+\eps |z|)(1+ |z|^2)}\right) e^{-\frac12 |v|}.
	\end{align}
	Under assumption \eqref{decayassumption} we obtain an a priori estimate on the solutions and
	their time derivatives:
	\begin{equation} \label{symbolicaprio}
	\begin{aligned}
	\|\Psi_{\delta_1}(f,F)\|_{X^n_{A,\lambda}}+\|\partial_t \Psi_{\delta_1}(f,F)\|_{X^{n-2}_{A,\lambda}} &\leq C \\
	\|\Psi_{\delta_1}(f,F)\|_{X^{n+1}_{A,\lambda}} +\|\partial_t \Psi_{\delta_1}(f,F)\|_{X^{n+1}_{A,\lambda}} &\leq C(\gamma).
	\end{aligned}			
	\end{equation}
	It is crucial that the first estimate is uniform in the mollifying parameter $\gamma>0$.
	In Section \ref{Sec:Boundary} we prove that the operator $\Psi_{\delta_1}$ introduced in \eqref{Psi1}
	leaves the set $\Omega$ invariant, for $\delta_1>0$ small, close to the Maxwellian and $\eps>0$ small. 
	
	Now, for $\gamma>0$, we infer the existence of a fixed point of $\Psi_{\delta_1}$ from \eqref{symbolicaprio} and Schauder's theorem.
	Here we use bounded sequences in  $X^{n+1}_{A,\lambda}$ with bounded time derivative are precompact in $X^{n}_{A,\tilde{\lambda}}$, as proved in Lemma \ref{Rellich}.	
	This compactness property allows to take the limit $\gamma\rightarrow 0$ and thus to prove Theorem \ref{mainthm1pf}. 
	Here we make use of the uniform estimate in \eqref{symbolicaprio}. The proof of Theorem \ref{mainthm2pf} follows by passing $\eps\rightarrow 0$
	using Lemma \ref{Rellich} yet again.
\end{enumerate}

A key point of the analysis is the invariance of the set $\Omega$ under $\Psi_{\delta_1}$, which is proved in Section \ref{Sec:Boundary}. 
The proof relies on recovering the decay assumption \eqref{decayassumption}. We can think of
functions $f$ satisfying \eqref{decayassumption} as a sum $f=f_1+f_2$.
Here $f_1$ satisfies $|\La(f_1)(z)|\leq \frac{\delta}{1+|z|^2}$, which can be thought of as an estimate of the form $\|\partial^2_{t t} f_1\|_{L^1}\lesssim \delta$,
and $f_2$ satisfies $|\La(f_2)(z)|\leq \frac{R \eps |z|}{(1+\eps |z|)(1+ |z|^2)}$, which can be understood as $\|\partial_t f_2\|_{L^1} \lesssim R\eps$ and $\|\partial^2_{t t} f_2\|_{L^1}\lesssim R$. This is only a heuristic consideration, since $L^\infty$/$L^1$ duality does not hold for Laplace transform.
A typical function of this form is $f_2^\eps(t)= \eps^2 \Phi(t/\eps)$. The behavior of $f_1$ close to $t=0$ is more complicated, since it involves a boundary layer.
Indeed, there is necessarily a boundary layer in $\partial_{tt} u_\eps$ in equation \eqref{cutfixedeq2}.
To see this, let $u$ be the solution
of the limit (Landau-) equation \eqref{limiteq}, and $u_\eps$ the solution to \eqref{cutfixedeq2}. Then, starting
away from equilibrium, we have:
\begin{align*}
\partial_{t} u_\eps(0,v) = 0, \quad \partial_{t} u(0,v) \neq 0.
\end{align*}
So in the limit $\eps\rightarrow 0$, the second derivative necessarily grows infinitely large close to the origin.

The quadratic decay of the Laplace transforms can be obtained by a bootstrap argument.
To fix ideas, we observe that \eqref{cutfixedeq2} in Laplace variables is similar to:
\begin{align} \label{symbollaplace}
z \La(u-u_0) &= \nabla \cdot \left(\tilde{K}(\eps z) (\nabla \La(u) + \nabla \La(u) * \La(f))\right). 
\end{align}		
In Subsection \ref{Psisubsection} we prove that $\nabla^m \La(u)$ are bounded in a weighted $L^2$ space in time and velocities.
This can be bootstrapped to pointwise estimates: First we remark that localizing $\supp u \subset [0,1] \times \Reals^3$
gives an $L^\infty$ estimate for $\nabla^m \La(u)$. Assuming $|\tilde{K}(z)|\leq \frac{1}{1+|z|}$, equation 
\eqref{symbollaplace} gives an estimate like:
\begin{align*}
|\nabla^m \La(u-u_0)(z,v) | \leq \frac{C}{(1+\eps|z|)|z|} e^{-\frac12 |v|}.
\end{align*}
Plugging this estimate back into \eqref{symbollaplace} proves quadratic decay of the Laplace transforms: 
\begin{align*}
|\nabla^m \La(u-u_0)(z,v) | \leq \frac{C}{(1+\eps|z|)|z|^2} e^{-\frac12 |v|}.
\end{align*}
In order to show invariance of the set $\Omega$ we need 
the same estimate with a small prefactor, as in estimate \eqref{decayassumption}.
We split the solution into a well-behaved part and the boundary layer mentioned before. 
For the first part, we use smallness of the cutoff  time $\delta_1>0$ to get a small prefactor additional to the quadratic decay.
The estimate of the boundary layer, close to the Maxwellian, is obtained by isolating and estimating it explicitly.
This is the content of Subsection \ref{subsecboundarylayer}, and the most delicate part of the analysis. 

We remark that there are two points where our proof is non-constructive, namely the proof of existence of solutions $u_\eps$ via  Schauder's fixed point theorem, and the convergence of the sequence $u_\eps$ to the solution $u$ of the Landau equation. Therefore, an explicit rate of convergence of the sequence $u_\eps$ to $u$ cannot directly be derived with our method. 

\subsection{A well-posedness result for the regularized problem (\ref{cutfixedeq2})}
Before we start with the analysis of the equation in more detail, we first prove
that the equation \eqref{cutfixedeq2} with frozen nonlinearity indeed has a solution.
This standard Picard-iteration argument is given in the following Lemma. 
\begin{lemma} \label{lemexistence}
	Let $n\in \Naturals$, $\gamma,\eps>0$ and $u_0 \in H^n_{\lambda}$. Further
	assume there is a constant $C>0$ such that $|f_\eps(t,v)| \leq C e^{-\frac12 |v|}$
	and $\supp f_\eps \subset [0,1]$.  Then there exists a (unique)
	global in time solution $u_\eps \in C^1([0,\infty);H^n_{\lambda})$ to:
	\begin{equation}
	\begin{aligned} \label{gammaepseq} 
	\partial_t u_\eps 	= 	&\frac1\eps {}^\gamma \nabla \cdot \left(\int_0^t K[u_0+f_\eps(s)]\left(\frac{t-s}{\eps},v\right) {}^\gamma \nabla 						u_\eps(s,v) 			\ud{s}\right) \\
	-  	&\frac1\eps {}^\gamma \nabla \cdot \left(\int_0^t P_\gamma[u_0+f_\eps(s)]\left(\frac{t-s}{\eps},v\right) u_\eps(s,v)  \ud{s}\right) \\ 
	u_\eps(0,\cdot)	&= u_0(\cdot).
	\end{aligned} 
	\end{equation}
\end{lemma}
\begin{proof}
	For better notation, we introduce a shorthand for the right-hand side of the equation:
	\begin{align*}
	\B(u)(t,t',v) := 	&\frac1\eps {}^\gamma \nabla \cdot \left(\int_{t'}^t K[u_0+f_\eps(s)]\left(\frac{t-s}\eps,v\right) {}^\gamma \nabla u(s,v) \ud{s}\right) \\
	-  		&\frac1\eps {}^\gamma \nabla \cdot \left(\int_{t'}^t P_\gamma[u_0+f_\eps(s)]\left(\frac{t-s}\eps,v\right) u(s,v)  \ud{s}\right). 
	\end{align*}
	The claim follows from a standard Picard-type argument. Let $T>0$ to be chosen later. Consider the mapping
	\begin{align*}
	\D : C^1([0,T];H^n_{\lambda}) &\rightarrow C^1([0,T];H^n_{\lambda}) \\
	u 			&\mapsto \D(u),
	\end{align*} 
	where $\D(u)$ is given by:
	\begin{align}
	\D(u)(t,v) := u_0(v) + \int_0^t \B(u)(s,v) \ud{s}. 
	\end{align}
	The mapping is $\D$ contractive for small times. More precisely we have:
	\begin{align}
	\|\B(u)(t,t',\cdot) \|_{H^n_{\lambda}} \leq C |t-t'| \sup_{t'\leq s \leq t} \|u(s,\cdot)\|_{L^2_\lambda}. \label{contrest}
	\end{align}
	Hence, there exists a $T_1>0$ such that $\D$ is contractive and we obtain a unique solution for $T\leq T_1$ . 
	Assume we already have constructed the solution $u$ 
	up to time $m T_1$ for $m\in \Naturals$. Consider the mapping:
	\begin{align*}
	\D_m : C^1([m T_1,(m+1)T_1];H^n_{\lambda}) &\rightarrow C^1([m T_1,(m+1)T_1];H^n_{\lambda}) \\
	w 			\mapsto \D_m(w)&= u(m T_1,v) +\int_{m T_1}^{T} \B(w)(s,v) \ud{s}.
	\end{align*}
	By \eqref{contrest} this mapping is contractive and we can pick the same small time $T_1$ in each step of 	the	induction.
\end{proof}

\section{The linear equation (\ref{linearequation})} \label{Sec:Linear}
The linear equation \eqref{linearequation} has an averaged-in-time coercivity property. 
We will prove this using geometric arguments that resemble the ones used for the Landau equation,
see for instance \cite{desvillettes_spatially_2000}. For shortness we introduce the following notation. 
\begin{notation}
	For $z\in \Complex$ and $v \in \Reals^3$ define: 
	\begin{align}
	\alpha(z,v) &:= \frac{|\Im(z)|}{1+|v|},\quad		\beta(z,v)	:= \frac{|\Re(z)|}{1+|v|}.
	\end{align}
	Further we define the following positive functions $C_1$, $C_2$ and $C_3$:
	\begin{align}
	C_1(z,v)	&= 	\frac{1}{(1+|v|)(1+\alpha(z,v))^2} \label{C1} \\
	C_2(z,v)	&= 	\frac{\beta(z,v)+\alpha(z,v)^2 }{(1+|v|)(1+ \alpha(z,v))^4} \label{C2} \\
	C_3(z,v)	&= 	\frac{\beta(z,v)+\alpha(z,v)+\alpha(z,v)^2 }{(1+|v|)(1+ \alpha(z,v))^4} \label{C3}. 
	\end{align}	
	Let  $0 \neq v \in \Reals^3$, $V,W\in \Complex^3$. We define the anisotropic norm:
	\begin{align}
	|W|_v := |P_{v}^\perp W| + \frac{|P_v W|}{1+|v|}  \label{Vnrom}, 
	\end{align}
	and the weight functionals $B_1(z,v)(V,W)$, $B_2(z,v)(V,W)$ given by:
	\begin{align}
	B_1(V,W) = C_1(z,v)|V|_v |W|_v +  C_2(z,v) |P_v V| |P_v W| \label{quadrC} \\
	B_2(V,W) = C_1(z,v)|V|_v |W|_v +  C_3(z,v) |P_v V| |P_v W|. \label{quadrCu}	
	\end{align}
\end{notation}
The following straightforward analysis lemma we will use to bound real and imaginary part of the matrices
$M_i$ defined in \eqref{defM} from above and below.
\begin{lemma}
	Let $z\in \Complex$ with $0\leq \Re(z) \leq 1$. The following bounds hold:
	\begin{align}
	\Re(\frac{z}{(1+z)^2}) 		&\geq c \frac{\Re(z) + |\Im(z)|^2}{(1+|\Im(z)|)^4}   \label{Complex1} \\
	|\Im(\frac{z}{(1+z)^2})| 	&\leq C	 \frac{\Re(z) + |\Im(z)|+|\Im(z)|^2}{(1+|\Im(z)|)^3}	\label{Complex2} \\
	\Re(\frac{1}{1+z}) 		&\geq c \frac{1}{(1+|\Im(z)|)^2}   \label{Complex3} \\
	|\Im(\frac{1}{(1+z)})| 	&\leq C	 \frac{|\Im(z)|}{(1+|\Im(z)|)^2}	\label{Complex4}.
	\end{align}
\end{lemma}
\begin{proof}
	To prove \eqref{Complex1}-\eqref{Complex2}, we rewrite the fraction as:
	\begin{align*}
	\frac{z}{(1+z)^2} 	&= \frac{z + 2 |z|^2 + \ol{z} |z|^2}{|1+z|^4}.
	\end{align*}
	Since the real part of $z$ is bounded and
	nonnegative by assumption, \eqref{Complex1} follows immediately.
	For the proof of \eqref{Complex2} we include the computation:
	\begin{align*}
	|\Im(\frac{z}{(1+z)^2})|	&\leq C \frac{|\Im(z)| + (\Re(z)^2 + \Im(z)^2)(1+|\Im(z)|) } {|1+z|^4} \\
	&\leq C \frac{\Re(z) +|\Im(z)|+ |\Im(z)|^2 } {|1+z|^3},
	\end{align*}
	proving also the second claim. The inequalities \eqref{Complex3} and \eqref{Complex4} are immediate.
\end{proof}
The following simple lemma provides an estimate for the derivatives of the matrices $M_i$ defined in \eqref{defM}.
\begin{lemma} \label{derlemma}
	For a multi-index $\beta \in \Naturals^3$, $\Re(z)\geq 0$, $i=1,2$ and $v\in \Reals^3$, $V,W \in \Complex^3$, we can estimate:
	\begin{align}
	|\langle V, D^\beta( M_i(z,v) \eta(|v|^2)) W\rangle| \leq \frac{C_{|\beta|}|V||W|}{(1+|v|^{|\beta|+1})(1+\alpha(z,v))}  \eta(16|v|^2). 
	\end{align}
	Here $\eta$ is the cutoff function introduced in Notation \ref{defpotential}.
\end{lemma}
\begin{proof}
	With Leibniz's rule, we can split the derivative into:
	\begin{align*}
	D^\beta ((M_1 + M_2)(z,v) \eta(|v|^2)) &= \sum_{\beta_2 \leq \beta} \binom{\beta}{\beta_2} 
	D^{\beta-\beta_2} ((M_1 + M_2)(z,v)) D^{\beta_2} (\eta(|v|^2)).	
	\end{align*}
	By construction of the fixed cutoff function $\eta$ we can estimate:
	\begin{align}
	|\nabla^m \eta(r)| \leq \frac{C}{1+|r|^m} |\eta(16r)|. \label{etaderiv}
	\end{align}
	We write $M_1$, $M_2$ defined in \eqref{defM} as :
	\begin{align*}
	M_1(z,v) &= \frac{\pi^2}{4(z+|v|)} P_{v}^\perp, \quad	M_2(z,v) = \frac{\pi^2 z}{4(z+|v|)^2} P_v.
	\end{align*}
	The operators $P_v$, $P^\perp_v$ are zero-homogeneous in $v$. So for every $c>0$ we can estimate:
	\begin{align}
	|\nabla^n_v M_i(z,v)| \leq  \frac{C |M_i(z,v)|}{1+|v|^n}\leq \frac{C}{(1+|v|)^{n+1}(1+\alpha(z,v))}   \quad \text{for $i=1,2$, $|v|\geq c>0$.}\label{Mderiv}
	\end{align} 
	Combining \eqref{etaderiv} and \eqref{Mderiv} gives the claim.
\end{proof} 
The following Lemmas prove coercivity of the matrix $\La(K)[u](v)$, which becomes anisotropic as $|v|\rightarrow \infty$. 
The crucial geometric argument is contained in the following Lemma, that in our setting needs to be valid for complex vectors (since we apply
it to Laplace transforms).
\begin{lemma} \label{geometry}
	For $0 \neq V \in \Complex^3$ and $0 \leq r \leq 1$, let $D_V(r)$ be given by: 
	\begin{align*}		
	D_V(r) = \{v' \in \Reals^3: \frac12 \leq |v'|\leq 1,\, \frac{|\langle v', V \rangle |}{|v'||V|}\geq r\}.
	\end{align*}
	There exists a constant $c>0$ such that for all  $v\in \Reals^3$, $|v|\geq 2$ the following statements hold:
	\begin{align}
	&\text{for $0\neq V \in \Complex^3$:} 	&\operatorname{Vol}(D_V(1/8))	&\geq c, 											\label{geometry1}\\
	&\text{for $V \in \Complex^3$ $\exists$  $0\neq W \in \Complex^3$ $\forall$ $v'\in D_W(1/8)$ :} 	&|P_{v-v'}^\perp V| +|P_{v-(-v')}^\perp V|&\geq c |V|_v 				\label{geometry2},
	\end{align}
	where the anisotropic norm $|\cdot|_v$ was introduced in \eqref{Vnrom}. Furthermore for $v\in \Reals^3$, $V \in \Complex^3$, define 
	$$E(v,V)=\{v'\in B_1(0)\subset \Reals^3: |\langle v'+v, V \rangle|\geq |\langle v, V \rangle | \}.$$
	There exists $c>0$ such that  for all $v\in \Reals^3$, $|v|\geq 2$:
	\begin{align}
	|P_{v-v'} V| 			&\geq c |P_v V| 		&&\text{for $v'\in E(v,V)$ }															\label{geometry4} \\
	\operatorname{Vol}(E(v,V))					&\geq c >0.																\label{geometry5}	
	\end{align}
\end{lemma} 
\begin{proof}
	The inequality \eqref{geometry1} is clear if $0 \neq V\in \Reals^3$ is real.
	Moreover, there is a constant $c>0$ such that $\operatorname{Vol}(D_V(r))\geq c>0$ for $0\leq r\leq \frac34$ and $V\in \Reals^3$.
	Let now $V= V_R + i V_I \in \Complex^3$, where at least one of the vectors $V_R, V_I \in \Reals^3$ is nonzero, and let $W$ be the longer vector of $V_R, V_I$. 
	We define $\tilde{D}_V= D_W(\frac12)$. Then we have $\frac{|\langle v',V \rangle|}{|v'||V|}\geq \frac14 \frac{|W|}{|V|}\geq  \frac18 $
	for $v'\in \tilde{D}_V$. Since $W\in \Reals^3$ we have $\operatorname{Vol}(D_W(\frac12))\geq c>0$, so in particular
	\begin{align*}
	U(v,V):= \{v' \in \Reals^3: \frac{|\langle v', V \rangle |}{|v'||V|}\geq \frac{1}{8}\}	
	\end{align*}	
	satisfies $\operatorname{Vol}(U(v,V))\geq c>0$. Since $U(v,V)$ is homogeneous, the set
	$$ U(v,V) \cap \{v' \in \Reals^3: \frac12 \leq |v'| \leq 1 \} \subset D_V(\frac18) $$ 
	also has volume uniformly bounded below, which implies the claim \eqref{geometry1}.	
	For the proof of \eqref{geometry2}, let $v\in \Reals^3$, $|v|\geq 2$ and $V\in \Complex^3$ be a unit vector such that $V=V_1+V_2$, $V_1= P_v V$, $V_2=P_{v}^\perp V$. Let us first assume that $V_2\neq 0$. We claim that \eqref{geometry2} holds with $W=V_2$.
	To this end, let $|v|\geq2$ and $v' \in D_{V_2}(1/8)$, so in particular $|v'|\leq 1$. Then the angle $\psi$ between $v$ and $v-v'$ is bounded by $|\psi|\leq \frac{\pi}{6}$,
	hence:
	\begin{align}
	|P_{v-v'}V_2| 			&= |P_{v-v'} P_v^\perp V |\leq \frac12 |V_2|, \text{ therefore:} \notag\\
	|P_{v-v'}^\perp V|		&= |V_1-P_{v-v'}V_1 + V_2-P_{v-v'}V_2| \geq |V_1-P_{v-v'}V_1 + V_2|- \frac12 |V_2|   \notag\\
	&\geq |P_{V_2}(V_1-P_{v-v'}V_1 + V_2)|- \frac12 |V_2|= |V_2-P_{V_2}P_{v-v'}V_1|- \frac12 |V_2| \label{geomeq1}.
	\end{align}	
	We rewrite the first term on the right-hand side as:
	\begin{align} \label{Pridentity}
	|V_2-P_{V_2}P_{v-v'}V_1| = ||V_2| -\langle \frac{V_2}{|V_2|}, P_{v-v'} V_1\rangle|.
	\end{align}
	Let $\zeta(v')=\langle\frac{V_2}{|V_2|}, P_{v-v'} V_1  \rangle $.
	We observe that $V_2= P_v^\perp V$ and $V_1 = p v$ for some $p \in \Complex$, so:
	\begin{align} \label{zetaref}
	\zeta(v') 	= \langle \frac{V_2}{|V_2|}, \frac{v-v'}{|v-v'|}\rangle \langle \frac{v-v'}{|v-v'|}, V_1 \rangle 
	= \frac{p}{|v-v'|}\langle \frac{V_2}{|V_2|},-v'\rangle   \langle \frac{v-v'}{|v-v'|}, v \rangle  .
	\end{align}
	Since $|v'|\leq \frac12 |v|$, we have $\langle \frac{v-v'}{|v-v'|}, v \rangle\geq \frac12 |v|$. This implies the lower bound:
	\begin{align} \label{geomlowerbd}
	|\zeta(v')|\geq \frac14  \frac{|p v|}{1+|v|} |\langle \frac{V_2}{|V_2|},-v'\rangle|  \geq    \frac{c |V_1|}{1+|v|} \quad \text{for $v'\in D_{V_2}(1/8)$.}	
	\end{align} 
	Now we claim that the real part of $\zeta(v')$ is nonpositive, after possibly changing the sign of $v'$:
	\begin{align} \label{signcondit}
	\Re(\zeta(v')) \leq 0, \text{ or } \Re(\zeta(-v')) \leq 0.
	\end{align}
	To see this, we use \eqref{zetaref} and $\langle \frac{v-v'}{|v-v'|}, v \rangle \geq 0$.
	Inserting the estimates \eqref{geomlowerbd}, \eqref{signcondit} and the lower bound $|z|\geq \frac{1}{\sqrt{2}} \big(|\Re(z)|+|\Im(z)|\big)$ into
	\eqref{Pridentity} we obtain:
	\begin{align*}
	|V_2-P_{V_2}P_{v-v'}V_1| +|V_2-P_{V_2}P_{v-(-v')}V_1| \geq \frac{1}{\sqrt{2}} \left(|V_2|  + \frac{c |V_1|}{1+|v|}\right). 
	\end{align*} 
	We plug this back into \eqref{geomeq1} and add the corresponding term for $-v'$ to prove \eqref{geometry2} in the case $V_2\neq 0$.	
	In order to prove \eqref{geometry2} for $V_2=0$, we remark that the estimate is homogeneous in $V$, so it suffices to prove it for $|V|=1$, when
	it follows by continuity from the case $V_2=0$. 
	
	The estimate \eqref{geometry4} follows from the observation that for $v' \in E(v,V)$ we have
	\begin{align*}
	|P_{v-v'}V| 	&= 	\left|\langle\frac{v-v'}{|v-v'|},V\rangle \right| 	\geq \frac12 |P_v V|.					
	\end{align*}
	Finally \eqref{geometry5} is a consequence of $E(v,V)$ containing either $v'$ or $-v'$ for every $v'\in B_1(0)$.
\end{proof} 

Lemma \ref{geometry} proves lower bounds for the projections $|P_{v-v'}V|$ respectively $|P_{v-v'}^\perp V|$
on a set (of $v'$) with uniformly positive Lebesgue measure. We now show that this implies a lower bound
for the integrals \eqref{KLapl}, \eqref{PLapl} representing $\La(K)$, $\La(P)$.

\begin{lemma} \label{anisotrope}
	Let $z\in \Complex$ with $0\leq \Re(z)\leq 1$ and $\beta$ be a multi-index.
	Let $V,W \in \Complex^3$ be  complex vectors. Further let $n\geq 1$ and $u_0 \in H^n_\lambda$ satisfy the pointwise estimates:
	\begin{align*} 
	c \cf_{|v|\leq 4}(v) \leq u_0(v) \leq C e^{-\frac12 |v|}, \quad \text{for $c>0$}.
	\end{align*}
	Recall $B_1$, $B_2$ as defined in \eqref{quadrC}-\eqref{quadrCu} and $C_1$ defined in \eqref{C1}.
	Then there holds:
	\begin{align}
	\int_{\Reals^3} \langle V, \Re(M_1+M_2)(z,v-v') V \rangle u_0(v') \eta \ud{v'}
	\geq c 	&B_1(z,v)(V,V) \label{matrixest1} \\
	\int_{\Reals^3} |\langle V,(M_1+M_2)(z,v-v')  W \rangle | u_0(v') \eta  \ud{v'} 
	\leq C	&(1+\alpha(z,v)) B_2(z,v)(V,W)
	\label{matrixest2} \\
	\int_{\Reals^3} |\langle V, D^\beta \left((M_1+M_2)(z,v-v')\eta\right) W \rangle | u_0(v')  \ud{v'} 
	\leq C	&\frac{(1+\alpha(z,v))}{(1+|v|)^{|\beta|}} C_1(z,v)|V||W|.
	\label{matrixest3}
	\end{align}
\end{lemma}
\begin{proof}
	First we prove \eqref{matrixest1}. We remark that the integrand is nonnegative:
	\begin{align*}
	\langle V, \Re(M_1) V\rangle 	&= \langle V, \Re\left(\frac{\pi^2}{4 |v|} \frac{1}{1+\frac{z}{|v|}}\right) P_{v}^\perp V \rangle \\
	&= \Re\left(\frac{\pi^2}{4 |v|} \frac{1}{1+\frac{z}{|v|}}\right) |P_{v}^\perp V|^2 \geq 0,
	\end{align*}
	by \eqref{Complex3}. By a similar computation the same is true for $M_2$. 
	We use \eqref{Complex1}  to bound the real part of $M_2$ (cf.\eqref{defM}) below. Using nonnegativity of the integrand, 
	the lower bound on $u_0(v')$ and $\eta(|r|)=1$ for $|r|\geq 1$ we can estimate from below by ($C_2$ as in \eqref{C2}):  
	\begin{align*}
	&\int_{\Reals^3} \langle V, \Re(M_2)(z,v-v') V \rangle  u_0(v') \eta(|v-v'|^2)\ud{v'} 
	\geq c	\int_{B_4(0) \setminus B_1(v)} C_2(z,v-v') |P_{v-v'} V|^2   \ud{v'}.
	\end{align*}
	Now there are $c_1,c_2>0$ s.t. for $|v|\leq 2$  we have $|P_{v-v'} V|\geq c_1 |V|_v$ for all $v'$ in a set $G(v,V)\subset B_4(0) \setminus B_1(v)$  with $|G(v,V)|\geq c_2 $.
	To see this we remark that the inequality is homogeneous in $V$, so we can restrict to $|V|=1$ and $v$ bounded, when the claim follows
	by contradiction. For $|v|\geq 2$ we use \eqref{geometry4}-\eqref{geometry5} to obtain a set of positive measure on which we have
	$|P_{v-v'}V| \geq c |P_v V|$. We find the lower bound:
	\begin{align} \label{B1a}
	&\int_{\Reals^3} \langle V, \Re(M_2)(z,v-v') V \rangle  u_0(v') \eta(|v-v'|^2)\ud{v'}  \geq 	c	C_2(z,v) |P_v V|^2.
	\end{align}
	We apply the same strategy for the term containing $M_1$ (cf. \eqref{defM}):
	\begin{align*}
	&\int_{\Reals^3} \langle V, \Re(M_1)(z,v-v') V \rangle u_0(v') \eta(|v-v'|^2) \ud{v'} 		
	\geq c	C_1(z,v) \int_{B_4(0)\setminus B_1(v)}  |P_{(v-v')}^\perp V|^2  \ud{v'}.
	\end{align*}
	For $|v|\geq 2$ we  use \eqref{geometry1}-\eqref{geometry2}	to obtain:
	\begin{align} \label{B1b}
	&\int_{\Reals^3} \langle V, \Re(M_1)(z,v-v') V\rangle  u_0(v') \eta(|v-v'|^2) \ud{v'} 
	\geq	c	C_1(z,v) |V|_v^2,
	\end{align}
	for $|v|\leq 2$ the same follows again by rescaling $|V|=1$ and contradiction. 
	Combining \eqref{B1a} and \eqref{B1b} we obtain \eqref{matrixest1}. 
	We now show the upper bound \eqref{matrixest2}. The estimates \eqref{Complex1}-\eqref{Complex2} allow to estimate 
	the contribution of $M_2$ (cf. \eqref{defM}) 
	by $C_3$ as defined in \eqref{C3}:
	\begin{align*}
	&\int_{\Reals^3} |\langle V, M_2(z,v-v') W \rangle | u_0(v') \eta(|v-v'|^2) \ud{v'} \\
	\leq 	C	&\int_{\Reals^3} |M_2(z,v-v')| |\langle P_{v-v'}V, P_{v-v'} W \rangle | u_0(v') \eta(|v-v'|^2) \ud{v'} \\
	\leq 	C	&\int_{\Reals^3} (1+\alpha) C_3(z,v-v') (|P_v V|+ \frac{|v'|}{|v|} |P_{v}^\perp V|)(|P_v W|+ \frac{|v'|}{|v|} |P_{v}^\perp W|)  e^{-\frac12 |v'|} \eta  \ud{v'} \\
	\leq 		C& (1+\alpha(z,v)) C_3(z,v)\left(|P_v V| |P_v W|+ |V|_v |W|_v \right) .
	\end{align*}
	Since $C_3(z,v)\leq C C_1(z,v)$ for $0\leq \Re(z)\leq 1$, this shows the contribution of $M_2$ can be estimated by the right-hand side of \eqref{matrixest2}. 
	For bounding the contribution of $M_1$ we proceed similarly, using \eqref{Complex4}:
	\begin{align*}
	&\int_{\Reals^3} |\langle V, M_1(z,v-v') W \rangle | u_0(v') \eta(|v-v'|^2)\ud{v'} \\
	\leq C  	&\int_{\Reals^3} (1+\alpha(z,v-v')) C_1(z,v-v') |P_{(v-v')}^\perp V||P_{(v-v')}^\perp W| e^{-\frac12 |v'|} \eta(|v-v'|^2) \ud{v'}.	
	\end{align*}
	Write $V= P_v V + P_v^\perp V =V_1+V_2$ and $W=W_1+ W_2$ respectively. Then we have
	\begin{align*}
	|P_{v-v'}^\perp V| \leq C \left( \frac{|V_1||v'|}{1+|v|} + |V_2|\right) .
	\end{align*}
	This implies that we can bound:
	\begin{align*}
	&\int_{\Reals^3} |\langle V, M_1(z,v-v') W \rangle| u_0(v') \eta \ud{v'} \\
	\leq C  	&\int_{\Reals^3} (1+\alpha(z,v-v')) C_1(z,v-v') (\frac{|V_1||v'|}{1+|v|} + |V_2|)(\frac{|W_1||v'|}{1+|v|} + |W_2|) e^{-\frac12 |v'|} \eta \ud{v'} \\
	\leq C  	& (1+\alpha(z,v)) C_1(z,v)|V|_v |W|_v 	,	
	\end{align*}
	which concludes the proof of \eqref{matrixest2}. Estimate \eqref{matrixest3} follows from
	a similar computation, using Lemma \ref{derlemma}.
\end{proof}

The following Lemma uses the symmetry of the highest order term in the functionals $Q$ to show it can be expressed by 
the real part of $\La(K)$, $\La(P)$ only, which surprisingly has a sign.

\begin{lemma} \label{realpart}
	Let $n\geq 1$ and $u_0 \in H^n_\lambda$ satisfy the pointwise estimates
	\begin{align}
	c \cf_{|v|\leq 4}(v) \leq u_0(v) \leq C e^{-\frac12|v|}, \quad \text{for $c>0$}.	
	\end{align}
	Furthermore let $\eps>0$, $A>0$ such that $\eps A \leq 1$ and write $z=a+i\omega = \frac{A}{2} +i\omega$.
	Let $u \in V^n_{A,\lambda}$ for some $n\in \Naturals$ and $\gamma \in (0,1]$.
	The  term in $Q^{\alpha,\alpha}_{\eps,A}$ (as defined in \eqref{Qdef}-\eqref{Pterm}), where $|\alpha|\leq n$,  depends on the real part of $\La(K)$ only.
	Writing $V=\nabla D^\alpha \La(u)(z,v)$ we have:
	\begin{equation} \label{RedReal}
	\begin{aligned}
	(2\pi)^{\frac12}Q^{\alpha,\alpha}_{\eps,A}[u_0](u)=	&\int_\Reals \int \langle  {}^\gamma V(z,v) \lambda(v),    \La( K)[u_0]	(\eps z,v)  {}^\gamma V(z,v)\rangle  \ud{v} 	\ud{\omega} \\
	= 			&\int_\Reals \int  \langle {}^\gamma V(z,v) \lambda(v),    \Re(\La( K))[u_0]	(\eps z,v)  {}^\gamma V(z,v)\rangle  \ud{v} 	\ud{\omega}.
	\end{aligned}
	\end{equation}
\end{lemma}
\begin{proof}
	Follows from the observation that the left-hand side is real by Plancherel's Lemma and that $K$ is a symmetric matrix.
\end{proof}

The following lemma amounts to a coercivity result, and shows that for a function $u \in V^n_{A,\lambda}$ the functional $Q^\alpha_{\eps,A}[u_0](u)$ can be controlled by the first $n$ derivatives of $u$ only. 
Here we use that to leading order, the functional is actually dissipative. The exact form of the dissipation $D$ is of particular importance, since we use it later to show that the nonlinearity
can be handled as a perturbation.

\begin{lemma} \label{coercivitylemma}
	Let $n\geq 1$ and $u_0 \in H^n_\lambda$ satisfy the pointwise estimates
	\begin{align}
	c \cf_{|v|\leq 4}(v) \leq u_0(v) \leq C e^{-\frac12|v|}, \quad \text{for $c>0$}.	
	\end{align}
	For $A>0$, let $a= \frac{A}{2}$ and assume $\eps \in  (0,\frac1{a}]$, $\gamma \in (0,1]$ arbitrary and $|\alpha|\leq n$ for an $\alpha \in \Naturals^3$.
	Define the dissipation $D^{\alpha}_{\eps,A}$ as ($z=a+i\omega$):
	\begin{align} \label{dissipexplic}
	D^{\alpha}_{\eps,A}(u):=  \int \int  B_1(\eps z,v) [{}^\gamma \nabla D^\alpha \La (u)(z,v),{}^\gamma \nabla D^\alpha \La (u)(z,v)]  \lambda(v) \ud{v} \ud{\omega} .
	\end{align}	
	Then the leading order quadratic form satisfies the lower bound:
	\begin{align}
	Q^{\alpha,\alpha}_{\eps,A}[u_0](u) \geq c D^{\alpha}_{\eps,A}(u) - C \|u\|_{V^n_{A,\lambda}}^2 					\label{quadrcoerc}. 
	\end{align}
	We will denote by $D^{\alpha}_{\eps,A}$ the dissipation of the equation.
	The lower order terms can be estimated by the dissipation:
	\begin{align}
	\sum_{\beta<\alpha} \binom{\alpha}{\beta}|Q^{\alpha,\beta}_{\eps,A}[u_0](u)|\leq \frac{c}{2} D^\alpha_{\eps,A}(u)+ C \|u\|^2_{V^n_{A,\lambda}}  \label{quadrbili}.
	\end{align}
	The constants can depend on $u_0$ and $n$, but not on $A\geq 1$, $\eps>0$.	
\end{lemma}
\begin{proof}
	In the proof, we drop the dependence on $\gamma$ for shortness.
	We start with proving the lower bound \eqref{quadrcoerc}. As a first step we rewrite $Q_{\eps,A}^{\alpha,\alpha}[u_0](u)$ in terms of Laplace transforms  (write $z=a+i\omega$ for shortness):
	\begin{align}
	Q_{\eps,A}^{\alpha,\alpha}[u_0](u) = 	&\frac1{\eps} \int_0^\infty e^{-At} \int   \nabla (D^\alpha u(t)\lambda)  \left( \int_0^t   K[u_0]					(\frac{t-s}{\eps},v)  \nabla D^\alpha u(s) \ud{s}\right)\ud{v} 	\ud{t} \notag \\
	-	&\frac{1}{\eps} \int_0^\infty e^{-At} \int   \nabla (D^\alpha  u(t) \lambda) \left( \int_0^t  P[u_0](\frac{t-s}{\eps},v)
	D^\alpha u(s) \ud{s}\right)\ud{v} \ud{t}\notag	\\
	=	&(2\pi)^{-\frac12}\int_\Reals  \int   \langle \nabla (D^\alpha \La(u)(z,v))\lambda),   \La( K)[u_0]	(\eps z,v)  \nabla D^\alpha \La(u)(z,v) \rangle \ud{v} 	\ud{\omega} \notag\\
	-	& (2\pi)^{-\frac12} \int_\Reals \int  \langle \nabla (D^\alpha \La(u)(z,v))\lambda),    \La(P)[u_0]	(\eps z,v)  D^\alpha \La(u)(z,v)\rangle \ud{v} \ud{\omega} \notag \\
	=	&J_1+ J_2 \label{J1J2}.			
	\end{align}
	We recall the representation of $\La(K)$ given in Lemma \ref{Laplacerep}:
	\begin{align} \label{LaKrep}
	\La(K[u])(z,v) 	&= \int (M_1+M_2)(z,v-v') u(v')\eta(|v-v'|^2)  \ud{v'}.  
	\end{align}
	We start by estimating $J_1$. For shortness, we write $V=\nabla D^\alpha \La(u)$. 	Then use \eqref{LaKrep}, Lemma~\ref{realpart} and  the pointwise estimates proven in Lemma \ref{anisotrope} :
	\begin{equation} \label{I3}
	\begin{aligned}
	J_1	=	&(2\pi)^{-\frac12}\int_\Reals  \int  \langle V(z,v)\lambda(v),    \La( K)[u_0]	(\eps z,v)  V(z,v) \rangle \ud{v} 	\ud{\omega}  \\
	+	&(2\pi)^{-\frac12}\int_\Reals  \int  \langle D^\alpha \La(u)(z,v) \nabla (\lambda(v)), \La( K)[u_0](\eps z,v)  V(z,v) \rangle \ud{v} 	\ud{\omega} \\
	\geq	&c D^\alpha_{\eps,A}(u)	+	(2\pi)^{-\frac12}\int_\Reals  \int  \langle D^\alpha \La(u)(z,v) \nabla (\lambda(v)),    \La(K)[u_0]	(\eps z,v)  V(z,v)\rangle \ud{v} 	\ud{\omega} \\
	= &c D^\alpha_{\eps,A}(u)	+ I_3 .
	\end{aligned}
	\end{equation}
	It remains to estimate $J_2$ given by \eqref{J1J2} and  $I_3$ given by \eqref{I3}. To this end, we recall the definition
	of $\|\cdot \|_{V^n_{\eps,A}}$ in \eqref{Vndef} and use the Plancherel identity in Lemma \ref{Plancherel} to estimate:
	\begin{align}
	\int_{\Reals} \int_{\Reals^3} |D^\alpha \La(u)(z,v)|^2 \lambda(v) \ud{\omega} \ud{v} \leq C \|u\|^2_{V^n_{A,\lambda}}. \label{planch2}
	\end{align}
	In order to estimate $I_3$, we observe that $\nabla \lambda = P_v \nabla \lambda$. Then we combine \eqref{LaKrep} with \eqref{matrixest2} in Lemma \ref{anisotrope} to obtain the estimate (recall $B_2$, cf. \eqref{quadrCu}):
	\begin{align*}
	|I_3| 	\leq 	& C\int_\Reals  \int   |D^\alpha \La(u)| \lambda   C(1+\alpha(\eps z,v)) B_2(\eps z,v)[ P_v \nabla \lambda(v),V]  \ud{v} 	\ud{\omega} \\
	\leq 	& C\int_\Reals  \int   \left(|D^\alpha \La(u)| \right) \lambda    \left(  \frac{|V(z,v)|}{(1+\alpha(\eps z,v))(1+|v|^2)} + \frac{(\beta + \alpha + \alpha^2)|P_v V(z,v)|}{(1+\alpha)^ 3(1+|v|)}\right)   \ud{v} 	\ud{\omega}.
	\end{align*} 
	We apply Young's inequality and \eqref{planch2} to get the bound ($D^\alpha_{\eps,A}$ defined in \eqref{dissipexplic}):
	\begin{align} \label{I3est}
	|I_3| 	\leq  	&\frac{c}{4} D^\alpha_{\eps,A} + C \|D^\alpha u\|^2_{V^n_{A,\lambda}}.
	\end{align}
	It remains to estimate $J_2$ to finish the proof of \eqref{quadrcoerc}. We recall that $P[u_0]= \nabla \cdot K[u_0]$.  We apply \eqref{matrixest3} with $|\beta|=1$ 
	and recall the definition of $C_1$ (cf. \eqref{C1}) to obtain an upper estimate for $J_2$:
	\begin{align*}
	|J_2| \leq 	& C\int_\Reals \int \left(\lambda^\frac12(v) \frac{1+\alpha(\eps z,v)}{1+|v|} C_1(\eps z,v) |V|\right)  \left(\lambda^\frac12(v) |D^\alpha \La(u)(z,v)|\right)   \ud{v} \ud{\omega} .
	\end{align*}
	Notice that \eqref{matrixest3} provides $\frac{1}{|v|}$ more decay than naively expected, which is essential here. 
	Young's inequality in combination with \eqref{planch2} implies:
	\begin{align} \label{J2est}
	|J_2|	\leq   	& \frac{c}{4} D^\alpha_{\eps,A}(u) + C\|u\|^2_{V^n_{A,\lambda}}.
	\end{align} 
	Combining the estimates \eqref{J1J2}, \eqref{I3est} and \eqref{J2est} proves \eqref{quadrcoerc}.
	In the case $\beta<\alpha$ we use \eqref{matrixest3} in Lemma \ref{anisotrope} and Young's inequality to prove \eqref{quadrbili}.
\end{proof}

The linear result follows as a corollary. The statement can be generalized significantly,
the assumptions in our a priori estimates are designed for the nonlinear case and therefore more restrictive than needed
for the linear equation. 
\begin{theorem}
	Let $n\geq 6$ and $u_0 \in H^n_\lambda$ satisfy the pointwise estimate
	\begin{align}
	c \cf_{|v|\leq 4}(v) \leq u_0(v) \leq C e^{-\frac12|v|}, \quad \text{$c,C>0$.}	
	\end{align}
	There exists $A>0$ s.t. for $\eps>0$ small, there is a solution $u_\eps \in V^n_{A,\lambda} \cap  C^1(\Reals^+;H^{n-2}_\lambda)$ to:
	\begin{equation} \label{epseq} 
	\begin{aligned}
	\partial_t u_\eps 	= 	&\frac1\eps \nabla \cdot \left(\int_0^t K[u_0]\left(\frac{t-s}\eps,v\right) \nabla u_\eps(s,v) \ud{s}\right)\\
	-  		&\frac1\eps \nabla \cdot \left(\int_0^t P[u_0]\left(\frac{t-s}\eps,v\right) u_\eps(s,v)  \ud{s}\right)  \\ 
	u_\eps(0,\cdot)		&= u_0(\cdot) .	
	\end{aligned}
	\end{equation}
	There is a function $u \in V^{n}_{A,\lambda}\cap C^1(\Reals^+;H^{n-4}_\lambda)$ s.t.  
	$u_{\eps_j}  \rightharpoonup u $ in $V^{n}_{A,\lambda}$ along a sequence $\eps_j \rightarrow 0$. 
	The function $u$ solves the limit equation ($\K$, $\Pe$ defined in \eqref{limiteq}):
	\begin{equation} \label{linlandrefer}
	\begin{aligned}
	\partial_t 	u 		&= \nabla \cdot \left( \K[u_0] \nabla u\right) - \nabla \cdot \left( \Pe[u_0]  u\right) \\
	u(0,v)		&= u_0(v). 
	\end{aligned}
	\end{equation}	
\end{theorem}
\begin{proof}
	For  $0<\gamma\leq 1$, the existence of solutions $u_{\eps,\gamma}$ to \eqref{gammaepseq} follows from Lemma~\ref{lemexistence}. 
	In order to prove well-posedness for \eqref{epseq}, i.e. $\gamma=0$, we derive a priori estimates that are uniform in $\gamma$.
	Combining Lemma \ref{lemReduc} and Lemma \ref{coercivitylemma} shows that for $A>0$ large enough
	\begin{align} \label{est:uepsgamma}
	\|u_{\eps,\gamma}\|_ {V^n_{A,\lambda} }\leq C
	\end{align}
	are uniformly bounded in $0<\gamma,\eps\leq \frac{1}{A}$. 
	Now we use the Laplace representation in Lemma~\ref{Laplacerep} to infer the uniform boundedness:
	\begin{align} \label{bmultiplier}
	|\nabla^m \La(K[u_0])(z,v)|+|\nabla^m \La(P_\gamma[u_0])(z,v)|\leq C(m) \quad \text{for $m \in \Naturals$}.	
	\end{align}
	We rewrite \eqref{gammaepseq} in Laplace variables and obtain:
	\begin{align} \label{gammaepslaplace}
	z \La(u_{\eps,\gamma}) = {}^\gamma\nabla \cdot \left(\La(K[u_0])(\eps z) {}^\gamma \nabla \La(u_{\eps,\gamma}) - \La(P_{\gamma}[u_0])(\eps z) \La (u_{\eps,\gamma}) \right)+ u_0(v).
	\end{align}
	The right-hand side of \eqref{gammaepslaplace}
	is bounded in $V^{n-2}_{A,\lambda}$ due to \eqref{bmultiplier} and \eqref{est:uepsgamma},  so  we get a bound of:
	\begin{align} \label{lintotest}
	\|u_{\eps,\gamma}\|_ {V^n_{A,\lambda} }+\|\partial_t u_{\eps,\gamma}\|_{V^{n-2}_{A,\lambda}}	\leq C.
	\end{align}  
	By the Rellich type Lemma \ref{Rellich}, and the fact that $V^n_{A,\lambda}$ is a separable Hilbert space, 
	there is a $u_\eps \in V^n_{A,\lambda}$ and a sequence $\gamma_j\rightarrow 0$ s.t. $u_{\eps,\gamma_j} \rightharpoonup u_\eps $ in $V^n_{A,\lambda}$
	and $u_{\eps,\gamma_j} \rightarrow u_\eps $ in $V^{n-3}_{A,\tilde{\lambda}}$.
	We need to show that the weak limit $u_\eps$ indeed solves the equation \eqref{epseq}.  
	Both sides of \eqref{gammaepslaplace} converge pointwise a.e. to the respective sides with $\gamma=0$ along a subsequence of $\gamma_j \rightarrow 0$. 
	Since the Laplace transform defines the function uniquely, $u_\eps$ is indeed a solution.
	Finally, the solutions $u_\eps$ are in $C^1(\Reals^+;H^{n-2}_\lambda)$ since they are bounded in $V^n_{A,\lambda}$ and the equation
	\eqref{epseq} in combination with $|\nabla^m K[u_0]|+ |\nabla^m P[u_0]|\leq C(m)$ allows to control the time derivative in $C^0(\Reals^+;H^{n-2}_{\lambda})$. 
	
	The convergence of $u_\eps$ to a solution $u$ of \eqref{linlandrefer} follows similarly. 
	We use the uniform bound \eqref{lintotest} to find a subsequence $\eps_j \rightarrow 0$
	and $u \in V^n_{A,\lambda}$ such that $u_{\eps_j} \rightharpoonup u$ in $V^n_{A,\lambda}$
	and $u_{\eps_j} \rightarrow u$ in $V^{n-3}_{A,\tilde{\lambda}}$.
	Now the claim follows from the observation that for $\gamma=0$ we can take the limits on both sides of \eqref{gammaepslaplace} and pointwise a.e. along a subsequence there holds:
	\begin{align*}
	\La(u_{\eps_j}) \rightarrow \La(u),\,  \La(K)[u_0](\eps_j z,v) \rightarrow \K[u_0](v),\,  \La(P)[u_0](\eps_j z,v) \rightarrow \Pe[u_0](v).	
	\end{align*}
	Repeating the argument above, we find that the weak limit $u_{\eps_j} \rightharpoonup u \in V^n_{A,\lambda}$ is actually  $u \in V^{n}_{A,\lambda}\cap C^1(\Reals^+;H^{n-4}_\lambda)$  and  is indeed a solution of the equation \eqref{linlandrefer}.
\end{proof}

\section{A priori estimate for the nonlinear problem} \label{Sec:Freezing}

\subsection{Continuity of the fixed point mapping $\Psi$} \label{Psisubsection}
In this subsection we prove that solutions of equation \eqref{cutfixedeq} satisfy an a priori estimate, for
small perturbations $f_\eps$. Here smallness is measured in terms of
the size and decay of the Laplace transform, i.e. the smoothness of the perturbation $f_\eps$.
The necessary framework is provided by the definition below. Notice
that we always assume that $f_\eps= \nabla \cdot g_\eps$ is a divergence, so it has zero average.
This is the key point to obtain an additional decay $\frac{1}{|v|}$ in Lemma \ref{N1lemma}.
Furthermore it is essential that the highest order term $Q^{\alpha,\alpha}_{\eps,\delta}[f_\eps](u)$
introduced in \eqref{Qdef} is a symmetric integral, which induces a cancellation for large Laplace frequencies. 
In the subsequent subsection we will prove that our smallness assumption is consistent, i.e. if the condition is 
satisfied by $f_\eps$, then it is also satisfied by $u_\eps-u_0$ when $u_\eps$ solves \eqref{cutfixedeq}.  
\begin{definition} \label{bdryass}
	We define a sequence of cutoff functions $\kappa_{\delta_1}\in C^\infty_c (\Reals)$ by 
	\begin{align}
	\kappa_{\delta_1}(s) := \kappa\big(\frac{s}{\delta_1}\big), \label{defkappa}
	\end{align}
	where $\kappa\in C^\infty_c(\Reals)$, $0\leq \kappa\leq 1$, $\kappa(s)=1$ for $|s|\leq 1$ and
	$\kappa(s)=0$ for $|s|\geq 2$.
	Let $R,\eps,\delta>0$ and $z\in \Complex$. We define $Y_{R,\eps,\delta} (z)$ by
	\begin{align}
	Y_{R,\eps,\delta} (z) := \frac{\delta}{1+|z|^2} + \frac{R \eps |z|}{(1+\eps |z|)(1+ |z|^2)} \label{Ydef}.	
	\end{align}			
	We will consider $u=(f,g) \in X^n_{A,\tilde{\lambda}}$ (defined in \eqref{Xdef}), s.t. a.e. on the line $\Re(z)=\frac{A}{2}=a>0$:
	\begin{align}
	|\La(f)(z,v)| 	&\leq Y_{R,\eps,\delta}(z) e^{-\frac12 |v|}, 				&|\La(g)(z,v)| 	&\leq Y_{R,\eps,\delta}(z) e^{-\frac12 |v|} \label{gaprio1} \\
	|\La(f)(z,v)| 	&\leq \frac{R e^{-\frac12 |v|}}{|1+\eps z| (1+|z|^2)}, 	&|\La(g)(z,v)| 	&\leq \frac{R e^{-\frac12 |v|}}{|1+\eps z| (1+|z|^2)} \label{gaprio2} \\
	|\partial_t f(t,v)| &\leq R e^{-\frac12 |v|} \label{dercond} . 
	\end{align}
	For $R,\delta,\eps>0$, $A\geq 1$, $a=\frac{A}{2}$ and $n \in \Naturals$, let 
	$\Omega^n_{A,R,\delta,\eps}\subset X^n_{A,\tilde{\lambda}}$ be the set of functions given by:
	\begin{align}
	\Omega^n_{A,R,\delta,\eps} = \{u=(f,g) \in X^n_{A,\tilde{\lambda}}: \|u\|_{X^n_{A,\tilde{\lambda}}} \leq R, \text{\eqref{dercond} 
		and \eqref{gaprio1}-\eqref{gaprio2} for $\Re(z)=a$} \} \label{omegadef}.	
	\end{align}  
\end{definition}
Since the estimates \eqref{gaprio1}-\eqref{gaprio2} are stable under convex combinations of functions,
we have:
\begin{lemma} \label{closedconvex}
	For all $R,\delta,\eps>0$, $A\geq 1$ and $n \in \Naturals$, the set $\Omega^n_{A,R,\delta,\eps}$ is a nonempty, bounded, closed and convex subset
	of $X^n_{A,\tilde{\lambda}}$. 
\end{lemma}
The following theorem is the main result of this subsection, giving an a priori estimate for the solution operator to
\eqref{cutfixedeq} under the smallness assumption $(f,g)\in \Omega^n_{A,R,\delta,\eps}$ for small $\eps,\delta$. 
We prove the error term can be controlled by the dissipation $D^\alpha_{\eps,A}$ (cf. \eqref{dissipexplic}) provided by the linear equation. Observe that existence of (unique) global solutions of \eqref{cutfixedeq} has been proved in Lemma \ref{lemexistence}. Here we will prove 
a priori estimates that are uniform in the mollifying parameter $\gamma>0$ and $\eps>0$. 
\begin{theorem} \label{thmfreez} 
	Let $n \in \Naturals$, $n\geq 2$.  Assume $u_0\in H^n_{\lambda}$ satisfies:
	\begin{align*}
	c \cf_{|v|\leq 4}(v) \leq u_0(v) \leq C e^{-\frac12 |v|}.
	\end{align*}	
	Then there exist $A,\delta>0$ such that for all $R>0$ there is an $\eps_0>0$ with the property that the operator 					$\psi_{\delta_1}$ given 			by:
	\begin{equation} \label{contthmeq}
	\begin{aligned}
	\Psi_{\delta_1}: 	\Omega^n_{A,R,\delta,\eps} 	&\longrightarrow X^n_{A,\lambda} \\
	(f,g)						&\mapsto \left((u-u_0)k_{\delta_1}, \A^{0,0}_\gamma[f](u)k_{\delta_1}	\right)\text{, $\A^{0,0}_\gamma$ as in \eqref{Adef} and $u$ solution to:}\\	 
	\partial_t u 	= 	&\frac1\eps {}^\gamma \nabla \cdot \left(\int_0^t K[u_0+f(s)]\left(\frac{t-										s}\eps,v\right) {}^\gamma \nabla u(s,v) \ud{s}\right) \\
	-  	&\frac1\eps {}^\gamma \nabla \cdot \left(\int_0^t P_\gamma[u_0+f(s)]\left(\frac{t-s}\eps,v\right) 										u(s,v)  \ud{s}\right)  \\ 
	u(0,\cdot)		&= u_0(\cdot),
	\end{aligned}
	\end{equation}
	is well-defined and continuous (w.r.t. the topologies of $X^n_{A,\tilde{\lambda}}$, $X^n_{A,\lambda}$) for all $\gamma, \delta_1 \in (0,\frac12]$ and 
	$\eps \in (0,\eps_0)$. Furthermore, the solutions satisfy the following estimate:
	\begin{align} \label{PsiEst}
	\|\Psi_{\delta_1} (f,g)\|_{X^{n}_{A,\lambda}} + \|\partial_t \Psi_{\delta_1} (f,g)\|_{X^{n-2}_{A,\lambda}} \leq C(A,\delta_1).
	\end{align}
\end{theorem}

Notice that the operator $\psi_{\delta_1}$ maps functions in $X^n_{A,\tilde{\lambda}}$ to functions in $X^n_{A,\lambda}$, thus yields better decay.
As can be seen from Lemma \ref{lemReduc} this follows from the fast decay of the initial datum, provided we can control the quadratic terms $Q$.
In Section \ref{Sec:Linear} we have shown that the quadratic functionals $Q[u_0]$ defined in \eqref{Qdef} satisfy a coercivity estimate. In this subsection we will prove smallness for the perturbation $Q[f]$, so the sum $Q[u_0+f]$ still has a sign. To this end we first include an auxiliary Lemma to represent 
those functionals in Laplace variables.
\begin{lemma} \label{freezrep}
	The quadratic functionals $Q_{\eps,A}^{\alpha,\beta}[\nu](u)$ defined in \eqref{Qdef} can be represented by means of the
	Laplace transform of $u$ as:
	\begin{align*}
	(2\pi)^\frac12	Q^{\alpha,\beta}_{\eps,A}[\nu](u)	
	= 	&\int_{\Reals} \int_\Reals \int \langle \nabla (D^\alpha \La(u) \lambda)(z), D^{\alpha-\beta}\Lambda[\nu](\eps z, \omega-\theta) \La(\nabla D^\beta u)(p) 	\rangle \ud{v}\ud{\theta} \ud{\omega} \\
	-	&\int_{\Reals} \int_\Reals \int \langle \nabla (D^\alpha  \La(u) \lambda)(z), \nabla D^{\alpha-\beta} \Lambda[\nu](\eps z, \omega-\theta) \La(D^\beta u)(p) \rangle 					\ud{v}\ud{\theta} \ud{\omega}. 
	\end{align*} 
	We use the short notation $z=a+i\omega$, $p=a+i\theta$ and $\Lambda$ is given by $M_1,M_2$ (cf. \eqref{defM}) as:
	\begin{align}
	\Lambda[\nu](z,\tau,v)	&= \int_{\Reals^3} \int_{\Reals}  \left(M_1+M_2\right)(z,v-v') e^{-i\tau s}\eta(|v-v'|^2)\nu(s,v') \ud{s} \ud{v'} \label{Lambdadef} .	
	\end{align}
\end{lemma}
\begin{proof}
	Follows directly from the elementary properties of the Laplace Transform.
\end{proof}
Exploiting the symmetry properties of the functional $Q^{\alpha,\alpha}_{\eps,A}[f_\eps]$ 
is essential to proving that this term is small compared to the dissipation $D^\alpha_{\eps,A}$ (cf. \eqref{dissipexplic}).
For better notation we first include some definitions. 
\begin{definition} \label{defkernels}
	For $\eps>0$, $v\in \Reals^3$, $z=a+i\omega, p=a+i\theta \in \Complex$, define the matrices $L_1$, $L_2$:
	\begin{align}
	L_1(\eps,z,p,v) &:=  	\frac12 ( M_1(\eps z, v) + M_1(\eps \ol{p}, v))  \label{L1def}\\
	L_2(\eps,z,p,v) &:= 		\frac12 ( M_2(\eps z, v) + M_2(\eps \ol{p}, v)) \label{L2def}
	\end{align}
	and the associated symmetrized kernel $\Lambda_s$ by:
	\begin{align}
	\Lambda_s[\nu ](\eps,z,p,v)&:= \Lambda_1[\nu ](\eps,z,p,v) + \Lambda_2[\nu ](\eps,z,p,v)  \label{Lambdas}  \\
	\Lambda_1[\nu ](\eps,z,p,v)&:=  \int_{\Reals^3} L_1(\eps,z,p,v-v') \left(\int_0^\infty e^{-is(\omega-\theta)}\nu(s,v') \ud{s}\right) \eta(|v-v'|^2) \ud{v'}  \notag\\
	\Lambda_2[\nu ](\eps,z,p,v)&:=  \int_{\Reals^3} L_2(\eps,z,p,v-v') \left(\int_0^\infty e^{-is(\omega-\theta)}\nu(s,v') \ud{s}\right) \eta(|v-v'|^2) \ud{v'} \notag .	
	\end{align}
	We split the kernel $L_2$ further into:
	\begin{align}
	N_2(\eps,z,p,v) 			&= L_2 (\eps,z,p,v)-N_1(\eps,z,p,v), \text{where} \label{N2}  \\
	N_1(\eps,z,p,v)								&= \frac{1}{|v|^2}\frac{\eps (a+i(\theta- \omega))}{(1+\frac{\eps z}{|			v|})^2(1+\frac{\eps \ol{p}}{|v|})^2} P_{v} 					\label{N1}.	
	\end{align} 	
\end{definition}

\begin{lemma} \label{matrixsymm}
	Let $a>0$ and $z=a+i\omega$, $p= a+ i\theta$. Further let $\eps \leq \frac{1}{a}$. For $V,W \in \Complex^3$ and $L_1$, $N_1$
	as in the definition above, and $|v|\geq c>0$, we have the estimates:
	\begin{align}	
	|\langle V,	L_1(\eps,z,p,v)W \rangle|	&\leq C 	\frac{|P_{v}^\perp V| |P_{v}^\perp W|}{1+|v|} \frac{1+ \eps |	\theta-\omega|}{(1+\alpha(\eps z,v))(1+ 						\alpha(\eps p,v))}  						\label{matrixsymm1}  \\ 									
	| \langle V, N_2(\eps,z,p,v)  W \rangle |  	&\leq C 	\frac{|V||W|}{1+|v|^3} \frac{\eps^2 |p| |z| +\eps^2|p| |z| 				(1+\eps|\theta-\omega|)}						{(1+\alpha(\eps z,v))^2(1+ \alpha(\eps p,v))^2}  \label{matrixsymm2}.
	\end{align}
\end{lemma}
\begin{proof}
	We start by proving \eqref{matrixsymm1}. Using $\eps \leq  \frac{1}{a}$, $|v|\geq c>0$ and the definition of $L_1$ (cf. \eqref{L1def}) and $M_1$ (cf. \eqref{defM})  we can bound:
	\begin{align*}
	|\langle V, L_1(\eps,z,p,v) W \rangle| 	\leq 	& \frac{C|P_{v}^\perp V| |P_{v}^\perp W|}{|v|}\left|\frac{1}{1+\eps z/|v|}+ \frac{1}{1+\eps \ol{p}/|v|} \right|  \\
	\leq 	& \frac{C|P_{v}^\perp V| |P_{v}^\perp W|}{|v|}  \left| \frac{1  + \eps |\theta-\omega|}{(1+\alpha(\eps z,v))(1+\alpha(\eps p,v))} \right|.
	\end{align*}
	The decomposition of $L_2$ (defined in \eqref{L2def}) follows from the identity:
	\begin{align} \label{N1N2dec}
	\frac{b}{(1+b)^2} + \frac{\ol{c}}{(1+\ol{c})^2} &= \frac{b+\ol{c}}{(1+b)^2 (1+\ol{c})^2} + 					\left(\frac{4 b \ol{c}}{(1+b)^2 (1+\ol{c})^2}+\frac{b\ol{c}(b+\ol{c})}{(1+b)^2 (1+\ol{c})^2}\right).
	\end{align}
	We insert $b=\frac{\eps z}{|v|}$, $c=\frac{\eps p}{|v|}$ and multiply \eqref{N1N2dec} with $\frac{\pi^2 P_v}{4|v|}$. 
	Then the first term on the right gives $N_1$, so the second gives $N_2$ as defined in \eqref{N2}. The latter is bounded by:
	\begin{align*}
	|N_2|	&\leq \frac{\pi^2 P_v}{4|v|}	\left(\frac{4 b \ol{c}}{(1+b)^2 (1+\ol{c})^2}+\frac{b\ol{c}(b+\ol{c})}{(1+b)^2 (1+\ol{c})^2}\right) \\ 	
	&\leq \frac{C}{|v|^3} 	\left| \frac{\eps^2 |p| |z|+ \eps^3 |p| |z| (a+|\theta-\omega|)}				{(1+\alpha(\eps z,v))^2(1+\alpha(\eps p,v))^2} \right| 
	\leq  \frac{C}{|v|^3} \frac{\eps^2 |p| |z| + \eps^2|p| |z| (1+\eps|\theta-\omega|)}{(1+\alpha(\eps 			z,v))^2(1+ \alpha(\eps p,v))^2}. 
	\end{align*}
	This proves estimate \eqref{matrixsymm2}.
\end{proof}
Our goal is to prove estimates for the functional $Q^{\alpha,\alpha}_{\eps,A}[f_\eps]$. This will be done
estimating $\Lambda_s$ as defined in \eqref{Lambdas}, which is given by $L_1$, $L_2$ (cf. \eqref{L1def}, \eqref{L2def}). 
We have decomposed $L_1+L_2 = L_1+N_1+N_2$, and Lemma \ref{matrixsymm} gives estimates for $L_1$ and $N_2$. It remains
to prove an estimate for $N_1$. Here we rely on the additional decay provided by the divergence 
property $f= \nabla \cdot g$ of functions in $\Omega$. Under the divergence assumption we get the following Lemma. 
\begin{lemma} \label{N1lemma}
	Let $N_1$ be given by \eqref{N1}. Let $h= \nabla \cdot G$, where $G\in H^1_{\tilde{\lambda}}$, $|G(v)|\leq R_1 e^{-\frac12 |v|}$. 
	For $a>0$, $\eps \in (0,\frac1{a}]$, $z=a+i\omega ,p= a+i\theta \in \Complex$ we have:
	\begin{equation} \label{N1est}
	\begin{aligned}
	\left| \int \langle V, N_1(\eps,z,p,v-v')W \rangle h(v') \eta(|v-v'|^2) \ud{v'}\right| 		
	\leq  	\frac{C R_1 |V||W|(1+\eps |\omega-\theta|)}{(1+|v|^3) (1+\alpha(\eps z,v))^2(1+ \alpha(\eps p,v))^2}. 
	\end{aligned}	
	\end{equation}
\end{lemma}
\begin{proof}
	We simply use that $h= \nabla \cdot G$ is a divergence and write:
	\begin{align} \label{N1convest}
	\int_{\Reals^3} N_1(\eps,z,p,v-v') \eta h(v')  \ud{v'} 	
	= 		-\int_{\Reals^3} \nabla_{v'} \left(N_1(\eps,z,p,v-v') \eta(|v-v'|^2) \right) G(v')	\ud{v'}.
	\end{align}
	Explicitly computing the derivative of $N_1$ as defined in \eqref{N1} gives:
	\begin{align*}
	|\nabla_{v} \left(N_1(\eps,z,p,v) \eta(|v|^2)\right)| \leq C\frac{1 + \eps |\theta -\omega|}{(1+|v|^3)(1+\alpha(\eps z,v))^2(1+ \alpha(\eps p,v))^2} .
	\end{align*}
	Now plugging the assumption $|G(v)|\leq R_1 e^{-\frac12|v|}$ into \eqref{N1convest} gives the claim.
\end{proof}
\begin{lemma} \label{Ystable}
	For $A>0, n \in \Naturals$, $n\geq 2$, $R,\delta,\eps>0$ and all $(f,g) \in \Omega^n_{A,R,\delta,\eps}$ 
	we have: 
	\begin{equation} \label{lapminest}
	\begin{aligned}
	\left|\int_0^\infty e^{-is\tau }f (s,v) \ud{s}\right| &\leq C(A)  \min\{Y_{R,\eps,\delta}(\tau),\frac{R }{(1+\eps |\tau|) (1+|\tau|^2)}\} e^{-\frac12|v|}\\
	\left|\int_0^\infty e^{-is\tau }g (s,v) \ud{s}\right| &\leq C(A) \min\{Y_{R,\eps,\delta}(\tau),\frac{R }{(1+\eps |\tau|) (1+|\tau|^2)}\} e^{-\frac12|v|},
	\end{aligned}
	\end{equation}
	for $\tau \in \Reals$. Here $Y_{R,\eps,\delta}(\tau)$ is the function defined in \eqref{Ydef}.
\end{lemma}
\begin{proof}
	By definition of $\Omega^n_{A,R,\delta,\eps}$ (see \eqref{omegadef}) for $\Re(z)=a$ there holds:
	\begin{equation} \label{lapminest1}
	\begin{aligned}
	|\La(f)(z,v)| &\leq \min\{Y_{R,\eps,\delta}(z),\frac{R }{(1+\eps |z|) (1+|z|^2)}\} e^{-\frac12|v|} \\
	|\La(g)(z,v)| &\leq \min\{Y_{R,\eps,\delta}(z),\frac{R }{(1+\eps |z|) (1+|z|^2)}\} e^{-\frac12|v|}.
	\end{aligned}
	\end{equation}
	Notice that the estimate is the same for $f$ and $g$. We rewrite the left-hand side  of \eqref{lapminest} as:
	\begin{align*}
	\left|\int_0^\infty e^{-is\tau }f (s,v) \ud{s}\right| &= \left|\int_0^\infty e^{-is\tau }e^{-as}
	f(s,v) \kappa_{2}(s)e^{as} \ud{s} \right|\\
	&= \frac{1}{2\pi} \left| \int_\Reals \La(f)(a+i(\tau -\omega))  F(\tau-\omega)  											\ud{\omega} \right|,
	\end{align*}
	where $F(\omega)=  \int_\Reals  e^{-is\omega  }\kappa_{2}(|s|)e^{as} \ud{s}  $. 
	The function $F$ is the Fourier	transformation of a fixed Schwartz function, hence decays faster
	than any polynomial. For the rational function $Y_{R,\eps,\delta}$ defined in \eqref{Ydef}, a straightforward computation shows $|Y_{R,\eps,\delta}* F|\leq C |Y_{R,\eps,\delta}|$ with  $C>0$ independent of $\eps \in (0,\frac{1}{a}]$ and $R,\delta>0$.
\end{proof}

\begin{lemma} \label{symmkernel}
	Let $A\geq 1$, $n\geq 2$, $\alpha$ a multi-index with $|\alpha|\leq n$ and $c>0$ arbitrary be given. 
	There exists $\delta_0(c,A,n)>0$ such that for all $\delta \in (0,\delta_0]$ and $R>0$, 
	we can estimate:
	\begin{align}
	|Q^\alpha_{\eps,A}[f](u)| \leq  \sum_{\beta\leq \alpha} \binom{\alpha}{\beta} |Q^{\alpha,\beta}_{\eps,A}[f](u)|	&\leq  c D^{\alpha}_{\eps,A}(u) +  \|u\|_{V^n_{A,\lambda}}^2, \label{freethm}
	\end{align}
	for all $(f,g) \in \Omega^n_{A,R,\delta,\eps}$, when $0<\eps\leq \eps_0(\delta,R,A,c,n)$ is small.	
\end{lemma} 
\begin{proof}
	Fix $A\geq 1$ and $n \in \Naturals$, $n\geq 2$ and $c>0$ as in the assumption.
	We first estimate the highest order term $\beta=\alpha$ in the quadratic form $Q$. 
	We start our estimate from the representation in Lemma \ref{freezrep} (we write $\nabla={}^\gamma \nabla$ for shortness):
	\begin{align}
	(2\pi)^\frac12 Q^{\alpha,\alpha}_{\eps,A}[\nu](u)	
	=	&\int_{\Reals} \int_\Reals \int \langle \lambda   \La( \nabla D^\alpha u) (z), \Lambda(\eps z, 				\omega-\theta) \La(\nabla D^\alpha u)(p)\rangle \ud{v}\ud{\theta} \ud{\omega} \label{est3}\\
	+	&	\int_{\Reals} \int_\Reals \int \langle  \La(D^\alpha u) (z) \nabla (\lambda), \Lambda(\eps z, 			\omega-\theta) \La(\nabla D^\alpha u)(p)  \rangle \ud{v}\ud{\theta} 	\ud{\omega} \label{est4} \\
	-	&\int_{\Reals} \int_\Reals \int \langle  \nabla(\La( D^\alpha u) \lambda)(z), \nabla \Lambda(\eps 			z, \omega-\theta) \La(D^\alpha u)(p) \rangle \ud{v}\ud{\theta} 	\ud{\omega} \label{est5} \\
	=	& J_1 + I_3 +J_2 \notag.
	\end{align}	
	We start with estimating the critical term $J_1$. We can symmetrize in $p,z$, and replace $\Lambda$ by $\Lambda_s$ 
	as introduced in Definition \ref{defkernels}. The symmetrization gives (for shortness write $V=\La(\nabla D^\alpha u )$):
	\begin{equation}	 \label{I1I2def}
	\begin{aligned} 
	J_1	=	&\int_{\Reals} \int_\Reals \int \lambda \langle V(z,v), \Lambda(\eps z, \omega-\theta,v) V (p,v) 				\rangle \ud{v}\ud{\theta} 			\ud{\omega}  \\
	=			&\int_{\Reals} \int_\Reals \int \lambda \langle V(z,v), (\Lambda_1+\Lambda_2) V(p,v) 									\rangle 				\ud{v}\ud{\theta}		\ud{\omega}  =  I_1 +I_2.
	\end{aligned}
	\end{equation}
	We estimate $I_1$ using the estimate on $L_1$ in \eqref{matrixsymm1} and use Lemma \ref{Ystable} to bound $\La(f)$:
	\begin{align} 
	|I_1| \leq 	&C \int_{\Reals} \int_\Reals \int \int \lambda 	\frac{|P_{(v-v')}^\perp V(z,v)||P_{(v-v')}^\perp V(p,v)|}{|v-v'|} \frac{(1+ \eps |\theta-\omega|)|\La(f)(i(\theta-\omega),v')|\eta}{(1+\alpha(\eps z,v-v'))(1+ \alpha(\eps p,v-v'))} \notag \\
	\leq 	&C(A) 	\int_{\Reals} \int_\Reals \int  	\frac{\lambda |V(z,v)|_v |V(p,v)|_v}{(1+|v|)(1+\alpha(\eps z,v))(1+ \alpha(\eps p,v))}  \frac{R \eps |\theta-\omega|} {(1+ \eps|\theta-\omega|) (1+|\theta-\omega|)^2} \label{I1estim} \\
	+ 		&C(A)  \int_{\Reals} \int_\Reals \int  	\frac{\lambda |V(z)|_v |V(p)|_v}{(1+|v|)(1+\alpha(\eps p,v))(1+ \alpha(\eps z,v))}  Y_{R,\eps,\delta}(\theta-\omega) \notag.
	\end{align}
	Observe the following straightforward integral estimates hold:
	\begin{equation} \label{loggain}
	\begin{aligned}
	\int_{\Reals} \frac{R\eps |\tau|}{(1+\eps|\tau|)(1+|\tau|)^2}\ud{\tau} &\leq C R \eps^\frac12, \quad
	\int_{\Reals} Y_{R,\eps,\delta}(\tau)\ud{\tau} \leq C (\delta + \eps^\frac12 R).
	\end{aligned}
	\end{equation}
	We apply Young's inequality to \eqref{I1estim} and use \eqref{loggain} to obtain a total bound of:
	\begin{align*}
	|I_1|\leq  \int_\Reals \int  	\frac{C(A)\lambda(v) |V(z,v)|_v^2}{(1+|v|)(1+ \alpha(z,v))^2} \ud{\omega} \left(\int_{\Reals}  \frac{R\eps |\tau|}{(1+\eps|\tau|)(1+|\tau|)^2} + Y_{R,\eps,\delta}(\tau)\right) \ud{v} 
	\leq \frac{c}{6} D^{\alpha}_{\eps,A}(u),
	\end{align*}
	for $0<\delta < \delta_0(n,A)$, $0<\eps \leq \eps_0(\delta,R,A,c,n)$ small and $D^{\alpha}_{\eps,A}(u)$
	as defined in \eqref{dissipexplic}. 
	The term $I_2$ (cf. \eqref{I1I2def}) can be controlled similarly. We split $I_2$ further into:
	\begin{align*}
	|I_2| 	\leq			&\left|\int_{\Reals} \int_\Reals \int \langle \lambda V(z,v) (z,v) N_1(\eps,z,p,v) V(p,v)
	\rangle  		\La(f)(i(\theta-\omega),v') \eta	\ud{v'}	\ud{v}\ud{\theta} 									\ud{\omega}	\right| 		 \\
	+		&\left|\int_{\Reals} \int_\Reals \int \langle \lambda V (z,v) (z,v), N_2(\eps,z,p,v) V(p,v) 		\rangle \La(f)(i(\theta-\omega),v') \eta	\ud{v'} \ud{v}\ud{\theta} 									\ud{\omega}	\right| \\
	=& I_{2,1} + I_{2,2}. 
	\end{align*}
	The integral $I_{2,2}$ can be bounded using \eqref{matrixsymm2} and \eqref{loggain} (adapting $0<\delta_0,\eps_0$ if needed):
	\begin{align*}
	|I_{2,2}|	\leq 	&\int_{\Reals} \int_\Reals \int \int \lambda 	\frac{|V(z)||V(p)|}{|v-v'|^3} \frac{\eps^2 |p| |z|+\eps^2|p| |z| (1+\eps|						\theta-\omega|)}{(1+\alpha(\eps z,v-v'))^2(1+ \alpha(\eps p,v-v'))^2}|\La(f)| \eta \ud{v} \ud{v'} \ud{\theta} \ud{\omega}\\
	\leq 	&C(A)\big(\delta+R \eps^\frac12 \big) \int_{\Reals}  \int \lambda 	\frac{|V(z,v)|^2}{1+|v|} \frac{\beta(\eps z, |v|)+\alpha^2(\eps z, |v|)}{(1+\alpha(\eps z,v-v'))^4}  \ud{v} \ud{\omega}  \leq \frac{c}{4} D^{\alpha}_{\eps,A}(u),
	\end{align*}
	where we use that $C_2 \leq C C_1$. It remains to control $I_{2,1}$, which we estimate
	by means of \eqref{N1est}. We obtain:
	\begin{align*}
	|I_{2,1}| 	\leq 		&C(A)\big(\delta + R \eps^\frac12\big) \int_{\Reals} \int_\Reals \int  	\frac{\lambda(v) |V(z,v)|^2}{(1+|v|^3) (1+\alpha(\eps z,v))^2}  \ud{\omega} \ud{v} \leq \frac{c}{4} D^{\alpha}_{\eps,A}(u) .
	\end{align*} 
	Therefore $|J_1|\leq \frac{c}2 D^{\alpha}_{\eps,A}(u)$.
	The remaining terms can be estimated by:
	\begin{align} \label{eq:nonlinlowerord}
	|J_2| + |I_3| + \sum_{\beta<\alpha} \binom{\alpha}{\beta}|Q^{\alpha,\beta}_{\eps,A}[f](u)| \leq \frac{c}{2} D^{\alpha}_{\eps,A}(u) +  \|u\|^2_{V^n_{A,\lambda}}.
	\end{align}
	The estimate for $Q^{\alpha,\beta}_{\eps,A}$, $\beta<\alpha$ can be seen as follows: Let $V,W \in \Complex^3$ be arbitrary. By the definition \eqref{Lambdadef} of $\Lambda$ and the estimate
	for $\La(f)$ in Lemma \ref{Ystable} we can bound $\langle V,\Lambda^\beta W\rangle$ by:
	\begin{align*}
	|\langle V, D^\beta \Lambda[f](z,\tau,v)W\rangle|			=  &\left|\int_{\Reals^3} \int_{0}^\infty  \langle V,D^\beta \left((M_1+M_2)(z,v-v')\eta\right)W \rangle e^{-i\tau s} f(s,v') \ud{s} \ud{v'}\right| \\
	\leq &C(A)\int_{\Reals^3}  \left| \langle V,D^\beta \left((M_1+M_2)(z,v-v')\eta\right)W\rangle \right|   Y_{R,\eps,\delta}(\tau) e^{-\frac12 |v'|} \ud{v'}. 	
	\end{align*}	
	We use Lemma \ref{derlemma} to estimate the velocity integral by:
	\begin{align}\label{derlemma2}
	|\langle V, D^\beta \Lambda[f](z,\tau,v)W\rangle|	&\leq  \frac{C(A)(1+\alpha(\eps z,v))}{(1+|v|)^{|\beta|}} C_1(\eps z,v) |V||W| Y_{R,\eps,\delta}(\tau). 
	\end{align}
	Now we can argue as in the proof of Lemma \ref{coercivitylemma}. The term $J_2$ is estimated as the corresponding term in the proof of Lemma \ref{coercivitylemma},
	using \eqref{derlemma2}.
	For estimating $I_3$ (given by \eqref{est4}), some care is needed. We rewrite $I_3$, integrating by parts (we use the shorthand $W(z,v)=\La( D^\alpha u)(z,v)$):
	\begin{align*} 
	I_3 =	&	-\int_{\Reals} \int_\Reals \int \langle W(z,v) \nabla^2 (\lambda), \Lambda(\eps z,\omega-\theta,v) W(p,v)  \rangle \ud{v}\ud{\theta} 	\ud{\omega}\\
	&-\int_{\Reals} \int_\Reals \int \langle  W(z,v)  \nabla (\lambda), \nabla \cdot\Lambda(\eps z,\omega-\theta,v) W(p,v)  \rangle \ud{v}\ud{\theta} 							\ud{\omega} \\				
	&	-\int_{\Reals} \int_\Reals \int \langle \nabla W(z,v) \otimes \nabla (\lambda), \Lambda(\eps z,\omega-\theta,v) W(p,v)  \rangle \ud{v}\ud{\theta} 	\ud{\omega} . 
	\end{align*}
	Here we use the notation $\langle A, B\rangle = \sum_{i,j} \ol{A}_{i,j} B_{i,j}$ for matrices $A,B$.
	The first two lines are bounded by $\frac12 \|u\|^2_{V^n_{A,\lambda}}$ using \eqref{derlemma2} and the Plancherel Lemma \ref{Plancherel}. The third
	line can be estimated like 	the corresponding $I_3$ in Lemma~\ref{coercivitylemma}.
	The lower order terms $\beta<\alpha$ are estimated
	in the same way using Lemma \ref{anisotrope},
	so we indeed obtain \eqref{eq:nonlinlowerord}.
	Combining all the estimates, we obtain the upper estimate $|Q^\alpha_{\eps,A}[f](u)| \leq c D^{\alpha}_{\eps,A}(u)+\|u\|^2_{V^n_{\eps,A}}$ as claimed.
\end{proof}
We obtain the main result of this subsection, Theorem \ref{thmfreez}, as a Corollary.

\begin{proofof}[Proof of Theorem \ref{thmfreez}]
	We have proved the existence of solutions $u$ to \eqref{contthmeq} in Lemma \ref{lemexistence}.
	We need to show continuity of the mapping $\Psi_{\delta_1}$ and the a priori estimate \eqref{PsiEst}.
	First we use Lemma \ref{lemReduc} to bound the norm of the solution by:
	\begin{align} \label{Redapplied}
	A \|u\|^2_{V^n_{A,\lambda}} 
	\leq  C \|u_0\|^2_{H^n_{\lambda}} -2 \sum_{|\alpha|\leq n} Q^\alpha_{\eps,A}[u_0](u)+ Q^\alpha_{\eps,A}[f](u) .
	\end{align}
	Applying Lemma \ref{coercivitylemma} to $Q^\alpha_{\eps,A}[u_0](u)$ and
	Lemma \ref{symmkernel} to $Q^\alpha_{\eps,A}[f](u)$  we find that for $A>0$ and $\delta>0$ sufficiently small, $R>0$ and $\eps>0$ small enough we have:
	\begin{align*}
	& Q^\alpha_{\eps,A}[u_0](u)+ Q^\alpha_{\eps,A}[f](u)	\geq \frac{c}{2} D^\alpha_{\eps,A} - C \|u\|_{V^n_{A,\lambda}}.	
	\end{align*}
	Plugging this back into \eqref{Redapplied} we find $A,\delta>0$ such that for all $R>0$ and $\eps>0$
	small we have, independently of $0<\gamma\leq 1$:
	\begin{align} \label{apriouV}
	\|u\|_{V^n_{A,\lambda}} \leq \|u_0\|_{H^n_{\lambda}}.	
	\end{align}
	Now define $U:= \int_0^t A^{0,0}_\gamma[u_0+f](u) $ ($\A$ as in Notation \ref{Adef}). Then by equation \eqref{contthmeq} we have $(u - u_0) = \nabla \cdot U$ .
	Using Lemma \ref{Mlemma} we write:
	\begin{equation} 
	\begin{aligned} \label{contKP}
	\La(\partial_t U)(z,v) =	&\int (M_1+M_2)(\eps z,v-v')  \eta \La \left((u_0(v')+f(\cdot,v')) {}^\gamma \nabla u(\cdot,v)\right)(z) \ud{v'} \\
	-	&\int {}^\gamma \nabla \cdot (M_1+M_2)(\eps z,v-v')  \eta \La \left((u_0(v')+f(\cdot,v'))  u(\cdot,v)\right)(z) \ud{v'}. 
	\end{aligned}
	\end{equation}
	Now $M_1$, $M_2$ as well as their derivatives are bounded. Further Lemma \ref{Ystable} and \eqref{loggain} imply:
	\begin{align} \label{thmfreezf}
	\|\La(f)(z,v)\|_{L^1_{\Re(z)=0}}\leq C(A) (\delta+R\eps^\frac12) e^{-\frac12 |v|}.
	\end{align}
	Hence for $\delta>0$ and $\eps(A,R)>0$ sufficiently small, combining \eqref{contKP}, \eqref{thmfreezf}, and the Plancherel
	Lemma \ref{Plancherel} gives the desired estimate for $U$ in \eqref{PsiEst}.
	Plugging this back into \eqref{contthmeq} gives \eqref{PsiEst}:
	\begin{align*}
	\|((u-u_0) \kappa_{\delta_1},U\kappa_{\delta_1}) \|_{X^n_{A,\lambda}}+\|\partial_t ((u-u_0) \kappa_{\delta_1},U\kappa_{\delta_1}) \|_{X^{n-2}_{A,\lambda}}\leq C.
	\end{align*}	
	It remains to show continuity of the operator $\Psi_{\delta_1}$ for positive $\gamma,\eps$. Let 
	$(f_i,g_i) \in \Omega^n_{A,R,\delta,\eps}$, $i=1,2$ and $u_1,u_2 $ the corresponding solutions
	to \eqref{contthmeq}. For shortness write 
	\begin{align*}
	K_i= \frac1\eps K[u_0+f_i(s)]\left(\frac{t-s}\eps,v\right), \quad P_i=\frac1\eps P[u_0+f_i(s)]\left(\frac{t-s}\eps,v\right).	
	\end{align*} Then the difference $u_1-u_2$ satisfies $(u_1(0)-u_2(0))=0$ and:
	\begin{align} \label{Gronwall} 
	\partial_t(u_1-u_2) = 	& {}^\gamma \nabla \cdot \left(\int_0^t K_1  {}^\gamma \nabla u_1(s,v) - P_1 u_1 - K_2 {}^\gamma \nabla u_2(s,v) + P_2 u_2  \ud{s}\right).
	\end{align}
	For $m\in \Naturals$ arbitrary, $\|K[f]\|_{L^2([0,1];C^m(\Reals^+;\Reals^3))}+ \|P[f]\|_{L^2([0,1];C^m(\Reals^+;\Reals^3))} \leq C \|(f,g)\|_{X^n_{A,\lambda}} $
	are continuous. Recalling that ${}^\gamma \nabla$ are mollifying operators, the continuity of $\Psi_{\delta_1}$ now follows from \eqref{Gronwall} by Gronwall's Lemma.  
\end{proofof}

\subsection{Invariance of the set $\Omega$ under the mapping $\Psi$} \label{Sec:Boundary}

\subsubsection{Recovering the quadratic decay in Laplace variables}
In the last subsection we have shown that for $(f_\eps,g_\eps) \in \Omega^n_{A,R,\delta,\eps}$ as defined in \eqref{omegadef}, the equation 
\begin{equation} \label{fpeq}
\begin{aligned}
\partial_t u_\eps 	= 	&\frac1\eps {}^\gamma \nabla \cdot \left(\int_0^t K[u_0+f_\eps(s)]\left(\frac{t-s}\eps,v\right)  {}^\gamma \nabla u_\eps(s,v) \ud{s}\right) \\
-  	&\frac1\eps {}^\gamma \nabla \cdot \left(\int_0^t P[u_0+f_\eps(s)]\left(\frac{t-s}\eps,v\right) u_\eps(s,v)  \ud{s}\right)  \\ 
u_\eps(0,\cdot)	&= u_0(\cdot),
\end{aligned}
\end{equation}
has solutions in $X^n_{A,\lambda}$. The goal of this section is to show that the associated solution operator $\Psi_{\delta_1}$ 
defined in \eqref{contthmeq} leaves the set $\Omega^n_{A,R,\delta,\eps}$ (cf. \eqref{omegadef}) invariant. More precisely, we will prove the following theorem.
\begin{theorem} \label{thmbdry} 
	Let $n\geq 6$ and assume $v_0\in H^n_{\lambda}$ satisfies the bounds:
	\begin{align*}
	0 \leq v_0(v) \leq C e^{-\frac12 |v|}.
	\end{align*} 	
	Let $A,\delta>0$ be as in Theorem \ref{thmfreez} and $\Psi_{\delta_1}$ the solution operator to \eqref{fpeq}:
	\begin{align*}
	\Psi_{\delta_1}: \Omega^n_{A,R,\delta,\eps} &\longrightarrow X^n_{A,\lambda} \\
	(f_\eps,g_\eps)		&\mapsto \left((u_\eps-u_0)k_{\delta_1}, \A^{0,0}_\gamma[u_0+f](u_\eps)k_{\delta_1}	\right), \text{$u_\eps$ solves \eqref{fpeq} with $u_0=m+\delta_2 v_0$.}
	\end{align*}
	There exist $\delta_1,\eps_0,R>0$  such that for
	$\delta_2,\eps \in (0,\eps_0]$ and for all $\gamma \in (0,1]$, the set $\Omega^n_{A,R,\delta,\eps}$ is invariant under the mapping $\Psi_{\delta_1}$. 
\end{theorem}
As a first step, we will prove estimate \eqref{gaprio2}. Differentiating equation \eqref{fpeq} yields, where $\A_\gamma^{\alpha,\beta}$ is defined in \eqref{curlyA}:
\begin{align}
\partial_t D^\alpha u_\eps &= \sum_{\beta\leq \alpha} \binom{\alpha}{\beta} \nabla \cdot \left(\A_\gamma^{\alpha,\beta}[u_0+f_\eps](u_\eps) \right) \label{fpeqlocal}.
\end{align} 
Therefore in order to characterize the properties of $D^\alpha u_\eps$ in Laplace variables, we first need to understand 
the right-hand side of the above equation in this framework.

\begin{lemma} \label{Alemma}
	Let $n\geq 0$ and $(f,g) \in \Omega^n_{A,R,\delta,\eps}$. Further let $u_0 \in C(\Reals^3)$ 
	satisfy 
	\begin{align*}
	0 \leq u_0(v) \leq C e^{-\frac12 |v|}.
	\end{align*}
	Let  $a= \frac{A}{2}\geq \frac12$, $\gamma \in (0,1]$ and $\beta\leq \alpha$ be multi-indexes with $|\alpha|=m< n$.
	Then for almost every $z\in \Complex$ with $\Re(z)=a$ we can estimate ($\A_\gamma^{\alpha,\beta}$ and $|\cdot|_{F^m}$ as in Notation \ref{Adef}):
	\begin{align*}
	|\La(\A_\gamma^{\alpha,\beta}[u_0](u))(z,v)| 		& \leq 	 \frac{C(A) |u|_{F^{m+1}}}{|1+\eps z|},	\quad	
	|\La(\A_\gamma^{\alpha,\beta}[f](u))(z,v)| 	 \leq C(A)  \frac{Y_{\eps,\delta} *_a |u|_{F^{m+1}}}{|1+\eps z|} . 
	\end{align*}
	Here the convolution $*_a$ is to be understood as ($z=a+i\omega$): 
	\begin{align}\label{aconvol}
	(f*_a g)(a+i\omega)= \int_{\Reals} f(i\theta) g(a+i(\omega-\theta))\ud{\theta}.
	\end{align}
\end{lemma}
\begin{proof}
	Is a direct consequence of elementary properties of the Laplace transform, Lemma
	\ref{Ystable} and the defining formula \eqref{curlyA} of $\A_\gamma^{\alpha,\beta}$.
\end{proof}
\begin{lemma} \label{cutofflaplace}
	Let $u\in V^n_{A,\lambda}$ for $n\geq 2$. 
	For $a\in (0,1]$ and $\delta_1\in  (0,1]$ we have:
	\begin{align}
	\|\La(u \kappa_{\delta_1})(\cdot,v)\|_{L^\infty_{\Re(z)=a}} &\leq C(a,\delta_1) \|\La(u)(\cdot,v)\|_{L^{2}_{\Re(z)=a}} \label{cutofflaplace1}\\
	\|\La(u \kappa_{\delta_1})(\cdot,v)\|_{L^2_{\Re(z)=a}} 		&\leq C(a,\delta_1) \|\La(u)(\cdot,v)\|_{L^{2}_{\Re(z)=a}}\label{cutofflaplace2}. 
	\end{align}
\end{lemma}
\begin{proof}
	We start by proving \eqref{cutofflaplace1}.
	Consider the two-sided Laplace transform $\tilde{\La}$:
	\begin{align*}
	\tilde{\La}(f)(z) &= \int_{-\infty}^\infty e^{-z t} f(t) \ud{t}.
	\end{align*}
	Extending $u(t)=0$ for negative $t$, we find that for $\Re(z)=a\geq \frac12$:
	\begin{align*}
	\tilde{\La}(u \kappa_{\delta_1})	&=  \tilde{\La}(\kappa_{\delta_1}) *_a \tilde{\La}(u).		
	\end{align*} 
	Since $\tilde{\La}(\kappa_{\delta_1})$ is a Schwartz function, the claim follows from Young's inequality and 
	the assumption $n\geq 2$ (so both sides of \eqref{cutofflaplace1}, \eqref{cutofflaplace2} are continuous). The proof of \eqref{cutofflaplace2} follows similarly.
\end{proof}

Now that we can characterize the properties of the operators $\A_\gamma^{\alpha,\beta}$ in Laplace variables, we are able to prove
bounds for the Laplace transforms of the solution $u_\eps$.

\begin{lemma}  \label{lemLaest} 
	Let $n\geq2$ and $A=2a\geq \frac12$, $\delta>0$ be as in Theorem \ref{thmfreez}. For $R>0$, $\gamma \in (0,1]$, $(f_\eps,g_\eps) \in \Omega^n_{A,R,\delta,\eps}$ let $u_\eps \in V^n_{A,\lambda}$ be the solution to \eqref{fpeq}, and let $|\alpha|=m\leq n-2$. 
	Recall the family of cutoff functions $\kappa_{\delta}$ defined in \eqref{defkappa}.
	For $\delta_3, \eps \in (0,1]$,  we have:
	\begin{align}	
	|\La(\kappa_{\delta_3}D^\alpha (u_\eps -u_0))| 			&\leq  \frac{ C(A,\delta_3)}{|1+\eps z|} 	\frac{1}{|z|} (|u_\eps\kappa_{2\delta_3}|_{F^{m+2}} + Y_{\eps,\delta} *_a |u_\eps\kappa_{2\delta_3} |_{F^{m+2}})	\label{pw1}\\ 
	|\La\left( \kappa_{\delta_3}\A^{0,0}_\gamma[u_0+f_\eps](u_\eps)\right)| 			&\leq  \frac{ C(A,\delta_3)}{|1+\eps z|} 	\frac{1}{|z|} (|u_\eps \kappa_{2\delta_3}|_{F^{m+2}} + Y_{\eps,\delta} *_a |u_\eps \kappa_{2\delta_3}|_{F^{m+2}}  ) \label{pw3},
	\end{align}
	a.e. on the line $\Re(z)=a$. Again we use the shorthand $*_a$ as introduced in \eqref{aconvol}.
\end{lemma}
\begin{proof}
	Integrating the equation \eqref{fpeq} we find:
	\begin{align*}
	(u_\eps -u_0)(T) = \int_0^T &\frac1\eps \nabla \cdot \left(\int_0^t K[u_0+f_\eps(s)]\left(\frac{t-s}\eps,v\right) \nabla u_\eps(s,v) \ud{s}\right) \\
	-  	&\frac1\eps \nabla \cdot \left(\int_0^t P[u_0+f_\eps(s)]\left(\frac{t-s}\eps,v\right) u_\eps(s,v)  \ud{s}\right) \ud{t}. 
	\end{align*}
	Since $\kappa_{2\delta_3}=1$ on the support of $\kappa_{\delta_3}$, the Volterra structure of the equation allows to rewrite:
	\begin{align}
	\kappa_{\delta_3} (u_\eps -u) 	= \kappa_{\delta_3} \int_0^T &\frac1\eps \nabla \cdot \left(\int_0^t K[u_0+f_\eps(s)]\left(\frac{t-s}\eps,v\right) \nabla (\kappa_{2\delta_3 }u_\eps)(s,v) \ud{s}\right) \label{cutofffpeq} \\
	-  	&\frac1\eps \nabla \cdot \left(\int_0^t P[u_0+f_\eps(s)]\left(\frac{t-s}\eps,v\right) (\kappa_{2\delta_3 }u_\eps)(s,v)  \ud{s}\right) \ud{t} \notag.  
	\end{align}
	Hence in Laplace variables we have:
	\begin{align*}
	z \La(D^\alpha (u_\eps -u_0) \kappa_{\delta_3})	
	&= \La\left(\kappa_{\delta_3} \sum_{\beta\leq \alpha} \binom{\alpha}{\beta} \nabla \cdot \left(\A_\gamma^{\alpha,\beta}[u_0+f_\eps](u_\eps \kappa_{2\delta_3})\right) \right)  .  
	\end{align*}
	Estimate \eqref{pw1} now follows from Lemma \ref{Alemma} and Lemma \ref{cutofflaplace}. 
	Estimate \eqref{pw3} is proved in the same way.
\end{proof}

\begin{lemma} [$L^\infty$ estimate in Laplace variables] \label{lemlinfty}
	Let $n\geq 2$ and $A=2a\geq \frac12$, $\delta>0$ be as in Theorem \ref{thmfreez}. For $R>0$, $\gamma, \delta_2 \in (0,1]$, $(f_\eps,g_\eps) \in \Omega^n_{A,R,\delta,\eps}$ let $u_\eps \in V^n_{A,\lambda}$ be the solution to \eqref{fpeq} with $u_0=m(v)+\delta_2 v_0(v)$, where
	$v_0\in H^n_\lambda$ satisfies:
	\begin{align*}
	0\leq v_0(v) \leq C e^{-\frac12|v|}.
	\end{align*}	
	Then for $m \in \Naturals$, $m\leq n-2$, $\eps>0$ small enough and $\delta_1 \in (0,1]$, there holds:	
	\begin{align*}
	\|\La(\nabla^m u_\eps \kappa_{\delta_1})\|_{L^\infty_{\Re(z)=a}} \leq C(A,\delta_1) e^{-\frac12 |v|} .
	\end{align*}
\end{lemma}
\begin{proof}
	We solve equation \eqref{fpeq} with $(f_\eps,g_\eps) \in \Omega^n_{A,R,\delta,\eps}$.
	Theorem \ref{symmkernel} shows there are $A,\delta,C(A)>0$ such that for all $R>0$ a solution $u_\eps$ to 		\eqref{fpeq} satisfies:
	\begin{align*}
	\|u_\eps\|_{V^n_{A,\lambda}} \leq C(A),
	\end{align*}
	provided $\eps>0$ is small enough.
	By Plancherel Lemma \ref{Plancherel} this implies in particular
	\begin{align*}
	\|\La(D^\alpha u_\eps)\|_{L^2_v L^2_{\Re(z)=a}} \leq C(A) \quad \text{for $|\alpha|\leq n$}.
	\end{align*} 
	With Sobolev inequality we can infer the existence of a constant $C(A)>0$ such that for every multi-index 		$\alpha$ with  $|\alpha|\leq n-2$ we have:
	\begin{align*}
	\|\La(D^\alpha u_\eps(\cdot,v))\|_{L^2_{\Re(z)=a}} \leq C(A) e^{-\frac12 |v|}.	
	\end{align*} 
	Now with Lemma \ref{cutofflaplace} we can estimate:
	\begin{align*}
	\|\La(\nabla^m u_\eps \kappa_{\delta_1})\|_{L^\infty_{\Re(z)=a}} \leq C(A,\delta_1)e^{-\frac12 |v|},	
	\end{align*}
	as claimed.
\end{proof}

We can plug the $L^\infty$ estimate for the Laplace transform back into \eqref{lemLaest} and bootstrap it to 
a pointwise estimate. 

\begin{lemma} [Linear decay in Laplace variables] \label{lemlin}
	Let $n\geq 4$ and $A=2a\geq \frac12$, $\delta>0$ be as in Theorem \ref{thmfreez}. For $R>0$, $\gamma, \delta_2 \in (0,1]$, $(f_\eps,g_\eps) \in \Omega^n_{A,R,\delta,\eps}$ let $u_\eps \in V^n_{A,\lambda}$ be the solution to \eqref{fpeq} with $u_0=m(v)+\delta_2 v_0(v)$, where
	$v_0\in H^n_\lambda$ satisfies:
	\begin{align*}
	0\leq v_0(v) \leq C e^{-\frac12|v|}.
	\end{align*}	
	Then for $m \in \Naturals$, $m\leq n-4$, $\eps>0$ small enough and $\delta_1 \in (0,1]$ there holds:
	\begin{align*}
	| \La(\nabla^{m}((u-u_\eps) \kappa_{\delta_1}) (z,v))| \leq \frac{C(A,\delta_1)e^{-\frac12 |v|}}{1+|z|} \\
	| \La(\nabla^{m}(\A^{0,0}_\gamma[u_0+f_\eps](u_\eps) \kappa_{\delta_1}) (z,v)| \leq \frac{C(A,\delta_1)e^{-\frac12 |v|}}{1+|z|}.	
	\end{align*}
\end{lemma}
\begin{proof}
	Follows by combining Lemma \ref{lemLaest} with Lemma \ref{lemlinfty}.
\end{proof}

Bootstrapping the estimate in Lemma \ref{lemLaest} gives an additional quadratic decay, which is
the content of the following Lemma.

\begin{lemma}[Quadratic decay of Laplace Transforms] \label{decaywithosmall}
	Let $n\geq 4$ and $A=2a\geq \frac12$, $\delta>0$ be as in Theorem \ref{thmfreez}. For $R>0$, $\gamma, \delta_2 \in (0,1]$, $(f_\eps,g_\eps) \in \Omega^n_{A,R,\delta,\eps}$ let $u_\eps \in V^n_{A,\lambda}$ be the solution to \eqref{fpeq} with $u_0=m(v)+\delta_2 v_0(v)$, where
	$v_0\in H^n_\lambda$ satisfies:
	\begin{align*}
	0\leq v_0(v) \leq C e^{-\frac12|v|}.
	\end{align*}	
	Then for $m \in \Naturals$, $m\leq n-4$, $\eps>0$ small enough and $\delta_1 \in (0,1]$ there holds:   
	\begin{align*}
	|\La(\nabla^m (u_\eps-u_0) \kappa_{\delta_1})(z,v)| &\leq \frac{ C(A,\delta_1) e^{-\frac12 |v|}}{|1+ \eps z| (1+|z|^2)} \\
	|\La(\nabla^{m}(\A^{0,0}_\gamma[u_0+f_\eps](u_\eps) \kappa_{\delta_1}) (z,v)| &\leq \frac{ C(A,\delta_1) e^{-\frac12 |v|}}{|1+ \eps z| (1+|z|^2)}.	
	\end{align*}
\end{lemma}
\begin{proof}
	Follows by iterating  Lemma \ref{lemLaest} further with the estimate Lemma \ref{lemlin}. For completeness we remark that the 
	linear decay  of $|u_\eps \kappa_{\delta_1}|_{F^{m+2}}$ is stable under convolution with
	$Y_{\eps,\delta}$. To see this we estimate the convolution explicitly ($z=a+i\omega$, $y=a+i\theta$, $a\geq \frac12$):
	\begin{align*}
	Y_{\eps,\delta} *_a |u_\eps \kappa_{\delta_1}|_{F^{m+2}}  
	\leq 	&\int_\Reals \frac{C(A,\delta_1)}{1+|z-y|} \left( \frac{\delta}{1+|z|^2} + \frac{R \eps |z|}{(1+\eps |z|)(1+ |z|^2)}\right) \ud{\omega} e^{-\frac12 |v|}  \\
	\leq &  \frac{C(A,\delta_1)}{1+|z|} + \int_\Reals \frac{C(A,\delta_1)}{1+|z-y|} \frac{R \eps |z|}{(1+\eps |z|)(1+ |z|^2)} \ud{\omega} e^{-\frac12 |v|}.
	\end{align*} 
	It remains to show that the last integral decays linearly with a prefactor independent of $R>0$. This can be seen by splitting the integral
	into the regions 
	\begin{align*}
	D_d(x)	&:= \{y: \Re(y)=a, \, |y|\geq 2|x| \, \text{ or } \, |y|\leq \frac12 |x|\} \\
	D_c(x)	&:= \{y: \Re(y)=a, \, \frac12 |x| \leq |y| \leq 2 |x| \},
	\end{align*}
	when the integral can be estimated as ($C(A,\delta_1)$ might change from line to line):
	\begin{align*}
	&\int_\Reals \frac{C(A,\delta_1)}{1+|z-y|} \frac{R \eps |y|}{(1+\eps |y|)(1+ |y|^2)} \ud{\theta} \\ 
	= 		&\int_{D_d(x)} \frac{C(A,\delta_1)}{1+|z-y|} \frac{R \eps |y|}{(1+\eps |y|)(1+ |y|^2)} \ud{\theta}	 + \int_{D_c(x)} \frac{C(A,\delta_1)}{1+|z-y|} \frac{R \eps |y|}{(1+\eps |y|)(1+ |y|^2)} \ud{\theta} \\
	\leq 	& \frac{C(A,\delta_1)}{1+|z|}\int_{D_d(x)}  \frac{R \eps }{(1+\eps |y|)(1+ |y|)} \ud{\theta}	 + \int_{D_c(x)} \frac{C(A,\delta_1)}{1+|z-y|} \frac{R \eps }{(1+\eps |y|)(1+ |y|)} \ud{\theta} \\
	\leq 	& \frac{C(A,\delta_1)}{1+|z|}	 + \int_{D_c(x)} \frac{C(A,\delta_1)}{1+|z-y|} \frac{R \eps }{(1+\eps |y|)(1+ |y|)} \ud{\theta}
	\end{align*}
	with $C(A,\delta_1)$ is independent of $R>0$, provided $\eps(R)>0$ is small enough. We can bound the second integral by:
	\begin{align*}
	\int_{D_c(x)} \frac{C(A,\delta_1)}{1+|z-y|} \frac{R \eps }{(1+\eps |y|)(1+ |y|)} \ud{\theta} 
	\leq 	&\frac{R \eps }{(1+\eps |z|)(1+ |z|)} \int_{D_c(x)} \frac{C(A,\delta_1)}{1+|z-y|}  \ud{\theta}	\\
	\leq 	&\frac{C(A,\delta_1) R \eps \log(1+|z|) }{(1+\eps |z|)(1+ |z|)} \leq 	\frac{1}{(1+ |z|)},
	\end{align*}	
	for $\eps>0$ small enough.
\end{proof}
As a corollary we obtain the uniform boundedness of the sequence $u_\eps$.
\begin{lemma}[Uniform boundedness] \label{unifbound}
	Let $n\geq 4$ and $A=2a\geq \frac12$, $\delta>0$ be as in Theorem~\ref{thmfreez}. For $R>0$, $\gamma, \delta_2 \in (0,1]$, $(f_\eps,g_\eps) \in \Omega^n_{A,R,\delta,\eps}$ let $u_\eps \in V^n_{A,\lambda}$ be the solution to \eqref{fpeq} with $u_0=m(v)+\delta_2 v_0(v)$, where
	$v_0\in H^n_\lambda$ satisfies:
	\begin{align*}
	0\leq v_0(v) \leq C e^{-\frac12|v|}.
	\end{align*}	
	Then for $m \in \Naturals$, $m\leq n-4$, $\eps>0$ small enough there holds:
	\begin{align} \label{ulinftyest}
	|\nabla^m (u_\eps-u_0)(t,v)| \leq C(A) e^{-\frac12 |v|}, \quad \text{for $0\leq t \leq 1$}.
	\end{align}
\end{lemma}

\subsubsection{Boundary Layer Estimate} \label{subsecboundarylayer}
To obtain smallness for the Laplace transforms, we separate the contributions of $M_1$ and $M_2$ to $u_\eps$.

\begin{lemma} [Decomposition] \label{lemmadec} 
	Let $(f_\eps,g_\eps) \in \Omega^n_{A,R,\delta,\eps}$ and $u_\eps \in V^n_{A,\lambda}$ a solution to \eqref{fpeq}. 		Then $u_\eps-u_0 = p_\eps + q_\eps$.
	Here $p_\eps = \nabla \cdot P_\eps$ is a divergence and $P_\eps$ is given by:
	\begin{equation}
	\begin{aligned}
	\partial_t P_\eps 	= 	& \left(\int_0^t \int \frac{\pi^2}{4}\frac{e^{-\frac{s|v'|}{\eps}}P_{v'}^\perp}{\eps} (u_0+f_\eps)(t-s,v-v') \eta(|v'|^2) \nabla u_\eps(t-s,v) \ud{v'} \ud{s} \right) 
	\label{pequation}\\
	-  	&  \left(\int_0^t \int \frac{\pi^2}{4}\frac{e^{-\frac{s|v'|}{\eps}}P_{v'}^\perp}{\eps} \nabla (u_0+f_\eps)(t-s,v-v') \eta(|v'|^2)  u_\eps(t-s,v) \ud{v'} \ud{s} \right)   \\ 
	P_\eps(0)&=0	.
	\end{aligned}
	\end{equation}	
	Similarly, $q_\eps= \nabla \cdot Q_\eps$, where $Q_\eps$ is given by:
	\begin{equation}
	\begin{aligned}
	z \La (Q_\eps) 	= 	& \left(\int M_2(\eps z,v') \eta(|v'|^2) \La\left((u_0+f_\eps)(s,v-v') \nabla u_\eps(s,v)\right) \ud{v'}\right) 
	\label{qequation}\\
	-  	&  \left(\int \nabla \cdot  M_2(\eps z,v')\eta(|v'|^2)  \La\left((u_0+ f_\eps)(s,v-v')) u_\eps(s,v)\right) \ud{v'}\right) . 	
	\end{aligned}
	\end{equation}
\end{lemma}
\begin{proof}
	We take the Laplace transform of equation \eqref{fpeq} and use Lemma \ref{Mlemma} to obtain:
	\begin{align*}
	&z \La(u_\eps)(z,v) - u_0(v) \\
	=	&\nabla \cdot \left( \int_{\Reals^3} (M_1+M_2)(\eps z,v') \eta(|v'|^2) \La\left((u_0+f_\eps)(s,v-v') \nabla u_\eps(s,v) \right)(z) \ud{v'} \right) \\
	-	&\nabla \cdot \left( \int_{\Reals^3} \nabla \cdot(M_1+M_2)(\eps z,v') \eta(|v'|^2) \La\left((u_0+f_\eps)(s,v-v')  u_\eps(s,v) \right)(z) \ud{v'} \right).
	\end{align*}
	Now introduce the functions $p_\eps, q_\eps$ given by the splitting:
	\begin{align}
	z \La(q_\eps)(z,v) 
	=	&\nabla \cdot \left( \int_{\Reals^3} M_2(\eps z,v') \eta(|v'|^2) \La\left((u_0+f_\eps)(s,v-v') \nabla u_\eps(s,v) \right)(z) \ud{v'} \right) \label{qpart} \\
	-	&\nabla \cdot \left( \int_{\Reals^3} \nabla \cdot M_2 (\eps z,v') \eta(|v'|^2) \La\left((u_0+f_\eps)(s,v-v')  u_\eps(s,v) \right)(z) \ud{v'} \right) \notag \\
	z \La(p_\eps)(z,v) 
	=	&\nabla \cdot \left( \int_{\Reals^3} M_1(\eps z,v') \eta \La\left((u_0+f_\eps)(s,v-v') \nabla u_\eps(s,v) \right)(z) \ud{v'} \right)  \label{ppart}\\
	-	&\nabla \cdot \left( \int_{\Reals^3} \nabla \cdot M_1(\eps z,v') \eta \La\left((u_0+f_\eps)(s,v-v')  u_\eps(s,v) \right)(z) \ud{v'} \right) \notag.
	\end{align}
	Therefore $q_\eps=\nabla \cdot Q_\eps$, with $Q_\eps$ as in \eqref{qequation}. To show $p_\eps=\nabla \cdot P_\eps$ we transform the equation for $p_\eps$ back to
	the variables $(t,v)$. To do so we remark that $M_1$ is the Laplace transform of:
	\begin{align*}
	\frac{\pi^2}{4}\La\left(\frac{e^{-\frac{t |v|}{\eps}}}{\eps  }\right)(z) P_{v}^\perp &= M_1(\eps z, v).
	\end{align*}
	Therefore  $p_\eps =\nabla \cdot Q_\eps$ and $u_\eps-u_0 = q_\eps + p_\eps$ as claimed.
\end{proof}

Splitting the function $u_\eps$ into $u_\eps=p_\eps + q_\eps$ allows to estimate the contributions of
$M_1$ and $M_2$ (as in \eqref{Mlemma}) separately. The function $q_\eps$ can be estimated in a straightforward fashion.

\begin{lemma} [Estimate for $q_\eps$] \label{qlemma}
	Let $n\geq 4$ and $A=2a\geq \frac12$, $\delta>0$ be as in Theorem \ref{thmfreez}. For $R>0$, $\gamma, \delta_2 \in (0,1]$, $(f_\eps,g_\eps) \in \Omega^n_{A,R,\delta,\eps}$ let $u_\eps \in V^n_{A,\lambda}$ be the solution to \eqref{fpeq} with $u_0=m(v)+\delta_2 v_0(v)$, where
	$v_0\in H^n_\lambda$ satisfies:
	\begin{align*}
	0\leq v_0(v) \leq C e^{-\frac12|v|}.
	\end{align*}	
	Let $\nabla \cdot Q_\eps= q_\eps \in V^n_{A,\lambda}$ be given by \eqref{qequation}. 
	Then for $m \in \Naturals$, $m\leq n-4$, $\eps>0$ small enough there holds: 
	\begin{align}
	|\La(\nabla^m q_\eps \kappa_{\delta_1})(z,v)| \leq \frac{C(A,\delta_1) \eps |z|}{(1+\eps |z|)^2(1+ |z|^2)}e^{-\frac12 |v|} \label{qest} \\
	|\La(\nabla^m Q_\eps \kappa_{\delta_1})(z,v)| \leq \frac{C(A,\delta_1) \eps |z|}{(1+\eps |z|)^2(1+ |z|^2)}e^{-\frac12 |v|} \label{Qest}. 
	\end{align}
	In particular, for $0\leq t \leq 1$, $m\leq n-4$ we have
	\begin{align}
	|\partial_t \nabla^ m q_\eps| &\leq C(A) e^{-\frac12 |v|} 			&|\partial_t \nabla^m Q_\eps| &\leq C(A) e^{-\frac12 |v|}. \label{dtqest}
	\end{align}
\end{lemma}

\begin{lemma}[$L^\infty$ estimate for time derivative] \label{dtest}
	Let $n\geq 4$ and $A=2a\geq \frac12$, $\delta>0$ be as in Theorem \ref{thmfreez}. For $R>0$, $\gamma, \delta_2 \in (0,1]$, $(f_\eps,g_\eps) \in \Omega^n_{A,R,\delta,\eps}$ let $u_\eps \in V^n_{A,\lambda}$ be the solution to \eqref{fpeq} with $u_0=m(v)+\delta_2 v_0(v)$, where
	$v_0\in H^n_\lambda$ satisfies:
	\begin{align*}
	0\leq v_0(v) \leq C e^{-\frac12|v|}.
	\end{align*}	
	Then for $m \in \Naturals$, $m\leq n-4$, $\eps>0$ small enough there holds: 
	\begin{align}
	|\partial_t \nabla^m u_\eps(t,v) | \leq C(A) e^{-\frac12 |v|} \quad \text{for $0\leq t \leq 1$}. \label{dtestform}
	\end{align}
\end{lemma}
\begin{proof}
	We use the decomposition $u_\eps = p_\eps + q_\eps$ introduced in Lemma \ref{lemmadec}. By the previous 
	Lemma \ref{qlemma} we know
	\begin{align*}
	|\partial_t \nabla^m q_\eps(t,v) | \leq C(A) e^{-\frac12 |v|} \quad \text{for $0\leq t \leq 1$}.
	\end{align*} 
	It remains to estimate $p_\eps$. The sequence $e^{-s/\eps}/\eps$ is bounded in $L^1$. Therefore the claim follows 
	by inserting the estimate \eqref{ulinftyest} into the definition \eqref{pequation}
	of $p_\eps = \nabla \cdot P_\eps$.
\end{proof}

\begin{notation}
	Let $b$ be the function given by:
	\begin{align*}
	b(t,r):= \frac{e^{-tr}}{r^2}  + \frac{t}{r} - \frac{1}{r^2}.
	\end{align*}
	For $u_0 \in H^n_{\lambda}$, define the boundary layer $B(t,v;u_0)= \nabla \cdot B_F(t,v;u_0)$ by:
	\begin{equation} \label{defBL}
	\begin{aligned}
	B_F(t,v;u_0) :=\int  \frac{\pi^2}{4}\frac{b(t,\frac{|v'|}{\eps})P_{v'}^\perp}{\eps}\eta 
	\left( u_0(v-v') \nabla u_0(v)  -  \nabla u_0(v-v')  \ u_0(v) \right) \ud{v'}.
	\end{aligned}
	\end{equation}
\end{notation} 

\begin{lemma}[Boundary Layer property] \label{blproperty}
	The function $B=\nabla \cdot B_F$, as defined in \eqref{defBL} satisfies:
	\begin{align*}
	\partial_{tt} B(t,v) &=   \nabla \cdot \left(\int \frac{\pi^2}{4} \frac{e^{-\frac{t|v'|}{\eps}}P_{v'}^\perp}{\eps} \eta(|v'|^2) \left( u_0(v-v')  \nabla 									u_0(v)
	-    \nabla u_0(v-v')  u_0(v)\right) \ud{v'} \right), \\
	B(0,v)	&=0  \quad  \partial_t B(0,v)=0.
	\end{align*}
\end{lemma}
\begin{proof}
	Differentiating $b$ gives:
	\begin{align*}
	\partial_t b(t,r) = \frac{1-e^{-rt}}{r}, \quad \partial_{tt} b(t,r) &= e^{-rt}.
	\end{align*}
	Therefore the second time derivative of $B$ is:
	\begin{align*}
	\partial_{tt} B(t,v) = 	& \nabla \cdot \left(\int  \frac{\pi^2}{4} \frac{e^{-\frac{|v'|}{\eps}t}P_{v'}^\perp}{\eps}\eta \left( u_0(v-v')  \nabla u_0(v) 
	-   \nabla u_0(v-v')   u_0(v)\right) \ud{v'}  									\right).	
	\end{align*}
	The initial data $B(0,v)=0$, $\partial_t B(0,v)=0$ follow by simply putting $t=0$. 
\end{proof}

\begin{lemma}[Remainder estimate]
	Let  $n\geq 4$ and $p_\eps$ solve \eqref{pequation} and $\|u_\eps\|_{V^n_{A,\lambda}} \leq C$. 
	There exists a $C_0>0$ such that for all
	$m\leq n-2$ there exists $\eps$ small enough such that: 
	\begin{align*}
	|\partial_{tt}(p_\eps-B)(t,v) | \leq C_0 e^{-\frac12 |v|},\quad \text{for $t\in[0,1]$}.
	\end{align*}
\end{lemma}
\begin{proof}
	Take the time derivative of \eqref{pequation}. We can split using Lemma \ref{blproperty}:
	\begin{align*}
	\partial_{tt} p_\eps 	= 	& \nabla \cdot \left(\int_0^t \int\frac{\pi^2}{4}  \frac{e^{-\frac{s|v'|}{\eps}}P_{v'}^\perp}{\eps } \partial_t((u_0+f_\eps)(t-s,v-v') \eta \nabla 									u_\eps(t-s,v)) \ud{v'} \ud{s} \right) \\
	-  	& \nabla \cdot \left(\int_0^t \int \frac{\pi^2}{4} \frac{e^{-\frac{s|v'|}{\eps}}P_{v'}^\perp}{\eps } \partial_t(\nabla (u_0+f_\eps)(t-s,v-v')) \eta  u_\eps(t-s,v) \ud{v'} \ud{s} \right) \\
	+	& \nabla \cdot \left(\int \frac{\pi^2}{4} \frac{e^{-\frac{t|v'|}{\eps}}P_{v'}^\perp}{\eps} \eta(|v'|^2) \left( u_0(v-v') \nabla 						u_0(v) - \nabla u_0(v-v')    u_0(v)\right) \ud{v'} \right) \\
	=  &R_1 + R_2 + \partial_{tt} B .
	\end{align*}
	Since $|\partial_t f_\eps|\leq C e^{-\frac12 |v|}$ by assumption, we obtain:
	\begin{align*}
	|\partial_{tt} (p_\eps - B)(t,v)|=|R_1(t,v) +R_2(t,v)| \leq C_0 e^{-\frac12 |v|},\quad \text{for $t\in[0,1]$},
	\end{align*}
	as claimed.
\end{proof}
\begin{lemma}[Smallness of $\La(p_\eps-B)$] \label{LaRem}
	Let $p_\eps$ solve \eqref{pequation} and $\|u_\eps\|_{V^n_{A,\lambda}} \leq C$ for some $A=2a>0$. 
	We have $p_\eps - B = \nabla \cdot (P_{\eps}-B_F)$, and there is a $C_0>0$ such that for all
	$m\leq n-2$, $\delta_1>0$ and $\eps>0$ small enough: 
	\begin{align*}
	|\La((p_\eps - B) \kappa_{\delta_1})(z,v)|+|\La((P_\eps-B_F) \kappa_{\delta_1})(z,v)| &\leq \frac{\delta_1 C_0 e^{-\frac12 |v|}}{1+|z|^2}.
	\end{align*}	
\end{lemma}
\begin{proof}
	By definition of $B$ the difference $p_\eps -B$ vanishes initially, as well as the time derivative:
	\begin{align*}
	(p_\eps - B)(0,v) = \partial_t (p_\eps - B)(0,v) =0. 
	\end{align*}
	Combined with the lemma above this shows:
	\begin{align*}
	|\partial_{tt} ((p_\eps - B) \kappa_{\delta_1})| \leq C_0 e^{-\frac12 |v|} (1 + \frac{t}{\delta_1} + \frac{t^2}{\delta_1^2}) \kappa_{\delta_1},\quad \text{for $t\in[0,1]$}. 
	\end{align*}
	After integrating by parts twice this allows to bound the Laplace transform by:
	\begin{align*}
	|\La((p_\eps - B) \kappa_{\delta_1})(z,v)| 	&\leq \frac{C_0 e^{-\frac12 |v|}}{|z|^2} \int_0^\infty C e^{-\frac12 |v|} (1 + \frac{t}{\delta_1} + \frac{t^2}{\delta_1^2}) \kappa_{\delta_1} \ud{t}	\\
	&\leq \frac{C_0 e^{-\frac12 |v|}}{|z|^2} \delta_1.
	\end{align*}
	The estimate for $P_\eps-B_F$ is proved similarly.
\end{proof}

\begin{lemma} [Stationarity of $m$] \label{mstation}
	Let  $\sigma^2,m_0>0$, $m(\sigma^2,M_0)(v)$ be the Maxwellian defined in \eqref{defmaxwellian}. Then 
	for all $t\geq 0$, $v\in \Reals^3$ we have:
	\begin{equation}
	B(t,v;m) = 0.
	\end{equation}
\end{lemma}
\begin{proof}
	The argument is identical to the one proving that $m$ is a stationary point of the Landau equation: First we observe that
	\begin{align*}
	\nabla m(v) &= - \frac{v}{\sigma^2} m(v).
	\end{align*}
	This however implies that:
	\begin{align*}
	P_{v'}^\perp \left(  m(v-v')\nabla m(v) - \nabla m(v-v')   m(v)\right) 
	=-	P_{v'}^\perp  \frac{v'}{\sigma^2}  m(v-v')m(v) = 0.  
	\end{align*}
	Inserting this into the definition of $B(t,v;m)$ in \eqref{defBL} gives the claim.
\end{proof}
We use the stationarity of the Maxwellian $m$ to obtain smallness of the boundary layer, provided
the evolution starts sufficiently close to $m$.
\begin{lemma}[Boundary layer estimate] \label{BLlemma}
	Let $u_0=m(v) + \delta_2 v_0$, for $v_0$ some fixed smooth function satisfying
	\begin{align*}
	0	\leq 	&v_0(v) \leq C e^{-\frac12 |v|}, \quad	|\nabla^i 	 v_0| \leq C e^{-\frac12 |v|} \quad \text{ for $i=0,1,2$}.
	\end{align*}
	Let $B$ be the associated Boundary Layer defined by \eqref{defBL}. Then the Laplace transforms of $B$ and $B_F$ satisfy: 
	\begin{align} \label{BLestimate}
	|\La(B \kappa_{\delta_1})(z,v)| + |\La(B_F \kappa_{\delta_1})(z,v)| \leq C(\delta_1) \frac{\delta_2 e^{-\frac12 |v|}}{1+|z|^2}. 
	\end{align}
\end{lemma}
\begin{proof}
	Using Lemma \ref{mstation} we can simplify $B$ to:
	\begin{align*}
	B(t,v) 	= 	& \nabla \cdot \left(\int  \frac{b(t,\frac{|v'|}{\eps})P_{v'}^\perp}{\eps } (m+ \delta_2 v_0)(v-v') \eta \nabla 															(m+ \delta_2 v_0)(v) \ud{v'}  \right) \\
	-  	& \nabla \cdot \left(\int  \frac{b(t,\frac{|v'|}{\eps})P_{v'}^\perp}{\eps} \nabla (m+ \delta_2 v_0)(v-v')  \eta  (m+ \delta_2 v_0)(v) \ud{v'}  									\right)\\
	=	& \nabla \cdot \left(\int  \frac{b(t,\frac{|v'|}{\eps})P_{v'}^\perp}{\eps} \eta 
	\left[\delta_2 v_0 (v')  \nabla m(v) +(\delta_2 v_0+ m)(v-v')\delta_2  \nabla v_0(v)  \right] \ud{v'}  \right) \\
	-  	& \nabla \cdot \left(\int  \frac{b(t,\frac{|v'|}{\eps})P_{v'}^\perp}{\eps} \eta 
	\left[\delta_2 \nabla v_0 (v')   m(v) + \nabla (\delta_2 v_0+ m)(v-v')\delta_2   v_0(v)  \right] \ud{v'} 									\right).	 
	\end{align*}
	The Laplace transform of $b$ can be computed explicitly:
	\begin{align*}
	\La(b(\cdot,r))(z) &= \frac{1}{r z^2}- \frac{1}{r(z+r)z}.
	\end{align*}
	Inserting this above we obtain the estimate:
	\begin{align*}
	|\La(B \kappa_{\delta_1})(z,v)| + |\La(B_F \kappa_{\delta_1})(z,v)| \leq C(\delta_1) \frac{\delta_2 e^{-\frac12 |v|}}{1+|z|^2}, 
	\end{align*}
	which is the claim of the Lemma.
\end{proof}

We are in the position to now prove Theorem \ref{thmbdry}.

\begin{proofof}[Proof of Theorem \ref{thmbdry}]
	Let $A,\delta>0$ as in Theorem \ref{thmfreez}. Then the theorem ensures that for $R>0$, $\delta_2 \in (0,1]$ arbitrary,
	and $(f,g) \in \Omega^n_{A,R,\delta_\eps}$ the solution $u_\eps$ to \eqref{fpeq} with $u_\eps-u_0= \nabla \cdot U_\eps$ can be bounded by:
	\begin{align} \label{startaprio}
	\|u_\eps \kappa_{\delta_1}\|_{V^n_{A,\lambda}} + \|U_\eps \kappa_{\delta_1}\|_{V^{n-1}_{A,\lambda}} \leq C .
	\end{align} 
	We use that $\psi_{\delta_1}(f,g)= (\kappa_{\delta_1}(u_\eps-u_0), \kappa_{\delta_1} U_\eps)$  and decompose $u_\eps$ into three pieces:
	\begin{equation}
	\begin{aligned} \label{termbytermlap}
	(u_{\eps}-u_0) \kappa_{\delta_1} 	&= 	 (p_\eps-B) \kappa_{\delta_1} + B \kappa_{\delta_1}  + q_\eps \kappa_{\delta_1}\\
	U_\eps 	\kappa_{\delta1}							&=  (P_\eps-B_F) \kappa_{\delta1} +  B_F \kappa_{\delta1} +  Q_\eps \kappa_{\delta_1}. 
	\end{aligned} 
	\end{equation}
	Using estimate \eqref{startaprio} and Lemmas \ref{qlemma}, \ref{LaRem}, \ref{BLlemma} we can find $\delta_1,\eps_0>0$ small enough and $R>0$ large enough, such that for 
	$\delta_2,\eps \in (0,\eps_0]$ the Laplace transforms of the summands in \eqref{termbytermlap} can be estimated by:
	\begin{align*}
	|\La (u_{\eps} \kappa_{\delta_1})| + |\La (U_{\eps} \kappa_{\delta_1})|	&\leq \frac{ \delta e^{-\frac12 |v|}}{1+|z|^2} + \frac{R \eps |z| e^{-\frac12 |v|}}{(1+\eps |z|)^2(1+|z|)^2}.
	\end{align*} 
	So we recover \eqref{gaprio1}, one of the defining estimates of $\Omega^n_{A,R,\delta,\eps}$ . The upper bound \eqref{gaprio2} is the content of Lemma \ref{decaywithosmall}. The remaining  estimate \eqref{dercond} is proved in Lemma \ref{dtest}. 
\end{proofof}
\section{Existence of solutions and Markovian Limit} \label{sec:Existence}

\subsection{Existence of a solution to the non-Markovian equation}
With the a priori estimates proved in the last section, we can now 
prove Theorem~\ref{mainthm1pf}.

\begin{proofof}[Proof of Theorem \ref{mainthm1pf}]
	Without loss of generality, let $m$ be the standard Gaussian, i.e. $\sigma=m_0=1$.
	First let $\gamma>0$. We invoke Theorems \ref{thmfreez} and \ref{thmbdry} to find $A,\delta,R, \delta_1>0$ and $\eps_0>0$ 		such that for all 
	$\eps,\delta_2 \in (0,\eps_0]$ the mapping $\Psi_{\delta_1}: \Omega^n_{A,R,			\delta,\eps} \rightarrow \Omega^n_{A,R,\delta,\eps}$ is continuous	with respect to the topologies of $X^n_{A,\tilde{\lambda}}$, $X^n_{A,\lambda}$, hence also
	as a map from $X^n_{A,\tilde{\lambda}}$ to itself. By Lemma \ref{closedconvex} we know that $\Omega^n_{A,R,\delta,\eps}$ is a closed, convex, bounded and 	nonempty subset of 
	$X^n_{A,\tilde{\lambda}}$. Therefore, existence of a fixed point of $\Psi_{\delta_1}$ follows from Schauder's theorem, provided 		we can show that the mapping is compact. To see this, we use that
	Theorem \ref{freethm} gives the estimate:
	\begin{align} \label{PsiEst2}
	\|\Psi_{\delta_1} (f,g)\|_{X^{n}_{A,\lambda}} + \|\partial_t \Psi_{\delta_1} (f,g)\|_{X^{n-2}_{A,				\lambda}} \leq C(A).
	\end{align}
	Since ${}^\gamma \nabla$ is smoothing, the defining equation \eqref{fpeq} of $\Psi_{\delta_1}$ implies:
	\begin{align*}
	\|\Psi_{\delta_1} (f,g)\|_{X^{n+1}_{A,\lambda}} + \|\partial_t \Psi_{\delta_1} (f,g)\|_{X^{n+1}_{A,				\lambda}} \leq C(A,\gamma).
	\end{align*}
	This implies compactness of the mapping $\Psi_{\delta_1}$ by the Rellich type Lemma \ref{Rellich}.
	Hence for $\gamma \in (0,1]$, we have proved the existence of solutions $u_{\eps,\gamma}$ to:
	\begin{equation} \label{mollifiedfp} 
	\begin{aligned}
	\partial_t u_{\eps,\gamma} 	= 	&\frac1\eps {}^\gamma \nabla \cdot \left(\int_0^t K[u_{\eps,					\gamma}]\left(\frac{t-s}\eps,v\right) {}^\gamma \nabla u_{\eps,\gamma}(s,v) \ud{s}\right) \\
	-  	&\frac1\eps {}^\gamma \nabla \cdot \left(\int_0^t P[u_{\eps,\gamma}]\left(\frac{t-s}\eps,v\right) u_{\eps,\gamma}(s,v)  \ud{s}\right) \\ 
	u_{\eps,\gamma}(0,\cdot)	&= u_0(\cdot),
	\end{aligned}  
	\end{equation} 
	for times $0\leq t\leq \delta_1$. It remains to pass $\gamma \rightarrow 0$ to obtain
	a solution of the non-mollified equation. The uniform estimate \eqref{PsiEst2} shows that
	for $\eps>0$ there is a sequence $\gamma_j \rightarrow 0$ such that $u_{\eps,\gamma_j}\rightarrow u_\eps$ in $V^{n-3}_{A,\tilde{\lambda}}$,
	$u_{\eps,\gamma_j}\rightharpoonup u_\eps$ in $V^{n}_{A,\lambda}$ and
	$\partial_t u_{\eps,\gamma_j}\rightharpoonup \partial_t u_\eps$ in $V^{n-2}_{A,\lambda}$. 
	Hence both sides of \eqref{mollifiedfp} converge weakly in $V^{n-2}_{A,\lambda}$, and it suffices
	to identify the limit of the right-hand side. Indeed, from the convergence in  $V^{n-3}_{A,\tilde{\lambda}}$
	we conclude that pointwise a.e. along a subsequence: 
	\begin{equation} \label{mollifiedfp2}
	\begin{aligned}
	&{}^{\gamma_j} \nabla \cdot \left(\int_0^t K[u_{\eps,{\gamma_j}}]\left(\frac{t-s}\eps,v\right) {}^{\gamma_j} \nabla u_{\eps,{\gamma_j}}(s,v) - P_{\gamma_j}[u_{\eps,{\gamma_j}}]\left(\frac{t-s}\eps,v\right) u_{\eps,{\gamma_j}}(s,v)  \ud{s}\right) \\
	\rightarrow & \nabla \cdot \left(\int_0^t 		K[u_{\eps}]\left(\frac{t-s}\eps,v\right)  \nabla u_{\eps}(s,v) -P[u_{\eps}]\left(\frac{t-s}\eps,v\right) u_{\eps}(s,v)  \ud{s}\right).
	\end{aligned}  
	\end{equation}
	Estimate \eqref{apriounif} follows from \eqref{PsiEst2}, and inserting the estimate back into equation \eqref{thmepseq} proves that $u_\eps \in C^1([0,\delta_1];H^{n-2}_{\lambda})$.	
\end{proofof}

\subsection{Non-Markovian to Markovian limit} \label{seclimits}
In this section we prove the transition from non-Markovian to Markovian
dynamics on the macroscopic timescale. As $\eps\rightarrow 0$, the solutions
$u_\eps$ to the non-Markovian equations \eqref{thmepseq} converge to
solutions of the Landau equation.

\begin{proofof}[Proof of Theorem \ref{mainthm2pf}]
	For the solutions $u_\eps$ of \eqref{thmepseq} constructed in Theorem \ref{mainthm1pf} we have the a priori bound: 
	\begin{align*} 
	\|((u_\eps-u_0)\kappa_{\delta_1},U_\eps \kappa_{\delta_1})\|_{X^n_{A,\lambda}} + \|\partial_t((u_\eps-u_0)\kappa_{\delta_1},U_\eps \kappa_{\delta_1})\|_{X^{n-2}_{A,\lambda}} \leq C(A).
	\end{align*}
	Using the compactness Lemma \ref{Rellich} and the fact that $V^n_{A,\lambda}$ is a separable Hilbert space, we can find $u \in V^n_{A,\lambda}$,
	s.t. along a sequence $\eps_j\rightarrow 0$ we have $u_{\eps_j} \rightarrow u$ in $V^{n-3}_{A,\tilde{\lambda}}$, $u_{\eps_j} \rightharpoonup u \in V^n_{A,\lambda}$ and $\partial_t u_{\eps_j} \rightharpoonup \partial_t u \in V^{n-2}_{A,\lambda}$. We need to show
	that $u$ solves the equation \eqref{uequation}. Since both sides
	of the equation are well-defined and have a well-defined Laplace transform,
	it is sufficient to show that $u$ solves the equation in Laplace variables.
	To this end, we take the Laplace transform of \eqref{thmepseq}:
	\begin{equation} \label{laplacelimit}
	\begin{aligned}		
	\La(\partial_t u_{\eps_j})(z,v) 	= 	&\nabla \cdot \left(\int_{\Reals^3} (M_1+M_2)({\eps_j} z,v') \La(u_{\eps_j}(s,v-v') \nabla u_{\eps_j}(s,v))(z) \eta \ud{v'}\right) \\
	-	&\nabla \cdot \left(\int_{\Reals^3} \nabla (M_1+M_2)({\eps_j} z,v') \La(u_{\eps_j}(s,v-v')  u_{\eps_j}(s,v))(z)\eta \ud{v'}\right).					\end{aligned}		
	\end{equation}
	The left-hand side converges pointwise to $\La(\partial_t u)=z\La( u)+u_0$, up to choosing a further subsequence.	
	The right-hand side of \eqref{laplacelimit} converges pointwise along a subsequence to:
	\begin{align*}
	&\nabla \cdot \left(\int_{\Reals^3} \frac{\pi^2}{4 |v'|} \frac{1}{1+\frac{z}{|v'|}} P_{v'}^\perp  \La(u(s,v-v') \nabla u(s,v))(z)\eta \ud{v'}\right) \\
	-	&\nabla \cdot \left(\int_{\Reals^3} \frac{\pi^2}{4 |v'|} \frac{1}{1+\frac{z}{|v|}} P_{v'}^\perp  \La(\nabla u(s,v-v')  u(s,v))(z) \eta \ud{v'}\right) \\
	=	 &\La\left(\nabla \cdot \left( \K[u] \nabla u\right) - \nabla \cdot \left( \Pe[u]  u\right) \right).	
	\end{align*} 
	Therefore $u \in V^n_{A,\lambda} \cap C^1([0,\delta_1];H^{n-4}_{\lambda})$ solves equation \eqref{uequation} as claimed.
\end{proofof} 

\textbf{Acknowledgment.}
We thank Alessia Nota, Mario Pulvirenti and Chiara Saffirio for interesting discussions
and useful suggestions on the topic.

The authors acknowledge support through the CRC 1060
\textit{The mathematics of emergent effects}
at the University of Bonn that is funded through the German Science
Foundation (DFG).

\newpage
 
\bibliography{Non-Markovian_Landau_VW17.bib}
\bibliographystyle{plain} 

%
%



\end{document}